\documentclass[pra,twocolumn,superscriptaddress]{revtex4-1}%
\usepackage{amsfonts}
\usepackage{amsmath}
\usepackage{amssymb}
\usepackage{graphicx}
\usepackage[ colorlinks = true, linkcolor = blue, urlcolor  = blue, citecolor
= red,              anchorcolor = green, ]{hyperref}
\usepackage[dvipsnames]{xcolor}%
\providecommand{\U}[1]{\protect\rule{.1in}{.1in}}

\newtheorem{lemma}{Lemma}

\newtheorem{proposition}{Proposition}
\newtheorem{remark}{Remark}

\newenvironment{proof}[1][Proof]{\noindent\textbf{#1.} }{\ \rule{0.5em}{0.5em}}
\allowdisplaybreaks

\begin{document}
\preprint{ }
\title[Resource theory of asymmetric distinguishability]{Resource theory of asymmetric distinguishability}
\author{Xin Wang}
\email{wangxinfelix@gmail.com}
\affiliation{Joint Center for Quantum Information and Computer Science, University of
Maryland, College Park, Maryland 20742, USA}
\affiliation{Institute for Quantum Computing, Baidu Research, Beijing 100193, China}

\author{Mark M. Wilde}
\email{mwilde@lsu.edu}
\affiliation{Hearne Institute for Theoretical Physics, Department of Physics and Astronomy,
Center for Computation and Technology, Louisiana State University, Baton
Rouge, Louisiana 70803, USA}

\keywords{}
\pacs{}

\begin{abstract}
This paper systematically develops the resource theory of asymmetric
distinguishability, as initiated roughly a decade ago [K.~Matsumoto,
arXiv:1010.1030 (2010)]. The key constituents of this resource theory are
quantum boxes, consisting of a pair of quantum states, which can be
manipulated for free by means of an arbitrary quantum channel. We introduce
bits of asymmetric distinguishability as the basic currency in this resource
theory, and we prove that it is a reversible resource theory in the asymptotic
limit, with the quantum relative entropy being the fundamental rate of
resource interconversion. The distillable distinguishability is the optimal
rate at which a quantum box consisting of independent and identically
distributed (i.i.d.)~states can be converted to bits of asymmetric
distinguishability, and the distinguishability cost is the optimal rate for
the reverse transformation. Both of these quantities are equal to the quantum
relative entropy. The exact one-shot distillable distinguishability is equal
to the min-relative entropy, and the exact one-shot distinguishability cost is
equal to the max-relative entropy. Generalizing these results, the approximate
one-shot distillable distinguishability is equal to the smooth min-relative
entropy, and the approximate one-shot distinguishability cost is equal to the
smooth max-relative entropy. As a notable application of the former results,
we prove that the optimal rate of asymptotic conversion from a pair of
i.i.d.~quantum states to another pair of i.i.d.~quantum states is fully
characterized by the ratio of their quantum relative entropies.

\end{abstract}
\volumeyear{ }
\volumenumber{ }
\issuenumber{ }
\eid{ }
\date{\today}
\startpage{1}
\endpage{10}
\maketitle



\section{Introduction}

Distinguishability plays a central role in all sciences. That is, the ability
to distinguish one possibility from another is what allows us to discover new
scientific laws and make predictions of future possibilities. In the process
of scientific discovery, we form a hypothesis based on conjecture, which is to
be tested against a conventional or null hypothesis by repeated trials or
experiments. With sufficient statistical evidence, one can determine which
hypothesis should be rejected in favor of the other. If the null hypothesis is
accepted, one can form alternative hypotheses to test against the null
hypothesis in future experiments.

What is essential in this approach is the ability to perform repeated trials.
Repetition allows for increasing the distinguishability between the two
hypotheses. A natural question in this context is to determine how many trials
are required to reach a given conclusion. If the two different hypotheses are
relatively distinguishable, then fewer trials are required to decide between
the possibilities. In this sense, distinguishability can be understood as a
\textit{resource}, because it limits the amount of effort that we need to
invest in order to make decisions.

One of the fundamental settings in which distinguishability can be studied in
a mathematically rigorous manner is statistical hypothesis testing. The basic
setup is that one draws a sample $x$\ from one of two probability
distributions $p\equiv\{p(x)\}_{x\in\mathcal{X}}$ or $q\equiv\{q(x)\}_{x\in
\mathcal{X}}$, with common alphabet $\mathcal{X}$, with the goal being to
decide from which distribution the sample $x$ has been drawn. Let $p$ be the
null hypothesis and $q$ the alternative. A Type~I error occurs if one decides
$q$ when the distribution being sampled from is in fact $p$, and a Type~II
error occurs if one decides $p$ when the distribution being sampled from is in
fact $q$. The goal of asymmetric hypothesis testing is to minimize the
probability of a Type~II\ error, subject to an upper bound constraint on the
probability of committing a Type~I error.

In the scientific spirit of repeated experiments, we can modify the above
scenario to allow for independent and identically distributed (i.i.d.)~samples
from either the distribution $p$ or $q$. One of the fundamental results of
asymptotic hypothesis testing is that, with a sufficiently large number of
samples, it becomes possible to meet any upper bound constraint on the
Type~I\ error probability while having the Type~II\ error probability decaying
exponentially fast with the number of samples, with the optimal error exponent
being given by the relative entropy \cite{SteinLemma,chernoff1956}:%
\begin{equation}
D(p\Vert q)=\sum_{x\in\mathcal{X}}p(x)\log_{2}[p(x)/q(x)].
\end{equation}
That is, there exists a sequence of schemes that can achieve this error
exponent for the Type~II\ error probability while making the Type~I\ error
probability arbitrarily small in the limit of a large number of samples. At the
same time, the strong converse property holds:\ any sequence of schemes that
has a fixed constraint on the Type~I error probability is such that its
Type~II\ error probability cannot decay any faster than the exponent $D(p\Vert
q)$. This gives a fundamental operational meaning to the relative entropy and
represents one core link between hypothesis testing and information theory
\cite{B74}, the latter being the fundamental mathematical theory of
communication \cite{bell1948shannon}.

Another perspective on the above process of decision making in hypothesis
testing, the \textit{resource-theoretic perspective} \cite{Mats10,M11}\ not
commonly adopted in the literature on the topic, is that it is a process by
which we \textit{distill} distinguishability from the original distributions
into a more standard form. That is, we can think of the distributions $p$ and
$q$ being presented as a black box or ordered pair $(p,q)$. Given a sample
$x\in\mathcal{X}$, we can perform a common transformation $\mathcal{T}%
:\mathcal{X}\rightarrow\{0,1\}$ that outputs a single bit, \textquotedblleft%
$0$\textquotedblright\ to decide $p$ and \textquotedblleft$1$%
\textquotedblright\ to decide $q$. The common transformation $\mathcal{T}$ can
even be stochastic. In this way, one transforms the initial box to a final box
as%
\begin{equation}
(p,q)\ \ \underrightarrow{\mathcal{T}}\ \ (p_{f},q_{f}),
\end{equation}
where $p_{f}\equiv\{p_{f}(y)\}_{y\in\left\{  0,1\right\}  }$ and $q_{f}%
\equiv\{q_{f}(y)\}_{y\in\left\{  0,1\right\}  }$ are binary distributions.
Then the probability of a Type~I error is $p_{f}(1)$, and the probability of a
Type~II\ error is $q_{f}(0)$. Since the goal is to extract or distill as much
distinguishability as possible, we would like for $q_{f}(0)$ to be as small as
possible given a constraint $\varepsilon\in\left[  0,1\right]  $ on $p_{f}(1)$
(i.e., $p_{f}(1)\leq\varepsilon$).

Once we have adopted this resource-theoretic approach to distinguishability,
it is natural to consider two other questions, the first of which is the
question of the \textit{reverse process} \cite{Mats10,M11}. That is, we would
like to start from initial binary distributions $p_{i}\equiv\{p_{i}%
(y)\}_{y\in\left\{  0,1\right\}  }$ and $q_{i}\equiv\{q_{i}(y)\}_{y\in\left\{
0,1\right\}  }$\ having as little distinguishability as possible, and act on
their samples with a common transformation $\mathcal{R}:\{0,1\}\rightarrow
\mathcal{X}$ in order to produce the distributions $p\equiv\{p(x)\}_{x\in
\mathcal{X}}$ and $q\equiv\{q(x)\}_{x\in\mathcal{X}}$, while allowing for a
slight error when reproducing $p$. That is, we would like to perform the
\textit{dilution} transformation%
\begin{equation}
(p_{i},q_{i})\ \ \underrightarrow{\mathcal{R}}\ \ (\tilde{p},q),
\label{eq:reverse-process}%
\end{equation}
where $\tilde{p}\equiv\{\tilde{p}(x)\}_{x\in\mathcal{X}}$ is a distribution
satisfying $d(p,\tilde{p})\leq\varepsilon$, for some suitable metric $d$ of
statistical distinguishability. In this way, we characterize the
distinguishability of $p$ and $q$ in terms of the least distinguishable
distributions $p_{i}$ and $q_{i}$ that can be diluted to prepare or simulate
$p$ and $q$, respectively. This dilution question is motivated by related
questions in the theory of quantum entanglement~\cite{BDSW96}.

The second, more general question is regarding the existence of a common
transformation $\mathcal{T}:\mathcal{X}\rightarrow\mathcal{Z}$ that converts
initial distributions $p$ and $q$ into final distributions $r\equiv
\{r(z)\}_{z\in\mathcal{Z}}$ and $t\equiv\{t(z)\}_{z\in\mathcal{Z}}$:%
\begin{equation}
(p,q)\ \ \underrightarrow{\mathcal{T}}\ \ (\tilde{r},t),
\label{eq:gen-trans-question}%
\end{equation}
where $\tilde{r}\equiv\{\tilde{r}(z)\}_{z\in\mathcal{Z}}$ is a distribution
satisfying $d(r,\tilde{r})\leq\varepsilon$. One can then ask about the rate or
efficiency at which it is possible to convert a pair of i.i.d.~distributions
to another pair of i.i.d.~distributions.

This resource-theoretic approach to distinguishability offers a unique and
powerful perspective on statistical hypothesis testing and distinguishability,
similar to the perspective brought about by the seminal work on the resource
theory of quantum entanglement \cite{BDSW96}, which has in turn inspired a
flurry of activity on resource theories in quantum information and beyond
\cite{CG18}. Although the reverse process in \eqref{eq:reverse-process} may
seem nonsensical at first glance (why would one want to dilute fresh water to
salt water? \cite{ieee2002bennett}), it plays a fundamental role in
characterizing distinguishability as a resource, as well as for addressing the
general question posed in \eqref{eq:gen-trans-question}. It is also natural
from a thermodynamic or physical perspective to consider reversibility and
cyclicity of processes. Another application for the reverse process is in
understanding the minimal resources required for simulation in various quantum
resource theories \cite{CG18}.

\section{Main results}

The main goal of this paper is to develop systematically the
resource-theoretic perspective on distinguishability, which was initiated in
\cite{Mats10,M11}. More precisely, the theory developed here is a
\textit{resource theory of asymmetric distinguishability}, given that
approximation is allowed for the first distribution in all of the
distillation, dilution, and general transformation tasks mentioned above. The
theory that we develop applies in the more general setting of \textit{quantum}
distinguishability, as it did in \cite{Mats10,M11}, in particular when the
distributions $p$ and $q$ are replaced by quantum states $\rho$ and $\sigma$,
respectively, and the common transformations allowed on a quantum box
$(\rho,\sigma)$ are quantum channels.

Some key findings of our work are as follows:

\begin{enumerate}
\item We introduce the fundamental unit or currency of this resource theory,
dubbed \textquotedblleft bits of asymmetric
distinguishability.\textquotedblright\ Then the distinguishability
distillation and dilution tasks amount to distilling bits of asymmetric
distinguishability from a box $(\rho,\sigma)$ and diluting bits of asymmetric
distinguishability to a box $(\rho,\sigma)$, respectively.

\item We formally define the exact one-shot distinguishability distillation
and dilution tasks, and we prove that the optimal number of bits of asymmetric
distinguishability that can be distilled from a box $(\rho,\sigma)$ is equal
to the min-relative entropy \cite{D09}\ (see
\eqref{eq:exact-distillable-distinguishability}), while the optimal number of
bits of asymmetric distinguishability that can be diluted to a box
$(\rho,\sigma)$ is equal to the max-relative entropy \cite{D09}\ (see
\eqref{eq:dist-cost}), giving both of these quantities fundamental operational
interpretations in the resource theory of asymmetric distinguishability.

\item We define the approximate one-shot distinguishability distillation and
dilution tasks, and we prove that the optimal number of bits of asymmetric
distinguishability that can be distilled from a box~$(\rho,\sigma)$ is equal
to the smooth min-relative entropy \cite{BD10,BD11,WR12}\ (see
\eqref{eq:approx-distill}), while the optimal number of bits of asymmetric
distinguishability that can be diluted to a box $(\rho,\sigma)$ is equal to
the smooth max-relative entropy \cite{D09} (see \eqref{eq:eps-approx-cost}),
giving both of these quantities fundamental operational interpretations in the
resource theory of asymmetric distinguishability.

\item We prove that the optimization problems corresponding to one-shot
distinguishability distillation and dilution, as well as the optimization
corresponding to the quantum generalization of the transformation problem
considered in \eqref{eq:gen-trans-question}, are characterized by
semi-definite programs (see Appendices~\ref{app:SDPs-smooth-min-max} and \ref{app:approx-box-tr-SDP}). Thus, all of these quantities can be computed efficiently.

\item We finally consider the asymptotic version of the resource theory and
prove that it is reversible in this setting, with the optimal rate of
distillation or dilution equal to the quantum relative entropy. The
implication of this result is that the rate or efficiency at which a pair of
i.i.d.~quantum states can be converted to another pair of i.i.d.~quantum
states is fully characterized by the ratio of their quantum relative entropies
(see~\eqref{eq:box-trans-main-result}). 
\end{enumerate}

In what follows, we provide more details of the resource theory of asymmetric
distinguishability and a full exposition of the main results stated above. We
relegate details of mathematical proofs to several appendices, and we note
here that some of the technical lemmas in the appendices may be of independent interest.

As far as we are aware, the first proposal for a resource theory of
distinguishability was given in \cite{Mats10,M11}, which we have highlighted
above. It appears that this aspect of the work \cite{Mats10,M11} has gone
largely unnoticed since its posting to the arXiv, given that there have been
several subsequent proposals or calls to formalize a resource theory of
distinguishability \cite{M09,BK15,BK17} that apparently were not aware of~\cite{Mats10,M11}.

\section{Resource theory of asymmetric distinguishability}

We begin by establishing the basics of the resource theory of asymmetric
distinguishability. The basics include the objects being manipulated, called
\textquotedblleft boxes,\textquotedblright\ the fundamental units of resource,
\textquotedblleft bits of asymmetric distinguishability,\textquotedblright%
\ and the free operations allowed, which are simply arbitrary quantum physical operations.

The basic object to manipulate in the resource theory of asymmetric
distinguishability is the following \textquotedblleft box\textquotedblright or
ordered pair:%
\begin{equation}
( \rho,\sigma) , \label{eq:basic-box}%
\end{equation}
where $\rho$ and $\sigma$ are quantum states acting on the same Hilbert space.
The interpretation of the box $( \rho,\sigma) $ is that it corresponds to two
different experiments or scenarios. In the first, the state $\rho$ is
prepared, and in the second, the state $\sigma$ is prepared. The box is handed
to another party, who is not aware of which experiment is being conducted
(i.e., which state has been prepared).

One basic manipulation in this resource theory is to transform this box into
another box by means of any quantum physical operation $\mathcal{N}$, as
allowed by quantum mechanics. Such physical operations are mathematically
described by completely positive, trace-preserving (CPTP) maps and are known
as quantum channels. By acting on the box $(\rho,\sigma)$ with the common
quantum channel $\mathcal{N}$, one obtains the transformed box $(\mathcal{N}%
(\rho),\mathcal{N}(\sigma))$. Observe that it is not necessary to know which
experiment is being conducted in order to perform this transformation; one can
perform it regardless of whether $\rho$ or $\sigma$ was prepared. For this
reason, all quantum channels are allowed for free in this resource theory, so
that the transformation
\begin{equation}
(\rho,\sigma)\quad \underrightarrow{\mathcal{N}%
}\quad (\mathcal{N}(\rho),\mathcal{N}(\sigma))
\end{equation}
is allowed for free.

If the channel being performed to transform the box in \eqref{eq:basic-box} is
an isometric channel $\mathcal{U}(\omega)=U\omega U^{\dag}$ (where $U$ is an
isometry satisfying $U^{\dag}U=I$ and $\omega$ is an arbitrary state),
resulting in the box%
\begin{equation}
(\mathcal{U}(\rho),\mathcal{U}(\sigma)),
\end{equation}
then it is possible to invert this transformation\ and return to the original
box in \eqref{eq:basic-box}. A quantum channel that inverts the action of
$\mathcal{U}$ is given by
\begin{equation}
\theta\rightarrow U^{\dag}\theta U+\operatorname{Tr}[(I-UU^{\dag})\theta
]\tau,\label{eq:invert-isometry}%
\end{equation}
where $\theta$ is an arbitrary state and $\tau$ is some state.

Another kind of
invertible transformation is the appending channel $\mathcal{A}_{\tau}%
(\omega)=\omega\otimes\tau$, which appends the state~$\tau$ and has the
following effect on the box:%
\begin{equation}
(\mathcal{A}_{\tau}(\rho),\mathcal{A}_{\tau}(\sigma))=(\rho\otimes\tau
,\sigma\otimes\tau).\label{eq:appending-channel}%
\end{equation}
One can recover the original box $(\rho,\sigma)$ from
\eqref{eq:appending-channel}\ by discarding the second system (described
mathematically by partial trace). Thus, isometric channels and appending
channels are perfectly reversible operations in this resource theory.

The fundamental goal of this resource theory is to determine how and whether
it is possible to transform an initial box $(\rho,\sigma)$ to another box
$(\tau,\omega)$ for states $\tau$ and$~\omega$, by means of a common quantum
channel $\mathcal{N}$. Mathematically, the question is to determine, for fixed
states $\rho$, $\sigma$, $\tau$, and $\omega$, whether there exists a
completely positive and trace-preserving map $\mathcal{N}$ such that
$\mathcal{N}(\rho)=\tau$ and $\mathcal{N}(\sigma)=\omega$. As it turns out,
various instantiations of this question have been studied considerably in
prior work
\cite{B53,AU80,CJW04,MOA11,Buscemi2012,HJRW12,BDS14,BHNOW15,Ren16,BD16,Buscemi2016,GJBDM18,B17,BG17}%
, and a variety of results are known regarding it. In this paper, we offer a
fresh resource-theoretic perspective on this matter.

Motivated by practical concerns, one important variation of the aforementioned
box transformation problem is to determine whether it is possible to
accomplish the transformation \textit{approximately} as
\begin{equation}
(\rho,\sigma
)\quad \underrightarrow{\mathcal{N}}\quad (\tau_{\varepsilon},\omega)
\end{equation}
 with some
tolerance $\varepsilon\in\left[  0,1\right]  $ allowed, such that the state
$\tau_{\varepsilon}$ is $\varepsilon$-close to the desired $\tau$. The precise
way in which we allow some tolerance is motivated exclusively by operational
concerns. In a single run of the first experiment in which $\rho$ is prepared,
the transformation $\mathcal{N}(\rho)=\tau_{\varepsilon}$ occurs. Then a third
party would like to assess how accurate the conversion is. Such an individual
can do so by performing a quantum measurement $\{\Lambda_{x}\}_{x}$ with
outcomes $x$ (satisfying $\Lambda_{x}\geq0$ for all $x$ and $\sum_{x}%
\Lambda_{x}=I$). The probability of obtaining a particular outcome
$\Lambda_{x}$ is given by the Born rule $\operatorname{Tr}[\Lambda_{x}%
\tau_{\varepsilon}]$. What we demand is that the deviation between the actual
probability $\operatorname{Tr}[\Lambda_{x}\tau_{\varepsilon}]$ and the ideal
probability $\operatorname{Tr}[\Lambda_{x}\tau]$ be no larger than the
tolerance $\varepsilon$. Since this should be the case for any possible
measurement outcome, what we demand mathematically is that
\begin{equation}
\sup_{0\leq\Lambda\leq I}\left\vert \operatorname{Tr}[\Lambda\tau
_{\varepsilon}]-\operatorname{Tr}[\Lambda\tau]\right\vert \leq\varepsilon
.\label{eq:trace-distance-error}%
\end{equation}
It is well known that%
\begin{equation}
\sup_{0\leq\Lambda\leq I}\left\vert \operatorname{Tr}[\Lambda\tau
_{\varepsilon}]-\operatorname{Tr}[\Lambda\tau]\right\vert =\frac{1}%
{2}\left\Vert \tau_{\varepsilon}-\tau\right\Vert _{1},
\end{equation}
indicating that our notion of approximation is most naturally quantified by
the normalized trace distance $\frac{1}{2}\left\Vert \tau_{\varepsilon}%
-\tau\right\Vert _{1}$.

Thus, the mathematical formulation of the approximate box transformation
problem is as follows:%
\begin{multline}
\varepsilon((\rho,\sigma)\rightarrow(\tau,\omega
)):=\label{eq:approx-box-trans}\\
\inf_{\mathcal{N}\in\text{CPTP}}\left\{  \varepsilon\in\left[  0,1\right]
:\mathcal{N}(\rho)\approx_{\varepsilon}\tau,\ \mathcal{N}(\sigma
)=\omega\right\}  ,
\end{multline}
where the notation $\zeta\approx_{\varepsilon}\xi$ for states $\zeta$ and
$\xi$ is a shorthand for $\frac{1}{2}\left\Vert \zeta-\xi\right\Vert _{1}%
\leq\varepsilon$; i.e.,%
\begin{equation}
\zeta\approx_{\varepsilon}\xi\qquad\Longleftrightarrow\qquad\frac{1}%
{2}\left\Vert \zeta-\xi\right\Vert _{1}\leq\varepsilon.
\end{equation}
The fact that we allow for approximate conversion for the first state but not
the second is related to the fact that the resource theory presented here is a
resource theory of \textit{asymmetric} distinguishability. In
Appendix~\ref{app:approx-box-tr-SDP}, we show that \eqref{eq:approx-box-trans}
is equivalent to a semi-definite program (SDP), implying that it is
efficiently computable with respect to the dimensions of the states involved.
In the case that $\varepsilon((\rho,\sigma)\rightarrow(\tau,\omega))=0$, this
means that it is possible to perform the desired transformation $(\rho
,\sigma)\rightarrow(\tau,\omega)$ exactly, reproducing the previous result
from \cite{GJBDM18}.

We can also consider the asymptotic version of the box transformation problem,
in which the box consists not just of a single copy of the states $\rho$ and
$\sigma$ but many copies of them (i.e., the box $(\rho^{\otimes n}%
,\sigma^{\otimes n})$ instead of the original $(\rho,\sigma)$). By considering
the asymptotic setting with approximation error, we can modify the original
box transformation question as follows: what is the optimal rate $R$ at which
the transformation
\begin{equation}
(\rho^{\otimes n},\sigma^{\otimes n})\rightarrow(\widetilde{\tau^{\otimes nR}%
},\omega^{\otimes nR})\label{eq:approx-box-transform-fundamental}%
\end{equation}
is possible, for large $n$ and arbitrarily small approximation error? In this
setting, the SDP\ characterization of $\varepsilon((\rho^{\otimes n}%
,\sigma^{\otimes n})\rightarrow(\tau^{\otimes nR},\omega^{\otimes nR}))$ is
not particularly useful, due to the fact that the computational complexity of
the optimization problem grows exponentially with increasing $n$, and so we
resort to other, information-theoretic methods to address it.

\subsection{Bits of asymmetric distinguishability}

\label{sec:bits-AD}One way of addressing the various formulations of the box
transformation problem is to break the transformation down into two steps, in
which we first \textit{distill} a standard box and then \textit{dilute} this
standard box to the desired one. It turns out that the most natural way to do
so is to consider the following basic unit of currency or fiducial box:%
\begin{equation}
(|0\rangle\langle0|,\pi),\label{eq:bit-AD}%
\end{equation}
where
\begin{equation}
\pi:=\frac{1}{2}\left(  |0\rangle\langle0|+|1\rangle\langle1|\right) 
\end{equation}
is the maximally mixed qubit state. We also refer to the object in
\eqref{eq:bit-AD} as \textquotedblleft one bit of asymmetric
distinguishability.\textquotedblright

As before, we should think of the box in \eqref{eq:bit-AD} as being in
correspondence with two different experiments. In the first experiment, the
first state $\rho=|0\rangle\langle0|$ (\textquotedblleft null
hypothesis\textquotedblright) is prepared, and in the second experiment, the
second state $\sigma=\pi$ (\textquotedblleft alternative
hypothesis\textquotedblright) is prepared. A distinguisher presented with this
box, and unaware of which experiment is being conducted, can try to determine
which state $\rho$ or $\sigma$ has been prepared. Suppose that the
distinguisher performs a measurement of the observable $\sigma_{Z}%
:=|0\rangle\langle0|-|1\rangle\langle1|$ and assigns the outcome $+1$ to the
decision \textquotedblleft$\rho$ was prepared\textquotedblright\ and $-1$ to
the decision \textquotedblleft$\sigma$ was prepared.\textquotedblright\ Then
in the case that the state $\rho$ was prepared, he can determine this with
zero chance of error; on the other hand, if the state $\sigma$ was prepared,
then he can determine this with probability equal to $1/2$. In other terms,
with this strategy, he has zero chance of making a Type~I error
(misidentifying $\rho$) and he has a 50\% chance of making a Type~II error
(misidentifying $\sigma$).

The above strategy of basing the decision rule on the outcome of a $\sigma
_{Z}$ measurement is not the only strategy that the distinguisher can perform.
By performing a quantum channel $\mathcal{N}$ that accepts a qubit as input
and outputs another quantum system, the distinguisher can convert the box in
\eqref{eq:bit-AD}\ to the following box:%
\begin{equation}
(\mathcal{N}(|0\rangle\langle0|),\mathcal{N}(\pi)).
\end{equation}
After doing so, the distinguisher can base his decision rule on the outcome of
a general quantum measurement. However, if the goal is to have zero chance of
making a Type~I error, then it is intuitive and can be proven that no strategy
can perform better than the $\sigma_{Z}$ measurement strategy given in the
previous paragraph. Thus, arbitrary channels acting on the box in
\eqref{eq:bit-AD} do not increase distinguishability.

One bit of asymmetric distinguishability is not a particularly strong
resource. Indeed, with only one bit of asymmetric distinguishability, there is
still a large chance of making a Type~II error. However, the following box,
consisting of $m$ bits of asymmetric distinguishability, improves the
situation:%
\begin{equation}
(|0\rangle\langle0|^{\otimes m},\pi^{\otimes m}).\label{eq:m-bits-AD}%
\end{equation}
For such a box, there is a much smaller chance of making a Type~II error.
Indeed, by performing $m$ independent measurements of the observable
$\sigma_{Z}$ on each qubit and assigning the outcome \textquotedblleft%
$(+1,\ldots,+1)$\textquotedblright\ to the decision \textquotedblleft%
$|0\rangle\langle0|^{\otimes m}$ was prepared\textquotedblright\ and the
outcome \textquotedblleft not $(+1,\ldots,+1)$\textquotedblright\ to the
decision \textquotedblleft$\pi^{\otimes m}$ was prepared,\textquotedblright%
\ the distinguisher still has zero chance of making a Type~I\ error, but now
has a one out of $2^{m}$ chance of making a Type~II error. So with each extra
bit of asymmetric distinguishability, the chance of making a Type~II error
decreases by a factor of two. This is the value of having more bits of
asymmetric distinguishability.

Note that the following transformation is forbidden when $n>m$:%
\begin{equation}
(|0\rangle\langle0|^{\otimes m},\pi^{\otimes m})\not \rightarrow
(|0\rangle\langle0|^{\otimes n},\pi^{\otimes n}).\label{eq:outlawed-trans}%
\end{equation}
That is, one cannot increase bits of distinguishability by the action of a
quantum channel; i.e., there is no quantum channel $\mathcal{N}$\ that
performs the map $\mathcal{N}(|0\rangle\langle0|^{\otimes m})=|0\rangle
\langle0|^{\otimes n}$ and $\mathcal{N}(\pi^{\otimes m})=\pi^{\otimes n}$ for
$n>m$. Quantum channels have a linear action on their inputs, and this
linearity forbids such transformations, as shown in
Appendix~\ref{app:impossibility-result}.

A major goal of any resource theory is to quantify the amount of resource. For
the simple boxes presented above, any R\'{e}nyi relative entropy suffices as a
good quantifier of the number of bits of asymmetric distinguishability
contained in them. Two prominent examples of measures were put forward roughly
a decade ago as measures of distinguishability and studied therein as quantum
information-theoretic quantities of interest \cite{D09}. They are known as the
min- and max-relative entropies, defined respectively as follows for states
$\rho$ and$~\sigma$:%
\begin{align}
D_{\min}(\rho\Vert\sigma)  &  :=-\log_{2}\operatorname{Tr}[\Pi_{\rho}%
\sigma],\label{eq:order-zero-Renyi-rel-ent}\\
D_{\max}(\rho\Vert\sigma)  &  :=\inf\left\{  \lambda\geq0:\rho\leq2^{\lambda
}\sigma\right\}  , \label{eq:max-rel-ent}%
\end{align}
where $\Pi_{\rho}$ denotes the projection onto the support of $\rho$. If
$\rho$ is orthogonal to $\sigma$, then $D_{\min}(\rho\Vert\sigma)=\infty$, and
if $\operatorname{supp}(\rho)\not \subseteq \operatorname{supp}(\sigma)$, then
there is no finite $\lambda\geq0$ such that $\rho\leq2^{\lambda}\sigma$,
implying that $D_{\max}(\rho\Vert\sigma)=\infty$. Evaluating these measures
for the box given in \eqref{eq:m-bits-AD}, one finds that%
\begin{align}
D_{\min}(|0\rangle\langle0|^{\otimes m}\Vert\pi^{\otimes m})  &  =mD_{\min
}(|0\rangle\langle0|\Vert\pi)=m,\\
D_{\max}(|0\rangle\langle0|^{\otimes m}\Vert\pi^{\otimes m})  &  =mD_{\min
}(|0\rangle\langle0|\Vert\pi)=m,
\end{align}
consistent with the notion that the box in \eqref{eq:m-bits-AD} contains $m$
bits of asymmetric distinguishability.

By performing the following quantum channel:%
\begin{equation}
\omega\rightarrow\operatorname{Tr}[|0\rangle\langle0|^{\otimes m}%
\omega]|0\rangle\langle0|+\operatorname{Tr}[(I^{\otimes m}-|0\rangle
\langle0|^{\otimes m})\omega]|1\rangle\langle1|,
\end{equation}
one can convert the box in \eqref{eq:m-bits-AD}\ to the following box:%
\begin{equation}
(|0\rangle\langle0|,2^{-m}|0\rangle\langle0|+\left(  1-2^{-m}\right)
|1\rangle\langle1|).\label{eq:compressed-box}%
\end{equation}
Furthermore, by performing the quantum channel%
\begin{equation}
\theta\rightarrow\langle0|\theta|0\rangle|0\rangle\langle0|^{\otimes
m}+\langle1|\theta|1\rangle\frac{I^{\otimes m}-|0\rangle\langle0|^{\otimes m}%
}{2^{m}-1},\label{eq:transform-to-m-bits}%
\end{equation}
one can convert the box in \eqref{eq:compressed-box} back to the box in
\eqref{eq:m-bits-AD}.\ For this reason, these boxes have an equivalent number
of bits of asymmetric distinguishability, being equivalent by free operations.
It also means that we can take the box in \eqref{eq:compressed-box} to be the
basic form of $m$ bits of asymmetric distinguishability. Once we have done
that, it is then sensible to allow $m$ in \eqref{eq:compressed-box} to be any
non-negative real number, so that the box in \eqref{eq:compressed-box} has $m$
bits of asymmetric distinguishability, with $m$ a non-negative real number.
For this case, we still find that%
\begin{equation}
D_{\min}(|0\rangle\langle0|\Vert\sigma)=D_{\max}(|0\rangle\langle0|\Vert
\sigma)=m,
\end{equation}
with $\sigma=2^{-m}|0\rangle\langle0|+\left(  1-2^{-m}\right)  |1\rangle
\langle1|$.

\textit{Going forward from here, we take the box in \eqref{eq:compressed-box}
to be the basic form of }$m$\textit{ bits of asymmetric distinguishability,
for }$m$\textit{ any non-negative real number.}

\subsection{Exact distillation and dilution tasks}

In any resource theory, the basic questions concern distillation and dilution
tasks, and whether and in what senses the resource theory might be reversible
\cite{BDSW96,CG18}. In a distillation task, the goal is to process a general
resource with free operations in order to distill as much of the basic
resource as possible, while in the dilution task, the goal is to perform the
opposite: process as little of the basic resource as possible, using free
operations, in order to generate or dilute from it a more general resource. A
prominent goal is to determine the ultimate rates at which these resource
interconversions are possible and from there one can determine whether the
resource theory is reversible.

In the resource theory of asymmetric distinguishability, the goal of
\textit{exact distinguishability distillation }is to process a general box
$(\rho,\sigma)$ with an arbitrary quantum channel in order to distill as many
bits of asymmetric distinguishability as possible. Mathematically, we can
phrase this task as the following optimization problem:
\begingroup\allowdisplaybreaks[0]
\begin{multline}
D_{d}^{0}(\rho,\sigma):=\\
\log_{2}\sup_{\mathcal{P}\in\text{CPTP}}\left\{  M:\mathcal{P}(\rho
)=|0\rangle\langle0|,\ \mathcal{P}(\sigma)=\pi_{M}\right\}  ,
\end{multline}
\endgroup where the choice of $D_{d}$ in $D_{d}^{0}(\rho,\sigma)$ stands for
\textit{distillable distinguishability},\ the \textquotedblleft$0$%
\textquotedblright\ in $D_{d}^{0}(\rho,\sigma)$ indicates that we do not allow
any error, CPTP\ denotes the set of CPTP\ maps (quantum channels), and%
\begin{equation}
\pi_{M}:=\frac{1}{M}|0\rangle\langle0|+\left(  1-\frac{1}{M}\right)
|1\rangle\langle1|.\
\end{equation}
As we show in Appendix~\ref{app:exact-distill-distin}, the following equality
holds%
\begin{equation}
D_{d}^{0}(\rho,\sigma)=D_{\min}(\rho\Vert\sigma
),\label{eq:exact-distillable-distinguishability}%
\end{equation}
where $D_{\min}(\rho\Vert\sigma)$ is the min-relative entropy \cite{D09}, as
defined in \eqref{eq:order-zero-Renyi-rel-ent}. The equality in
\eqref{eq:exact-distillable-distinguishability} thus assigns to $D_{\min}%
(\rho\Vert\sigma)$ a fundamental operational meaning as the exact distillable
distinguishability in the resource theory of asymmetric distinguishability. A
strongly related operational meaning for $D_{\min}(\rho\Vert\sigma)$ in
quantum hypothesis testing was already given in \cite{D09}. 

In the case that $\rho$ is orthogonal to $\sigma$, then this means that the
box $(\rho,\sigma)$ can be converted to the box $(|0\rangle\langle
0|,|1\rangle\langle1|)$, by means of the quantum channel%
\begin{equation}
\omega\rightarrow\operatorname{Tr}[\Pi_{\rho}\omega]|0\rangle\langle
0|+\operatorname{Tr}[\left(  I-\Pi_{\rho}\right)  \omega]|1\rangle\langle1|.
\end{equation}
From the latter box, one can obtain as many bits of asymmetric
distinguishability as desired. Indeed by performing the channel%
\begin{equation}
\mathcal{T}^{m}(\omega)=\langle0|\omega|0\rangle\ |0\rangle\langle
0|+\langle1|\omega|1\rangle\pi_{2^{m}},\label{eq:transform-ch-inf-bits}%
\end{equation}
where $\pi_{2^{m}}:=2^{-m}|0\rangle\langle0|+\left(  1-2^{-m}\right)
|1\rangle\langle1|$, one can obtain $m$ bits of asymmetric distinguishability
from the box $(|0\rangle\langle0|,|1\rangle\langle1|)$. Since this is possible
for any $m\geq0$, it follows that the box $(|0\rangle\langle0|,|1\rangle
\langle1|)$ has an infinite number of bits of asymmetric distinguishability,
consistent with the fact that $D_{\min}(\rho\Vert\sigma)=\infty$ when $\rho$
is orthogonal to $\sigma$.

The goal of \textit{exact distinguishability dilution} is the
opposite:\ process as few bits of asymmetric distinguishability as possible,
using free operations, in order to generate the box $(\rho,\sigma)$.
Mathematically, we can phrase this task as the following optimization problem:%
\begin{multline}
D_{c}^{0}(\rho,\sigma):=\\
\log_{2}\inf_{\mathcal{P}\in\text{CPTP}}\left\{  M:\mathcal{P}(|0\rangle
\langle0|)=\rho,\ \mathcal{P}(\pi_{M})=\sigma\right\}  ,
\end{multline}
where the choice of $D_{c}$ in $D_{c}^{0}(\rho,\sigma)$ stands for
\textit{distinguishability cost}\ and\ the \textquotedblleft$0$%
\textquotedblright\ in $D_{c}^{0}(\rho,\sigma)$ again indicates that we do not
allow any error. As we show in Appendix~\ref{app:exact-dist-cost}, the
following equality holds%
\begin{equation}
D_{c}^{0}(\rho,\sigma)=D_{\max}(\rho\Vert\sigma),\label{eq:dist-cost}%
\end{equation}
where $D_{\max}(\rho\Vert\sigma)$ is the max-relative entropy \cite{D09}, as
defined in \eqref{eq:max-rel-ent}. The equality in \eqref{eq:dist-cost}\ thus
assigns to the max-relative entropy $D_{\max}(\rho\Vert\sigma)$\ a fundamental
operational meaning as the exact distinguishability cost of the box$~(\rho
,\sigma)$.

In the case that the support of $\rho$ is not contained in the support of
$\sigma$, then there is no finite value of $M$ nor any quantum channel
$\mathcal{P}$ that performs the transformations $\mathcal{P}(|0\rangle
\langle0|)=\rho$ and $\mathcal{P}(\pi_{M})=\sigma$. However, in the limit
$M\rightarrow\infty$, the box $(|0\rangle\langle0|,\pi_M)$ becomes the box
$(|0\rangle\langle0|,|1\rangle\langle1|)$, which is interpreted as containing
an infinite number of bits of asymmetric distinguishability, as discussed
after \eqref{eq:transform-ch-inf-bits}. In this case, we can pick the channel
$\mathcal{P}$ as $\mathcal{P}(\omega)=\langle0|\omega|0\rangle\rho
+\langle1|\omega|1\rangle\sigma$, and then the transformation $\mathcal{P}%
(|0\rangle\langle0|)=\rho$ and $\mathcal{P}(|1\rangle\langle1|)=\sigma$ is
easily achieved. Thus, in this sense, if the support of $\rho$ is not
contained in the support of $\sigma$, then the distinguishability cost
$D_{c}^{0}(\rho,\sigma)=\infty$, consistent with the fact that $D_{\max}%
(\rho\Vert\sigma)=\infty$ in this case.

An important case to consider in any resource theory is the case of
independent and identically distributed (i.i.d.) resources. For our case, this
means that we should analyze the box $(\rho^{\otimes n},\sigma^{\otimes n})$
for arbitrary $n\geq1$. Due to the additivity of $D_{\min}(\rho\Vert\sigma)$
and $D_{\max}(\rho\Vert\sigma)$, it follows that%
\begin{align}
D_{d}^{0}(\rho^{\otimes n},\sigma^{\otimes n}) &  =nD_{\min}(\rho\Vert
\sigma),\\
D_{c}^{0}(\rho^{\otimes n},\sigma^{\otimes n}) &  =nD_{\max}(\rho\Vert\sigma),
\end{align}
so that the number of bits of asymmetric distinguishability distilled and
required in each respective task scales precisely linearly with $n$.

Due to the fact that we generally have $D_{\min}(\rho\Vert\sigma)\neq D_{\max
}(\rho\Vert\sigma)$ for states $\rho$ and $\sigma$, it follows that the
resource theory of asymmetric distinguishability is not reversible if we
demand exact conversions from one box to another. In fact, the irreversibility
in the exact case can be as extreme as desired. By picking $\rho
=|0\rangle\langle0|$ and $\sigma=|\psi\rangle\langle\psi|$ for $|\psi
\rangle=\sqrt{1-\delta}|0\rangle+\sqrt{\delta}|1\rangle$ and $\delta\in\left(
0,1\right)  $, we have that $D_{\min}(\rho\Vert\sigma)=-\log_{2}(1-\delta)$
while $D_{\max}(\rho\Vert\sigma)=\infty$ for all $\delta\in\left(  0,1\right)
$, so that the exact distillable distinguishability can be arbitrarily close
to zero while the exact distinguishability cost is always infinite in this case.

\subsection{Approximate distillation and dilution tasks}

In realistic experimental scenarios, it is typically not possible to perform
transformations exactly, thus motivating the need to consider approximate
transformations and approximations of the ideal resources. For the resource
theory of asymmetric distinguishability, we define an $\varepsilon
$-approximate bit of asymmetric distinguishability as%
\begin{equation}
(\widetilde{0}_{\varepsilon},\pi),\label{eq:approx-bit}%
\end{equation}
where $\varepsilon\in\left[  0,1\right]  $ and
\begin{equation}
\widetilde{0}_{\varepsilon}:=\left(  1-\varepsilon\right)  |0\rangle
\langle0|+\varepsilon|1\rangle\langle1|,
\end{equation}
so that $\widetilde{0}_{\varepsilon}\approx_{\varepsilon}|0\rangle\langle0|$.
The motivation for this choice is operational as before (see the discussion
before \eqref{eq:approx-box-trans}). Also, since the maximally mixed state
$\pi$ is diagonal in any basis, it suffices to consider
\eqref{eq:approx-bit}\ as the basic definition of an $\varepsilon$-approximate
bit of asymmetric distinguishability, because one could simply perform the
diagonalizing unitary for a general qubit state $\tau$ to bring a general box
$(  \tau,\pi)  $ into the form of~\eqref{eq:approx-bit}.

Generalizing  \eqref{eq:compressed-box} and \eqref{eq:approx-bit}, the
following box represents $m$ approximate bits of asymmetric
distinguishability:%
\begin{equation}
(\widetilde{0}_{\varepsilon},2^{-m}|0\rangle\langle0|+\left(  1-2^{-m}\right)
|1\rangle\langle1|).
\end{equation}
If $m$ is an integer, then this box is equivalent by the transformation in
\eqref{eq:transform-to-m-bits}\ to the following one:%
\begin{equation}
(\widetilde{0}_{\varepsilon}^{m},\pi^{\otimes m}),
\end{equation}
where
\begin{equation}
\widetilde{0}_{\varepsilon}^{m}:=\left(  1-\varepsilon\right)  |0\rangle
\langle0|^{\otimes m}+\varepsilon\frac{I^{\otimes m}-|0\rangle\langle
0|^{\otimes m}}{2^{m}-1},
\end{equation}
so that $\widetilde{0}_{\varepsilon}^{m}\approx_{\varepsilon}|0\rangle
\langle0|^{\otimes m}$.

With such a notion in place, we can now generalize exact distillation of
asymmetric distinguishability to its approximate version. The goal of
$\varepsilon$\textit{-approximate distinguishability distillation} is to
distill as many $\varepsilon$-approximate bits of asymmetric
distinguishability as possible from a given box $(\rho,\sigma)$.
Mathematically, it corresponds to the following optimization for
$\varepsilon\in\left[  0,1\right]  $:%
\begin{multline}
D_{d}^{\varepsilon}(\rho,\sigma):=\\
\log_{2}\sup_{\mathcal{P}\in\text{CPTP}}\{M:\mathcal{P}(\rho)\approx
_{\varepsilon}|0\rangle\langle0|,\ \mathcal{P}(\sigma)=\pi_{M}\}.
\end{multline}
As we show in Appendix~\ref{app:one-shot-distill-distin}, the following
equality holds%
\begin{equation}
D_{d}^{\varepsilon}(\rho,\sigma)=D_{\min}^{\varepsilon}(\rho\Vert
\sigma),\label{eq:approx-distill}%
\end{equation}
where $D_{\min}^{\varepsilon}(\rho\Vert\sigma)$ is the smooth min-relative
entropy \cite{BD10,BD11,WR12}, defined as%
\begin{equation}
D_{\min}^{\varepsilon}(\rho\Vert\sigma):=-\log_{2}\inf_{0\leq\Lambda\leq
I}\left\{  \operatorname{Tr}[\Lambda\sigma]:\operatorname{Tr}[\Lambda\rho
]\geq1-\varepsilon\right\}  .
\end{equation}
Thus, the equality in \eqref{eq:approx-distill}\ assigns to the smooth
min-relative entropy an operational meaning as the $\varepsilon$-approximate
distillable distinguishability of the box $(\rho,\sigma)$. This operational
interpretation is directly linked to the role of $D_{\min}^{\varepsilon}%
(\rho\Vert\sigma)$ in quantum hypothesis
testing~\cite{HP91,ON00,Hay03,Hay04,WR12,Hay17}. Note that $D_{\min
}^{\varepsilon}(\rho\Vert\sigma)$ is also known as \textquotedblleft
hypothesis testing relative entropy\textquotedblright\ in the literature, which is
terminology introduced in \cite{WR12}. This quantity can be computed efficiently by means of a semi-definite program~\cite{DKFRR12}, the proof of which we recall in Appendix~\ref{app:SDPs-smooth-min-max}.

Note that by combining \eqref{eq:exact-distillable-distinguishability},
\eqref{eq:approx-distill}, and the fact that $\lim_{\varepsilon\rightarrow
0}D_{d}^{\varepsilon}(\rho,\sigma)=D_{d}^{0}(\rho,\sigma)$, we conclude the
following limit:
\begin{equation}
\lim_{\varepsilon\rightarrow0}D_{\min}^{\varepsilon}(\rho\Vert\sigma)=D_{\min
}(\rho\Vert\sigma).\label{eq:limit-smooth-dmin-to-dmin}%
\end{equation}
We provide an alternative proof in Appendix~\ref{app:rel-ents-DP}.

We can also generalize the distinguishability dilution task to the approximate
case. In this case, we define the $\varepsilon$\textit{-approximate
distinguishability cost} of the box $(\rho,\sigma)$ to be the least number of
ideal bits of asymmetric distinguishability that are needed to generate the
box $(\rho_{\varepsilon},\sigma)$, where $\rho_{\varepsilon}\approx
_{\varepsilon}\rho$. This notion of approximate distinguishability cost is
fully operational and consistent with the more general problem in
\eqref{eq:approx-box-trans}. The precise definition of the $\varepsilon
$-approximate distinguishability cost of the box $(\rho,\sigma)$ is as
follows:%
\begin{multline}
D_{c}^{\varepsilon}(\rho,\sigma):=\\
\log_{2}\inf_{\mathcal{P}\in\text{CPTP}}\{M:\mathcal{P}(|0\rangle
\langle0|)\approx_{\varepsilon}\rho,\ \mathcal{P}(\pi_{M})=\sigma\}.
\end{multline}
As we show in Appendix~\ref{app:one-shot-dist-cost}, the following equality
holds%
\begin{equation}
D_{c}^{\varepsilon}(\rho,\sigma)=D_{\max}^{\varepsilon}(\rho\Vert
\sigma),\label{eq:eps-approx-cost}%
\end{equation}
where $D_{\max}^{\varepsilon}(\rho\Vert\sigma)$ is the smooth max-relative
entropy \cite{D09}, defined as%
\begin{equation}
D_{\max}^{\varepsilon}(\rho\Vert\sigma):=\inf_{\widetilde{\rho}:\frac{1}%
{2}\left\Vert \widetilde{\rho}-\rho\right\Vert _{1}\leq\varepsilon}D_{\max
}(\widetilde{\rho}\Vert\sigma).\label{eq:smooth-max-rel-ent}%
\end{equation}
Thus, the equality in \eqref{eq:eps-approx-cost}\ assigns to the smooth
max-relative entropy a fundamental operational meaning as the $\varepsilon
$-approximate distinguishability cost of the box$~(\rho,\sigma)$. The smooth max-relative entropy can also be efficiently calculated by means of a semi-definite program, the proof of which we give in Appendix~\ref{app:SDPs-smooth-min-max}.

Note that by combining \eqref{eq:dist-cost}, \eqref{eq:eps-approx-cost}, and
the fact that $\lim_{\varepsilon\rightarrow0}D_{c}^{\varepsilon}(\rho
,\sigma)=D_{c}^{0}(\rho,\sigma)$, we conclude the following limit:%
\begin{equation}
\lim_{\varepsilon\rightarrow0}D_{\max}^{\varepsilon}(\rho\Vert\sigma)=D_{\max
}(\rho\Vert\sigma).\label{eq:limit-smooth-dmax-to-dmax}%
\end{equation}
We provide an alternative proof in Appendix~\ref{app:rel-ents-DP}.

An application of the operational approach to distinguishability taken here is
the following bound relating$~D_{\min}^{\varepsilon}$ and $D_{\max
}^{\varepsilon}$:%
\begin{equation}
D_{\min}^{\varepsilon_{1}}(\rho\Vert\sigma)\leq D_{\max}^{\varepsilon_{2}%
}(\rho\Vert\sigma)+\log_{2}\!\left(  \frac{1}{1-\varepsilon_{1}-\varepsilon
_{2}}\right)  , \label{eq:dh-dmax-relation-operational}%
\end{equation}
where $\varepsilon_{1},\varepsilon_{2}\geq0$, and $\varepsilon_{1}%
+\varepsilon_{2}<1$. The bound in \eqref{eq:dh-dmax-relation-operational} is
most closely related to the upper bound in \cite[Theorem~11]{DMHB13}, but we
employ a different notion of smoothing for the smooth max-relative entropy. It
also generalizes the bound from \cite[Eq.~(47)]{DKFRR12} (by appropriately
working through the different conventions here and in \cite{DKFRR12}) and is
in the same spirit as \cite[Proposition~5.5]{tomamichel2012framework} and
\cite[Eq.~(22)]{tomamichel2013hierarchy}.

The main idea for arriving at the bound in
\eqref{eq:dh-dmax-relation-operational} follows from resource-theoretic
reasoning. Any approximate distillation protocol performed on the box
$(|0\rangle\langle0|,\pi_{M})$ that leads to the box $(\widetilde
{0}_{\varepsilon},\pi_{K})$, for $\varepsilon\in\lbrack0,1)$, is required to
obey the bound%
\begin{equation}
\log_{2}K\leq\log_{2}M+\log_{2}(1/[1-\varepsilon
]),\label{eq:fundamental-limit-Dh-Dmax}%
\end{equation}
which follows as a consequence of the fundamental limitation in
\eqref{eq:approx-distill}. One way to realize the transformation
$(|0\rangle\langle0|,\pi_{M})\rightarrow(\widetilde{0}_{\varepsilon},\pi_{K})$
is to proceed in two steps:\ first perform an optimal dilution protocol
$(|0\rangle\langle0|,\pi_{M})\rightarrow(\rho_{\varepsilon_{2}},\sigma)$ such
that $\log_{2}M=D_{\max}^{\varepsilon_{2}}(\rho\Vert\sigma)$ and then perform
an optimal distillation protocol $(\rho,\sigma)\rightarrow(\widetilde
{0}_{\varepsilon_{1}},\pi_{K})$ such that $\log_{2}K=D_{\min}^{\varepsilon
_{1}}(\rho\Vert\sigma)$. By employing the triangle inequality, the error of
the overall transformation is no larger than$~\varepsilon_{1}+\varepsilon_{2}%
$. Since the fundamental limitation in \eqref{eq:fundamental-limit-Dh-Dmax}
applies to any protocol, the bound in
\eqref{eq:dh-dmax-relation-operational}\ follows. We give a detailed proof in
Appendix~\ref{app:bnd-dmax-dheps}.

\subsection{Asymptotic distillable distinguishability and distinguishability
cost}

We can now reconsider the i.i.d.~case of a box $(\rho^{\otimes n}%
,\sigma^{\otimes n})$ in the context of approximate distillation and dilution.
Recall that the quantum relative entropy $D(\rho\Vert\sigma)$ is defined as
\cite{U62}%
\begin{equation}
D(\rho\Vert\sigma):=\operatorname{Tr}[\rho\left(  \log_{2}\rho-\log_{2}%
\sigma\right)  ],\label{eq:q-rel-ent}%
\end{equation}
if $\operatorname{supp}(\rho)\subseteq\operatorname{supp}(\sigma)$ and
$D(\rho\Vert\sigma)=\infty$ otherwise. By defining the asymptotic distillable
distinguishability and asymptotic distinguishability cost of the box
$(\rho,\sigma)$ as follows:%
\begin{align}
D_{d}(\rho,\sigma) &  :=\lim_{\varepsilon\rightarrow0}\lim_{n\rightarrow
\infty}\frac{1}{n}D_{d}^{\varepsilon}(\rho^{\otimes n},\sigma^{\otimes n}),\\
D_{c}(\rho,\sigma) &  :=\lim_{\varepsilon\rightarrow0}\lim_{n\rightarrow
\infty}\frac{1}{n}D_{c}^{\varepsilon}(\rho^{\otimes n},\sigma^{\otimes n}),
\end{align}
respectively, we conclude from the quantum Stein's lemma \cite{HP91,ON00}\ and
the asymptotic equipartition property for the smooth max-relative entropy
\cite{TCR09} that%
\begin{equation}
D_{d}(\rho,\sigma)=D_{c}(\rho,\sigma)=D(\rho\Vert\sigma
),\label{eq:distill-cost-equal-rel-ent}%
\end{equation}
thus demonstrating the fundamental operational interpretation of the quantum
relative entropy in the resource theory of asymmetric distinguishability. It
is worthwhile to note that we can conclude the stronger statement%
\begin{align}
D_{d}^{\varepsilon}(\rho^{\otimes n},\sigma^{\otimes n}) &  =nD(\rho
\Vert\sigma)+O(\sqrt{n}), \label{eq:approx-distill-expand}\\
  D_{c}^{\varepsilon}(\rho^{\otimes n},\sigma^{\otimes n}) & =nD(\rho
\Vert\sigma)+O(\sqrt{n}),
\label{eq:approx-cost-expand}
\end{align}
from \cite{tomamichel2012framework,tomamichel2013hierarchy,li14} (see
Appendix~\ref{app:asymp-dist-cost}). Thus, the equality of approximate distillable distinguishability and approximate distinguishability cost in the i.i.d.~case holds in the leading order term, with a difference in sublinear in $n$ terms. As discussed in Appendix~\ref{app:infidelity}, the second-order term in \eqref{eq:approx-distill-expand} can be identified exactly by appealing to \cite{li14,tomamichel2013hierarchy}. The second-order term in \eqref{eq:approx-cost-expand} can be identified also by appealing to \cite{tomamichel2013hierarchy}, but there is a need in this case to change the quantification of error in the resource theory of asymmetric distinguishability from normalized trace distance to infidelity. 

As a consequence of the fundamental equality in
\eqref{eq:distill-cost-equal-rel-ent}, we conclude that the resource theory of
asymmetric distinguishability is reversible in the asymptotic setting. That
is, for large $n$, by starting with the box $(\rho^{\otimes n},\sigma^{\otimes
n})$ one can distill it approximately to $nD(\rho\Vert\sigma)$ bits of
asymmetric distinguishability, and then one can dilute these $nD(\rho
\Vert\sigma)$ bits of asymmetric distinguishability back to the box
$(\rho^{\otimes n},\sigma^{\otimes n})$ approximately.

\subsection{Asymptotic box transformations}

We can also solve the asymptotic box transformation problem stated around
\eqref{eq:approx-box-transform-fundamental}. Before doing so, let us formalize
the problem. Let $n,m\in\mathbb{Z}^{+}$ and $\varepsilon\in\left[  0,1\right]
$. An $(n,m,\varepsilon)$ box transformation protocol for the boxes
$(\rho,\sigma)$ and $(\tau,\omega)$ consists of a channel $\mathcal{N}^{(n)}$
such that
\begin{align}
\mathcal{N}^{(n)}(\rho^{\otimes n}) & \approx_{\varepsilon}%
\tau^{\otimes m},\\
\mathcal{N}^{(n)}(\sigma^{\otimes n}) & =\omega^{\otimes
m}.
\end{align}
A rate $R$ is \textit{achievable} if for all $\varepsilon\in(0,1]$,
$\delta>0$, and sufficiently large $n$, there exists an $(n,n[R-\delta
],\varepsilon)$ box transformation protocol. The optimal box transformation
rate $R((\rho,\sigma)\rightarrow(\tau,\omega))$ is then equal to the supremum
of all achievable rates.

On the other hand, a rate $R$ is a \textit{strong converse rate} if for all
$\varepsilon\in\lbrack0,1)$, $\delta>0$, and sufficiently large $n$, there
does not exist an $(n,n[R+\delta],\varepsilon)$ box transformation protocol.
The strong converse box transformation rate $\widetilde{R}((\rho
,\sigma)\rightarrow(\tau,\omega))$ is then equal to the infimum of all strong
converse rates.

Note that the following inequality is a consequence of the definitions:%
\begin{equation}
R((\rho,\sigma)\rightarrow(\tau,\omega))\leq\widetilde{R}((\rho,\sigma
)\rightarrow(\tau,\omega)).
\end{equation}

The final result of our paper is the following fundamental equality for the
resource theory of asymmetric distinguishability:%
\begin{equation}
R((\rho,\sigma)\rightarrow(\tau,\omega))=\widetilde{R}((\rho,\sigma
)\rightarrow(\tau,\omega))=\frac{D(\rho\Vert\sigma)}{D(\tau\Vert\omega
)},\label{eq:box-trans-main-result}%
\end{equation}
indicating that the quantum relative entropy plays a central role as the
optimal conversion rate between boxes. 

We should clarify \eqref{eq:box-trans-main-result} a bit further. It holds
whenever $\operatorname{supp}(\rho)\subseteq\operatorname{supp}(\sigma)$ and
$\operatorname{supp}(\tau)\subseteq\operatorname{supp}(\omega)$. If
$\operatorname{supp}(\rho)\subseteq\operatorname{supp}(\sigma)$ but
$\operatorname{supp}(\tau)\not \subseteq \operatorname{supp}(\omega)$, then
$\frac{D(\rho\Vert\sigma)}{D(\tau\Vert\omega)} = 0$ and it is not possible to
perform the transformation at a non-negligible rate. If $\operatorname{supp}%
(\rho)\not \subseteq \operatorname{supp}(\sigma)$ but $\operatorname{supp}%
(\tau)\subseteq\operatorname{supp}(\omega)$, then $\frac{D(\rho\Vert\sigma
)}{D(\tau\Vert\omega)} = \infty$ and it is possible to produce as many copies
of $\tau$ and $\omega$ as desired.

The proof of this result consists of two parts:\ achievability and optimality.
For the achievability part, i.e., the bound
\begin{equation}
R((\rho,\sigma)\rightarrow(\tau,\omega)) \geq 
\frac{D(\rho\Vert\sigma)}{D(\tau\Vert\omega
)},
\end{equation}
we first distill bits of asymmetric
distinguishability from $(\rho^{\otimes n},\sigma^{\otimes n})$ at the rate
$\approx D(\rho\Vert\sigma)$. After doing so, we then dilute these $\approx
nD(\rho\Vert\sigma)$ bits of asymmetric distinguishability to the box
$(\tau^{\otimes m},\omega^{\otimes m})$, such that $m\approx n\left[
D(\rho\Vert\sigma)/D(\tau\Vert\omega)\right]  $, establishing that
$R((\rho,\sigma)\rightarrow(\tau,\omega))\geq\frac{D(\rho\Vert\sigma)}%
{D(\tau\Vert\omega)}$. For the optimality part, i.e., the strong converse bound
\begin{equation}
\widetilde{R}((\rho,\sigma)\rightarrow(\tau,\omega)) \leq 
\frac{D(\rho\Vert\sigma)}{D(\tau\Vert\omega
)},
\end{equation}
we suppose that there exists a
sequence of $(n,m,\varepsilon)$ box transformation protocols and then employ a
pseudo-continuity inequality for sandwiched R\'{e}nyi relative entropy
(Lemma~\ref{lem:app-continuity-sandwiched-Renyi}) and its data processing
inequality to conclude that $\widetilde{R}((\rho,\sigma)\rightarrow
(\tau,\omega))\leq\frac{D(\rho\Vert\sigma)}{D(\tau\Vert\omega)}$.
Alternatively, we can employ a pseudo-continuity inequality for the
Petz--R\'{e}nyi relative entropy (Lemma~\ref{lem:petz-renyi-continuity}) and
its data processing inequality. See Appendix~\ref{app:asymp-box-tr} for details. We note here that the bounds in Propositions~\ref{prop:strong-converse-exp} and \ref{prop:strong-converse-exp-Petz-Renyi} are \textit{exponential strong converse bounds}, demonstrating that the error in the transformation converges to one exponentially fast if the rate of conversion is strictly larger than $\frac{D(\rho\Vert\sigma)}{D(\tau\Vert\omega)}$.

\section{Conclusion}

In this paper, we have developed the resource theory of asymmetric
distinguishability. The main constituents consist of boxes as the objects of
manipulation, all quantum channels as the free operations, and bits of
asymmetric distinguishability as the fundamental currency of interconversion.
The resource theory is reversible in the asymptotic case, and the quantum
relative entropy emerges as the fundamental rate at which boxes can be
converted. Our one-shot results can be compactly summarized as follows:

\begin{enumerate}
\item \textit{The min-relative entropy is equal to the exact one-shot
distillable distinguishability.}

\item \textit{The max-relative entropy is equal to the exact one-shot
distinguishability cost.}

\item \textit{The smooth min-relative entropy is equal to the approximate
one-shot distillable distinguishability.}

\item \textit{The smooth max-relative entropy is equal to the approximate
one-shot distinguishability cost.}
\end{enumerate}

Thus, each of these one-shot entropies are fundamentally operational
quantities. Finally, the ratio of quantum relative entropies of two pairs of
quantum states is equal to the optimal rate of asymptotic box transformations
between them.

Going forward from here, there are many interesting directions to pursue. The
resource theory of asymmetric distinguishability for quantum channels has recently been developed in \cite{WW19}. The main constituents consist of a
channel box $(\mathcal{N},\mathcal{M})$, for quantum channels $\mathcal{N}$
and $\mathcal{M}$, as the basic objects of manipulation, superchannels
\cite{CDP08}\ as the free operations, and bits of asymmetric
distinguishability as the fundamental currency. Some basic results are that the one-shot distillable
distinguishability of a channel box is equal to the smooth channel 
min-relative entropy \cite{Cooney2016}, and the one-shot distinguishability
cost is equal to the smooth channel max-relative entropy \cite{GFWRSCW18,LW19}%
. The theory reduces to the theory for quantum states in the case that the
channels that are environment-seizable, as defined in \cite{BHKW18}.

It remains open to determine optimal error exponents and strong converse
exponents for the distinguishability dilution task, as well as for the more
general box transformation problem. These quantities have been established for
distinguishability distillation (i.e., hypothesis testing)
\cite{N06,Hay07,ANSV08,HMO08,MO13}, and so there is a strong possibility that
these operational quantities could be determined for the dual task. Some of
the bounds in Appendix~\ref{sec:app-dmaxeps-bnds} could be useful for this
purpose. The same questions remain open for second-order asymptotics.

In Appendix~\ref{app:infidelity}, we explore a variation of the resource
theory of asymmetric distinguishability in which the infidelity is employed as
a measure of approximation, rather than the normalized trace distance. There
are similar interesting questions regarding this variation, in particular,
whether error exponents and strong converse exponents for distinguishability
dilution could be proven to be optimal.

One could also consider the case in which the boxes consist of not just two
states but multiple states, connecting with the theory of quantum state
discrimination \cite{BC09,BK15SD}. The boxes could even consist of a continuum
of states or channels, connecting with quantum estimation theory
\cite{H69,Hel76}\ and the resource theoretic approach put forward in
\cite{Mat05}. The boxes could also consist of a state and a set of states,
with the set of free operations restricted, which allows for connecting with
general resource theories \cite{CG18,LBT19}. Extending this, the boxes could
consist of a channel and set of channels, with restricted free operations,
allowing to connect with general resource theories of quantum channels
\cite{LY19,LW19}.

A particularly interesting direction would be to consider reversibility of the
resource theory of asymmetric distinguishability beyond the first order and
investigate resource resonance effects. For this direction, the recent results
of \cite{KH17,Chubb2018beyondthermodynamic,KCT19,CTK19}\ are quite relevant. Related to this, one could investigate more fine-grained questions related to asymptotic reversibility along the lines of \cite{KH13}, where we expect similar findings  to hold in the resource theory of asymmetric distinguishability.

\textit{Note}: After completing our paper, we learned about the independent and related work of \cite{ST19}.

\begin{acknowledgments}
We are grateful to Francesco Buscemi, Nilanjana Datta, Sumeet Khatri, and
Marco Tomamichel for discussions. This work was ultimately inspired by the
talk of Robin Blume-Kohout at the 2017 APS\ March Meeting in New Orleans,
Louisiana \cite{BK17}. XW acknowledges support from the Department of Defense,
and MMW\ acknowledges support from the National Science Foundation under Grant No.~1907615.
\end{acknowledgments}

\bibliographystyle{alpha}
\bibliography{Ref}


\vspace{.5in}

\appendix

In the following appendices, we provide detailed proofs of all claims in the
main text. As a resource, we have included derivations of some of the dual semi-definite programs listed below as an ancillary file available for download with the arXiv posting of this paper. We begin by providing some background facts in
Appendix~\ref{sec:app-background}, some of which can be found in
\cite{W17book}.

\section{Background\label{sec:app-background}}

\subsection{Normalized trace distance}

A quantum state is described mathematically by a positive semi-definite
operator with trace equal to one. The normalized trace distance between two
quantum states $\rho$ and $\sigma$ is given by $\frac{1}{2}\left\Vert
\rho-\sigma\right\Vert _{1}$, where the trace norm of an operator $A$ is
defined as $\left\Vert A\right\Vert _{1}=\operatorname{Tr}[\sqrt{A^{\dag}A}]$.
The following variational characterization of the normalized trace distance is
well known \cite{H69}:%
\begin{equation}
\frac{1}{2}\left\Vert \rho-\sigma\right\Vert _{1}=\sup_{\Lambda\geq 0}\left\{
\operatorname{Tr}[\Lambda(\rho-\sigma)]:\Lambda\leq I\right\}  ,
\label{eq:tr-dist-sup}%
\end{equation}
endowing the normalized trace distance with its operational meaning as the largest
probability difference that a single POVM\ element can assign to two quantum
states. The right-hand side\ of \eqref{eq:tr-dist-sup} is a semi-definite
program as written, with the following dual:%
\begin{equation}
\inf_{Y\geq 0}\left\{  \operatorname{Tr}[Y]:Y\geq\rho-\sigma\right\}
=\frac{1}{2}\left\Vert \rho-\sigma\right\Vert _{1}, \label{eq:TD-dual}%
\end{equation}
where the equality holds from strong duality.

\subsection{Choi isomorphism}

The Choi isomorphism is a standard way of characterizing quantum channels that is
suitable for optimizing over them in semi-definite programs. For a quantum
channel $\mathcal{N}_{A\rightarrow B}$, its Choi operator is given by%
\begin{equation}
J_{RB}^{\mathcal{N}}:=\mathcal{N}_{A\rightarrow B}(\Gamma_{RA}),
\end{equation}
where $\Gamma_{RA}=|\Gamma\rangle\langle\Gamma|_{RA}$ and%
\begin{equation}
|\Gamma\rangle_{RA}:=\sum_{i}|i\rangle_{R}|i\rangle_{A},
\end{equation}
with $\{|i\rangle_{R}\}_{i}$ and $\{|i\rangle_{A}\}_i$ orthonormal bases. The
Choi operator is positive semi-definite $J_{RB}^{\mathcal{N}}\geq0$,
corresponding to $\mathcal{N}_{A\rightarrow B}$ being completely positive, and
satisfies $\operatorname{Tr}_{B}[J_{RB}^{\mathcal{N}}]=I_{R}$, the latter
corresponding to $\mathcal{N}_{A\rightarrow B}$ being trace preserving.

On the other hand, given an operator $J_{RB}^{\mathcal{M}}$ satisfying
$J_{RB}^{\mathcal{M}}\geq0$ and $\operatorname{Tr}_{B}[J_{RB}^{\mathcal{M}%
}]=I_{R}$, one realizes via postselected teleportation \cite{B05}\ the
following quantum channel:%
\begin{align}
\mathcal{M}_{A\rightarrow B}(\rho_{A})  &  =\langle\Gamma|_{SR}\left(
\rho_{S}\otimes J_{RB}^{\mathcal{M}}\right)  |\Gamma\rangle_{SR}\\
&  =\operatorname{Tr}_{R}[T_{R}(\rho_{R})J_{RB}^{\mathcal{M}}],
\end{align}
where systems $S$, $R$, and $A$ are isomorphic and the last line employs the
facts that $\left(  M_{S}\otimes I_{R}\right)  |\Gamma\rangle_{SR}=\left(
I_{S}\otimes T_{R}(M_{R})\right)  |\Gamma\rangle_{SR}$ for $T_{R}$ the
transpose map, defined as%
\begin{equation}
T_{R}(\rho_{R})=\sum_{i,j}|i\rangle\langle j|_{R}\rho_{R}|i\rangle\langle
j|_{R},
\end{equation}
and $\langle\Gamma|_{SR}\left(  I_{S}\otimes X_{RB}\right)  |\Gamma
\rangle_{SR}=\operatorname{Tr}_{R}[X_{RB}]$. We often abbreviate the transpose
map simply as%
\begin{equation}
\rho_{R}^{T}=T_{R}(\rho_{R}).
\end{equation}
Since the constraints $J_{RB}^{\mathcal{M}}\geq0$ and $\operatorname{Tr}%
_{B}[J_{RB}^{\mathcal{M}}]=I_{R}$ are semi-definite, this is a useful way of
incorporating optimizations over quantum channels into semi-definite programs.

\subsection{Relative entropies and data processing}

\label{app:rel-ents-DP}The Petz--R\'{e}nyi relative entropy is defined for
states $\rho$ and $\sigma$ as \cite{P86}%
\begin{align}
D_{\alpha}(\rho\Vert\sigma)  &  :=\frac{1}{\alpha-1}\log_{2}\operatorname{Tr}%
[\rho^{\alpha}\sigma^{1-\alpha}]\\
&  =\frac{2}{\alpha-1}\log_{2}\left\Vert \rho^{\alpha/2}\sigma^{(1-\alpha
)/2}\right\Vert _{2},
\end{align}
if $\alpha\in(0,1)$ or $\alpha\in(1,\infty)$ and $\operatorname{supp}%
(\rho)\subseteq\operatorname{supp}(\sigma)$. If $\alpha\in(1,\infty)$ and
$\operatorname{supp}(\rho)\not \subseteq \operatorname{supp}(\sigma)$, then
$D_{\alpha}(\rho\Vert\sigma)=\infty$ \cite{TCR09}. The Petz--R\'{e}nyi
relative entropy obeys the following data processing inequality
\cite{P86,TCR09,T15} for a quantum channel $\mathcal{N}$ and $\alpha
\in(0,1)\cup(1,2]$:%
\begin{equation}
D_{\alpha}(\rho\Vert\sigma)\geq D_{\alpha}(\mathcal{N}(\rho)\Vert
\mathcal{N}(\sigma)). \label{eq:DP-Petz-Renyi}%
\end{equation}
The following limits hold
\begin{align}
\lim_{\alpha\rightarrow1}D_{\alpha}(\rho\Vert\sigma)  &  =D(\rho\Vert
\sigma),\label{eq:conv-PR-q-rel-ent}\\
\lim_{\alpha\rightarrow0}D_{\alpha}(\rho\Vert\sigma)  &  =D_{\min}(\rho
\Vert\sigma),
\end{align}
where $D(\rho\Vert\sigma)$ is the quantum relative entropy defined
in~\eqref{eq:q-rel-ent} and $D_{\min}(\rho\Vert\sigma)$ is defined in
\eqref{eq:order-zero-Renyi-rel-ent}. The Petz--R\'{e}nyi relative entropies
are ordered in the following sense \cite{TCR09,T15}:%
\begin{equation}
D_{\alpha}(\rho\Vert\sigma)\geq D_{\beta}(\rho\Vert\sigma),
\label{eq:app-Petz-Renyi-ordered}%
\end{equation}
for $\alpha\geq\beta>0$.

The sandwiched R\'enyi relative entropy is defined for states $\rho$ and
$\sigma$ as \cite{muller2013quantum,WWY14}%
\begin{align}
\widetilde{D}_{\alpha}(\rho\Vert\sigma)  &  :=\frac{1}{\alpha-1}\log
_{2}\operatorname{Tr}[(\sigma^{(1-\alpha)/2\alpha}\rho\sigma^{(1-\alpha
)/2\alpha})^{\alpha}]\nonumber\\
&  =\frac{\alpha}{\alpha-1}\log_{2}\left\Vert \sigma^{(1-\alpha)/2\alpha}%
\rho\sigma^{(1-\alpha)/2\alpha}\right\Vert _{\alpha}\nonumber\\
&  =\frac{2\alpha}{\alpha-1}\log_{2}\left\Vert \sigma^{(1-\alpha)/2\alpha}%
\rho^{1/2}\right\Vert _{2\alpha},
\end{align}
if $\alpha\in(0,1)$ or $\alpha\in(1,\infty)$ and $\operatorname{supp}%
(\rho)\subseteq\operatorname{supp}(\sigma)$. If $\alpha\in(1,\infty)$ and
$\operatorname{supp}(\rho)\not \subseteq \operatorname{supp}(\sigma)$, then
$\widetilde{D}_{\alpha}(\rho\Vert\sigma)=\infty$. The sandwiched R\'enyi
relative entropy obeys the following data processing inequality \cite{FL13}%
\ for a quantum channel $\mathcal{N}$ and $\alpha\in\lbrack1/2,1)\cup
(1,\infty)$:%
\begin{equation}
\widetilde{D}_{\alpha}(\rho\Vert\sigma)\geq\widetilde{D}_{\alpha}%
(\mathcal{N}(\rho)\Vert\mathcal{N}(\sigma)).
\label{eq:app-qdp-sandwiched-Renyi}%
\end{equation}
(See~\cite{W17f} for an alternative proof of
\eqref{eq:app-qdp-sandwiched-Renyi}.) The following limits hold%
\begin{align}
\lim_{\alpha\rightarrow1}\widetilde{D}_{\alpha}(\rho\Vert\sigma)  &
=D(\rho\Vert\sigma),\label{eq:app-sandwiched-Renyi-a-1-rel-ent}\\
\lim_{\alpha\rightarrow\infty}\widetilde{D}_{\alpha}(\rho\Vert\sigma)  &
=D_{\max}(\rho\Vert\sigma),\label{eq:app-dmax-limit}\\
\lim_{\alpha\rightarrow1/2}\widetilde{D}_{\alpha}(\rho\Vert\sigma)  &  =-\log
F(\rho,\sigma),
\end{align}
where%
\begin{equation}
F(\rho,\sigma):=\left\Vert \sqrt{\rho}\sqrt{\sigma}\right\Vert _{1}^{2}%
\end{equation}
is the quantum fidelity \cite{U76}. The sandwiched R\'enyi relative entropies
are ordered in the following sense \cite{muller2013quantum}:%
\begin{equation}
\widetilde{D}_{\alpha}(\rho\Vert\sigma)\geq\widetilde{D}_{\beta}(\rho
\Vert\sigma), \label{eq:app-sandwiched-Renyi-ordered}%
\end{equation}
for $\alpha\geq\beta>0$.

Note that the following inequality holds
\begin{equation}
D_{\min}(\rho\Vert\sigma) \leq\widetilde{D}_{1/2}(\rho\Vert\sigma),
\end{equation}
as a consequence of the equality \cite{FC95}
\begin{equation}
F(\rho,\sigma) = \left(  \inf_{\{\Lambda_{x}\}_{x}} \sum_{x} \sqrt
{\operatorname{Tr}[\Lambda_{x} \rho]\operatorname{Tr}[\Lambda_{x}\sigma
]}\right)  ^{2},
\end{equation}
where the optimization is with respect to POVMs $\{\Lambda_{x}\}_{x}$, and by
choosing this POVM suboptimally as $\left\{  \Pi_{\rho}, I - \Pi_{\rho
}\right\}  $.

The min-relative entropy obeys the data processing inequality for states
$\rho$ and $\sigma$ and a quantum channel $\mathcal{N}$:%
\begin{equation}
D_{\min}(\rho\Vert\sigma)\geq D_{\min}(\mathcal{N}(\rho)\Vert\mathcal{N}%
(\sigma)). \label{eq:app-DP-PR-order-0}%
\end{equation}
This inequality was proved in \cite{D09} by utilizing its relation to the
Petz--R\'{e}nyi relative entropies. For an alternative proof, first note that
the inequality in \eqref{eq:app-DP-PR-order-0} is equivalent to%
\begin{equation}
\operatorname{Tr}[\Pi_{\rho}\sigma]\leq\operatorname{Tr}[\Pi_{\mathcal{N}%
(\rho)}\mathcal{N}(\sigma)].
\end{equation}
To see the latter, let $U$ be an isometric extension of the channel
$\mathcal{N}$, so that
\begin{equation}
\mathcal{N}_{A\rightarrow B}(\omega_{A})=\operatorname{Tr}_{E}[U_{A\rightarrow
BE}\omega_{A}(U_{A\rightarrow BE})^{\dag}].
\end{equation}
Then we find that%
\begin{align}
\operatorname{Tr}[\Pi_{\rho}\sigma]  &  =\operatorname{Tr}[U\Pi_{\rho}U^{\dag
}U\sigma U^{\dag}]\\
&  =\operatorname{Tr}[\Pi_{U\rho U^{\dag}}U\sigma U^{\dag}]\\
&  \leq\operatorname{Tr}[\left(  \Pi_{\mathcal{N}(\rho)}\otimes I_{E}\right)
U\sigma U^{\dag}]\\
&  =\operatorname{Tr}[\Pi_{\mathcal{N}(\rho)}\mathcal{N}(\sigma)].
\end{align}
The first equality follows because $U\Pi_{\rho}U^{\dag}=\Pi_{U\rho U^{\dag}}$.
The inequality follows because the support of $U\rho U^{\dag}$ is contained in
the support of $\mathcal{N}(\rho)\otimes I_{E}$ \cite[Appendix~B]{Renner2005}.

The smooth min-relative entropy obeys the data processing inequality as well,
in fact for any trace non-increasing positive map $\mathcal{N}$ and for all
$\varepsilon\in\left(  0,1\right)  $:%
\begin{equation}
D_{\min}^{\varepsilon}(\rho\Vert\sigma)\geq D_{\min}^{\varepsilon}%
(\mathcal{N}(\rho)\Vert\mathcal{N}(\sigma)). \label{eq:hypo-DP}%
\end{equation}
This follows from the definition. Let $\Lambda$ be an arbitrary operator such
that $\operatorname{Tr}[\Lambda\mathcal{N}(\rho)]\geq1-\varepsilon$ and
$0\leq\Lambda\leq I$. Then it follows that $\operatorname{Tr}[\mathcal{N}%
^{\dag}(\Lambda)\rho]=\operatorname{Tr}[\Lambda\mathcal{N}(\rho)]\geq
1-\varepsilon$ and $0\leq\mathcal{N}^{\dag}(\Lambda)\leq\mathcal{N}^{\dag
}(I)\leq I$, the latter inequalities following because $\mathcal{N}^{\dag}$ is
a positive map if $\mathcal{N}$ is and $\mathcal{N}^{\dag}$ is subunital if
$\mathcal{N}$ is trace non-increasing. So then $\mathcal{N}^{\dag}(\Lambda)$
is a candidate for $D_{\min}^{\varepsilon}(\rho\Vert\sigma)$ and thus
$D_{\min}^{\varepsilon}(\rho\Vert\sigma)\geq-\log\operatorname{Tr}%
[\mathcal{N}^{\dag}(\Lambda)\sigma]=-\log\operatorname{Tr}[\Lambda
\mathcal{N}(\sigma)]$. Since the argument holds for an arbitrary $\Lambda$
satisfying $\operatorname{Tr}[\Lambda\mathcal{N}(\rho)]\geq1-\varepsilon$ and
$0\leq\Lambda\leq I$, we conclude \eqref{eq:hypo-DP}.

The max-relative entropy also obeys the data processing inequality for an
arbitrary positive map $\mathcal{N}$:%
\begin{equation}
D_{\max}(\rho\Vert\sigma)\geq D_{\max}(\mathcal{N}(\rho)\Vert\mathcal{N}%
(\sigma)). \label{eq:dmax-DP}%
\end{equation}
To see this, let $\lambda$ be such that $\rho\leq2^{\lambda}\sigma$. Then from
the fact that $\mathcal{N}$ is positive, it follows that $\mathcal{N}%
(\rho)\leq2^{\lambda}\mathcal{N}(\sigma)$. It then follows that%
\begin{align}
\lambda &  \geq\inf\left\{  \mu:\mathcal{N}(\rho)\leq2^{\mu}\mathcal{N}%
(\sigma)\right\} \\
&  =D_{\max}(\mathcal{N}(\rho)\Vert\mathcal{N}(\sigma)).
\end{align}
Since this is true for arbitrary $\lambda$ satisfying $\rho\leq2^{\lambda
}\sigma$, we conclude \eqref{eq:dmax-DP}.

The smooth max-relative entropy obeys the data processing inequality for a
positive, trace-preserving map $\mathcal{N}$ and for all $\varepsilon
\in\left(  0,1\right)  $:%
\begin{equation}
D_{\max}^{\varepsilon}(\rho\Vert\sigma)\geq D_{\max}^{\varepsilon}%
(\mathcal{N}(\rho)\Vert\mathcal{N}(\sigma)). \label{eq:dp-smooth-d-max}%
\end{equation}
To see this, let $\widetilde{\rho}$ be an arbitrary state such that%
\begin{equation}
\frac{1}{2}\left\Vert \widetilde{\rho}-\rho\right\Vert _{1}\leq\varepsilon.
\label{eq:eps-close-d-max}%
\end{equation}
Then from the data processing inequality for normalized trace distance under
positive trace-preserving maps, it follows that%
\begin{equation}
\frac{1}{2}\left\Vert \mathcal{N}(\widetilde{\rho})-\mathcal{N}(\rho
)\right\Vert _{1}\leq\varepsilon.
\end{equation}
So it follows that%
\begin{align}
D_{\max}(\widetilde{\rho}\Vert\sigma)  &  \geq D_{\max}(\mathcal{N}%
(\widetilde{\rho})\Vert\mathcal{N}(\sigma))\\
&  \geq D_{\max}^{\varepsilon}(\mathcal{N}(\rho)\Vert\mathcal{N}(\sigma)).
\end{align}
Since the inequality holds for an arbitrary state $\widetilde{\rho}$
satisfying \eqref{eq:eps-close-d-max}, we conclude \eqref{eq:dp-smooth-d-max}.

Since all of the above quantities obey the data processing inequality for
quantum channels, we conclude that they are invariant under the action of an
isometric channel $\mathcal{U}(\cdot)=U(\cdot)U^{\dag}$:%
\begin{align}
D_{\min}(\rho\Vert\sigma)  &  =D_{\min}(\mathcal{U}(\rho)\Vert\mathcal{U}%
(\sigma)),\\
D_{\min}^{\varepsilon}(\rho\Vert\sigma)  &  =D_{\min}^{\varepsilon
}(\mathcal{U}(\rho)\Vert\mathcal{U}(\sigma)),\\
D_{\max}(\rho\Vert\sigma)  &  =D_{\max}(\mathcal{U}(\rho)\Vert\mathcal{U}%
(\sigma)),\\
D_{\max}^{\varepsilon}(\rho\Vert\sigma)  &  =D_{\max}^{\varepsilon
}(\mathcal{U}(\rho)\Vert\mathcal{U}(\sigma)),
\end{align}
which follows because $\mathcal{U}$ is a channel and the channel in
\eqref{eq:invert-isometry} perfectly reverses the action of $\mathcal{U}$.

As stated in \eqref{eq:limit-smooth-dmin-to-dmin}, the following limit holds%
\begin{equation}
\lim_{\varepsilon\rightarrow0}D_{\min}^{\varepsilon}(\rho\Vert\sigma)=D_{\min
}(\rho\Vert\sigma). \label{eq:lim-dmineps-to-dmin}%
\end{equation}
In the main text, we provided an operational proof of this limit. An
alternative proof goes as follows. Consider that the following inequality
holds for all $\varepsilon\in(0,1)$:%
\begin{equation}
D_{\min}^{\varepsilon}(\rho\Vert\sigma)\geq D_{\min}(\rho\Vert\sigma),
\end{equation}
because the measurement operator $\Pi_{\rho}$ (projection onto support of
$\rho$) satisfies $\operatorname{Tr}[\Pi_{\rho}\rho]\geq1-\varepsilon$ for all
$\varepsilon\in(0,1)$. So we conclude that%
\begin{equation}
\liminf_{\varepsilon\rightarrow0}D_{\min}^{\varepsilon}(\rho\Vert\sigma)\geq
D_{\min}(\rho\Vert\sigma). \label{eq:liminf-dmineps-to-dmin}%
\end{equation}
Alternatively, suppose that $\Lambda$ is a measurement operator satisfying
$\operatorname{Tr}[\Lambda\rho]=1-\varepsilon$ (note that when optimizing
$D_{\min}^{\varepsilon}$, it suffices to optimize over measurement operators
satisfying the constraint $\operatorname{Tr}[\Lambda\rho]\geq1-\varepsilon$
with equality \cite{KW17}). Then applying the data processing inequality for
$D_{\alpha}(\rho\Vert\sigma)$ under the measurement $\left\{  \Lambda
,I-\Lambda\right\}  $, which holds for $\alpha\in(0,1)$, we find that%
\begin{multline}
D_{\alpha}(\rho\Vert\sigma)\geq\\
\frac{1}{\alpha-1}\log_{2}\left[  \left(  1-\varepsilon\right)  ^{\alpha
}\operatorname{Tr}[\Lambda\sigma]^{1-\alpha}+\varepsilon^{\alpha}\left(
1-\operatorname{Tr}[\Lambda\sigma]\right)  ^{1-\alpha}\right]  .
\end{multline}
Since this bound holds for all measurement operators $\Lambda$ satisfying
$\operatorname{Tr}[\Lambda\rho]=1-\varepsilon$, we conclude the following
bound for all $\alpha\in(0,1)$:%
\begin{multline}
D_{\alpha}(\rho\Vert\sigma)\geq\\
\frac{1}{\alpha-1}\log_{2}\left[
\begin{array}
[c]{c}%
\left(  1-\varepsilon\right)  ^{\alpha}\left(  2^{-D_{\min}^{\varepsilon}%
(\rho\Vert\sigma)}\right)  ^{1-\alpha}\\
+\ \varepsilon^{\alpha}\left(  1-2^{-D_{\min}^{\varepsilon}(\rho\Vert\sigma
)}\right)  ^{1-\alpha}%
\end{array}
\right]  .
\end{multline}
Now taking the limit of the right-hand side as $\varepsilon\rightarrow0$, we
find that the following bound holds for all $\alpha\in(0,1)$:%
\begin{equation}
D_{\alpha}(\rho\Vert\sigma)\geq\limsup_{\varepsilon\rightarrow0}D_{\min
}^{\varepsilon}(\rho\Vert\sigma).
\end{equation}
Since the bound holds for all $\alpha\in(0,1)$, we can take the limit on the
left-hand side to arrive at%
\begin{equation}
\lim_{\alpha\rightarrow0}D_{\alpha}(\rho\Vert\sigma)=D_{\min}(\rho\Vert
\sigma)\geq\limsup_{\varepsilon\rightarrow0}D_{\min}^{\varepsilon}(\rho
\Vert\sigma). \label{eq:limsup-dmineps-to-dmin}%
\end{equation}
Now putting together \eqref{eq:liminf-dmineps-to-dmin} and
\eqref{eq:limsup-dmineps-to-dmin}, we conclude~\eqref{eq:lim-dmineps-to-dmin}.

As stated in \eqref{eq:limit-smooth-dmax-to-dmax}, the following limit holds%
\begin{equation}
\lim_{\varepsilon\rightarrow0}D_{\max}^{\varepsilon}(\rho\Vert\sigma)=D_{\max
}(\rho\Vert\sigma). \label{eq:lim-dmaxeps-to-dmax}%
\end{equation}
In the main text, we provided an operational proof of this limit. An
alternative proof goes as follows. Consider that the following bound holds for
all $\varepsilon\in(0,1)$:%
\begin{equation}
D_{\max}^{\varepsilon}(\rho\Vert\sigma)\leq D_{\max}(\rho\Vert\sigma),
\end{equation}
which follows as a simple consequence of the fact that we can always set
$\widetilde{\rho}=\rho$. Then the following limit holds%
\begin{equation}
\limsup_{\varepsilon\rightarrow0}D_{\max}^{\varepsilon}(\rho\Vert\sigma)\leq
D_{\max}(\rho\Vert\sigma). \label{eq:limsup-dmaxeps-to-dmax}%
\end{equation}
To see the other inequality, let $\widetilde{\rho}$ be a state satisfying
$\frac{1}{2}\left\Vert \widetilde{\rho}-\rho\right\Vert _{1}\leq\varepsilon$.
Then this means that $\left\Vert \widetilde{\rho}-\rho\right\Vert _{\infty
}\leq2\varepsilon$. Consider that%
\begin{align}
&  D_{\max}(\widetilde{\rho}\Vert\sigma)\nonumber\\
&  =\log_{2}\!\left\Vert \sigma^{-1/2}\widetilde{\rho}\sigma^{-1/2}\right\Vert
_{\infty}\nonumber\\
&  \geq\log_{2}\!\left(  \left\Vert \sigma^{-1/2}\rho\sigma^{-1/2}\right\Vert
_{\infty}-\left\Vert \sigma^{-1/2}\left(  \widetilde{\rho}-\rho\right)
\sigma^{-1/2}\right\Vert _{\infty}\right) \nonumber\\
&  \geq\log_{2}\!\left(  \left\Vert \sigma^{-1/2}\rho\sigma^{-1/2}\right\Vert
_{\infty}-\left\Vert \sigma^{-1/2}\right\Vert _{\infty}^{2}\left\Vert
\widetilde{\rho}-\rho\right\Vert _{\infty}\right) \nonumber\\
&  \geq\log_{2}\!\left(  \left\Vert \sigma^{-1/2}\rho\sigma^{-1/2}\right\Vert
_{\infty}-2\varepsilon\left\Vert \sigma^{-1}\right\Vert _{\infty}\right)  .
\end{align}
Since this bound holds for all $\widetilde{\rho}$ satisfying $\frac{1}%
{2}\left\Vert \widetilde{\rho}-\rho\right\Vert _{1}\leq\varepsilon$, we
conclude that%
\begin{equation}
D_{\max}^{\varepsilon}(\rho\Vert\sigma)\geq\log_{2}\left(  \left\Vert
\sigma^{-1/2}\rho\sigma^{-1/2}\right\Vert _{\infty}-2\varepsilon\left\Vert
\sigma^{-1}\right\Vert _{\infty}\right)  .
\end{equation}
Then taking the limit $\varepsilon\rightarrow0$, we find that%
\begin{align}
\liminf_{\varepsilon\rightarrow0}D_{\max}^{\varepsilon}(\rho\Vert\sigma)  &
\geq\log_{2}\left\Vert \sigma^{-1/2}\rho\sigma^{-1/2}\right\Vert _{\infty
}\nonumber\\
&  =D_{\max}(\rho\Vert\sigma). \label{eq:liminf-dmaxeps-to-dmax}%
\end{align}
Putting together \eqref{eq:limsup-dmaxeps-to-dmax} and
\eqref{eq:liminf-dmaxeps-to-dmax}, we conclude \eqref{eq:lim-dmaxeps-to-dmax}.

\section{SDPs for smooth min- and max-relative entropies}

\label{app:SDPs-smooth-min-max}

Here we show that the smooth min- and max-relative entropies are characterized
by semi-definite programs. We also give the dual programs for convenience.  

Consider that%
\begin{equation}
D_{\min}^{\varepsilon}(\rho\Vert\sigma)=-\log_{2}\inf_{\Lambda\geq 0}\left\{
\begin{array}
[c]{c}%
\operatorname{Tr}[\Lambda\sigma]:
\Lambda\leq I,\\
\quad\operatorname{Tr}[\Lambda\rho]\geq1-\varepsilon
\end{array}
\right\}  ,
\end{equation}
which is an SDP as written. The dual SDP is given by
\begin{equation}
-\log_{2}\sup_{\mu,X\geq 0}\left\{
\begin{array}
[c]{c}%
\mu\left(  1-\varepsilon\right)  -\operatorname{Tr}[X]:\\
\mu\rho\leq\sigma+X
\end{array}
\right\}  ,
\end{equation}
and is equal to $D_{\min}^{\varepsilon}(\rho\Vert\sigma)$ by strong duality.
See \cite{DKFRR12} in this context.

By employing the definition of the smooth max-relative entropy in
\eqref{eq:smooth-max-rel-ent} and the dual characterization of the normalized
trace distance in \eqref{eq:TD-dual}, we find that%
\begin{equation}
D_{\max}^{\varepsilon}(\rho\Vert\sigma)=\log\inf\left\{
\begin{array}
[c]{c}%
\lambda:\\
\widetilde{\rho}\leq\lambda\sigma,\ \operatorname{Tr}[Y]\leq\varepsilon,\\
\operatorname{Tr}[\widetilde{\rho}]=1,\ Y\geq\rho-\widetilde{\rho},\\
\widetilde{\rho},\ Y\geq0
\end{array}
\right\}  .
\end{equation}
The dual SDP\ is given by%
\begin{equation}
\log\sup\left\{
\begin{array}
[c]{c}%
\operatorname{Tr}[Q\rho]+\mu-\varepsilon t:\\
\operatorname{Tr}[X\sigma]\leq1,\ Q\leq tI,\\
Q+\mu I\leq X,\\
X,Q,t\geq0,\ \mu\in\mathbb{R}%
\end{array}
\right\}  ,
\end{equation}
and is equal to $D_{\max}^{\varepsilon}(\rho\Vert\sigma)$ by strong duality.

\section{Approximate box transformation is an SDP}

\label{app:approx-box-tr-SDP}

We prove that the approximate box transformation
problem can be computed by a semi-definite program. First, recall that the
problem is characterized by%
\begin{multline}
\varepsilon((\rho,\sigma)\rightarrow(\tau,\omega)):=\\
\inf_{\mathcal{N}\in\text{CPTP}}\left\{  \varepsilon\in\left[  0,1\right]
:\mathcal{N}(\rho)\approx_{\varepsilon}\tau,\ \mathcal{N}(\sigma
)=\omega\right\}  ,
\end{multline}
for states $\rho$, $\sigma$, $\tau$, and $\omega$. By employing the dual form
of the trace distance from \eqref{eq:TD-dual}, we find that%
\begin{multline}
\varepsilon((\rho,\sigma)\rightarrow(\tau,\omega))=\\
\inf_{Y_{B},J_{RB}^{\mathcal{N}}}\left\{
\begin{array}
[c]{c}%
\operatorname{Tr}[Y_{B}]:\\
Y_{B}\geq\tau_{B}-\operatorname{Tr}_{R}[\rho_{R}^{T}J_{RB}^{\mathcal{N}}],\\
\operatorname{Tr}_{R}[\sigma_{R}^{T}J_{RB}^{\mathcal{N}}]=\omega_{B},\\
\operatorname{Tr}_{B}[J_{RB}^{\mathcal{N}}]=I_{R},\ \ Y_{B},J_{RB}%
^{\mathcal{N}}\geq0
\end{array}
\right\}  .
\end{multline}
The dual program is given by%
\begin{multline}
\varepsilon((\rho,\sigma)\rightarrow(\tau,\omega))=\\
\sup_{X_{B},W_{B},Z_{R}}\left\{
\begin{array}
[c]{c}%
\operatorname{Tr}[\tau_{B}X_{B}]+\operatorname{Tr}[\omega_{B}W_{B}%
]+\operatorname{Tr}[Z_{R}]:\\
X_{B}\leq I_{B},\\
\rho_{R}^{T}\otimes X_{B}+\sigma_{R}^{T}\otimes W_{B}+Z_{R}\otimes I_{B}%
\leq0,\\
X_{B}\geq0,\ W_{B},Z_{R}\in\text{Herm}%
\end{array}
\right\}  ,
\end{multline}
with the equality holding from strong duality.

\section{Impossibility of distinguishability increasing transformations}

\label{app:impossibility-result}It is impossible for a quantum channel
$\mathcal{N}$\ to increase the distinguishability of a box $(\rho,\sigma)$.
That is, it impossible for the transformation $(\rho,\sigma)\ \underrightarrow
{\mathcal{N}}\ (\mathcal{N}(\rho),\mathcal{N}(\sigma))$ to be such that the
distinguishability of $(\mathcal{N}(\rho),\mathcal{N}(\sigma))$ is strictly
larger than the distinguishability of $(\rho,\sigma)$. This follows as a
direct consequence of the data processing inequality for quantum relative
entropy \cite{Lin75}:%
\begin{equation}
D(\rho\Vert\sigma)\geq D(\mathcal{N}(\rho)\Vert\mathcal{N}(\sigma
)),\label{eq:DP-QRE}%
\end{equation}
when using quantum relative entropy as a quantifier of distinguishability.

For the specific transformation in \eqref{eq:outlawed-trans}, we find that%
\begin{align}
m  &  =D(|0\rangle\langle0|^{\otimes m}\Vert\pi^{\otimes m}),\\
n  &  =D(|0\rangle\langle0|^{\otimes n}\Vert\pi^{\otimes n}),
\end{align}
so that if the transformation in \eqref{eq:outlawed-trans} existed, it would
violate \eqref{eq:DP-QRE}, due to the assumption $n>m$.

The fact that the transformation in \eqref{eq:outlawed-trans} does not exist
can also be seen as a consequence of the linearity of quantum channels. Let us
first suppose that the boxes $(|0\rangle\langle0|^{\otimes m},\pi^{\otimes
m})$ and $(|0\rangle\langle0|^{\otimes n},\pi^{\otimes n})$ have been
reversibly transformed to their standard form as%
\begin{align}
&  (  |0\rangle\langle0|,\pi_{2^{m}})  ,\\
&  (  |0\rangle\langle0|,\pi_{2^{n}})  ,
\end{align}
respectively, where we recall that $\pi_{2^{m}}=2^{-m}|0\rangle\langle
0|+\left(  1-2^{-m}\right)  |1\rangle\langle1|$. Then the original question is
equivalent to the question of whether there exists a channel $\mathcal{N}$
that takes the first box to the second for $n>m$. Such a channel would then
perform the transformations:%
\begin{equation}
\mathcal{N}(|0\rangle\langle0|)=|0\rangle\langle0|,
\end{equation}
\vspace{-0.3in}
\begin{multline}
\mathcal{N}(2^{-m}|0\rangle\langle0|+\left(  1-2^{-m}\right)  |1\rangle
\langle1|)\\
=2^{-n}|0\rangle\langle0|+\left(  1-2^{-n}\right)  |1\rangle\langle1|.
\end{multline}
By linearity of the channel, consider that we can conclude the action of the
channel on the orthogonal state $|1\rangle\langle1|$:%
\begin{align}
&  \mathcal{N}(|1\rangle\langle1|)\nonumber\\
&  =\mathcal{N}(\left(  1-2^{-m}\right)  ^{-1}\left(  \pi_{2^{m}}%
-2^{-m}|0\rangle\langle0|\right)  )\\
&  =\left(  1-2^{-m}\right)  ^{-1}\left[  \mathcal{N}(\pi_{2^{m}}%
)-2^{-m}\mathcal{N}(|0\rangle\langle0|)\right]  \\
&  =\left(  1-2^{-m}\right)  ^{-1}\pi_{2^{n}}-\left(  1-2^{-m}\right)
^{-1}2^{-m}|0\rangle\langle0|\\
&  =\left(  1-2^{-m}\right)  ^{-1}\left(  2^{-n}|0\rangle\langle0|+\left(
1-2^{-n}\right)  |1\rangle\langle1|\right)  \nonumber\\
&  \qquad-\left(  1-2^{-m}\right)  ^{-1}2^{-m}|0\rangle\langle0|\\
&  =\frac{2^{-n}-2^{-m}}{  1-2^{-m}  }|0\rangle\langle
0|+\frac{1-2^{-n}}{1-2^{-m}}|1\rangle\langle1|.
\end{align}
If $n>m$, then we have that $\frac{2^{-n}-2^{-m}}{\left(  1-2^{-m}\right)  }%
<0$, so that $\mathcal{N}(|1\rangle\langle1|)$ is not a quantum state. Thus,
there cannot exist a quantum channel performing the transformation in
\eqref{eq:outlawed-trans} whenever $n>m$.

By the same reasoning, we have that $(|0\rangle\langle0|,\pi_M)\not \rightarrow
(|0\rangle\langle0|,\pi_N)$ whenever $N>M$.

\section{Entropic characterizations of exact distinguishability distillation
and dilution}

\subsection{Exact distillable distinguishability}

\label{app:exact-distill-distin}We prove the equality in
\eqref{eq:exact-distillable-distinguishability}:%
\begin{equation}
D_{d}^{0}(\rho,\sigma)=D_{\min}(\rho\Vert\sigma).
\end{equation}
Recall that%
\begin{multline}
D_{d}^{0}(\rho,\sigma):=\label{eq:app-D0-distill}\\
\log_{2}\sup_{\mathcal{P}\in\text{CPTP}}\left\{  M:\mathcal{P}(\rho
)=|0\rangle\langle0|,\ \mathcal{P}(\sigma)=\pi_{M}\right\}  ,
\end{multline}
First suppose that $\operatorname{Tr}[\Pi_{\rho}\sigma]\neq0$. Consider that
the measurement channel%
\begin{equation}
\mathcal{M}(\omega)=\operatorname{Tr}[\Pi_{\rho}\omega]|0\rangle
\langle0|+\operatorname{Tr}[\left(  I-\Pi_{\rho}\right)  \omega]|1\rangle
\langle1|
\end{equation}
achieves%
\begin{align}
\mathcal{M}(\rho) &  =|0\rangle\langle0|,\\
\mathcal{M}(\sigma) &  =\operatorname{Tr}[\Pi_{\rho}\sigma]|0\rangle
\langle0|+\operatorname{Tr}[\left(  I-\Pi_{\rho}\right)  \sigma]|1\rangle
\langle1|\\
&  =\pi_{M=1/\operatorname{Tr}[\Pi_{\rho}\sigma]},
\end{align}
so that%
\begin{align}
D_{d}^{0}(\rho,\sigma) &  \geq\log_{2}\left(  1/\operatorname{Tr}[\Pi_{\rho
}\sigma]\right)  \label{eq:app-D0-distill-lower-D0-1}\\
&  =-\log_{2}\operatorname{Tr}[\Pi_{\rho}\sigma]\\
&  =D_{\min}(\rho\Vert\sigma).\label{eq:app-D0-distill-lower-D0-3}%
\end{align}

Now let $\mathcal{P}$ be a particular quantum channel such that $\mathcal{P}%
(\rho)=|0\rangle\langle0|$ and $\mathcal{P}(\sigma)=\pi_{M}$. Then by the
data-processing inequality for $D_{\min}$ as recalled in
\eqref{eq:app-DP-PR-order-0}, we find that%
\begin{align}
D_{\min}(\rho\Vert\sigma) &  \geq D_{\min}(\mathcal{P}(\rho)\Vert
\mathcal{P}(\sigma))\\
&  =D_{\min}(|0\rangle\langle0|\Vert\pi_{M})\\
&  =\log_{2}M.
\end{align}
Since the inequality $D_{\min}(\rho\Vert\sigma)\geq\log_{2}M$ holds for all
channels $\mathcal{P}$ satisfying the constraints in
\eqref{eq:app-D0-distill}, we conclude that%
\begin{equation}
D_{\min}(\rho\Vert\sigma)\geq D_{d}^{0}(\rho,\sigma
).\label{eq:app-D0-distill-upper-D0}%
\end{equation}
Combining
\eqref{eq:app-D0-distill-lower-D0-1}--\eqref{eq:app-D0-distill-lower-D0-3} and
\eqref{eq:app-D0-distill-upper-D0}, we conclude the equality in
\eqref{eq:exact-distillable-distinguishability}, i.e., $D_{\min}(\rho
\Vert\sigma)=D_{d}^{0}(\rho,\sigma)$.

In the case that $\operatorname{Tr}[\Pi_{\rho}\sigma]=0$, then this means that
the measurement channel above is such that $\mathcal{M}(\rho)=|0\rangle
\langle0|$ and $\mathcal{M}(\sigma)=|1\rangle\langle1|$. In this case, as
stated in the main text, the interpretation is that the box $(\rho,\sigma)$
contains an infinite number of bits of asymmetric distinguishability, so that
$D_{d}^{0}(\rho,\sigma)=\infty$. This is consistent with $D_{\min}(\rho
\Vert\sigma)=\infty$ in this case.

\subsection{Exact distinguishability cost}

\label{app:exact-dist-cost}We now prove the equality in \eqref{eq:dist-cost}:%
\begin{equation}
D_{c}^{0}(\rho,\sigma)=D_{\max}(\rho\Vert\sigma).
\end{equation}
First recall that%
\begin{multline}
D_{c}^{0}(\rho,\sigma):=\label{eq:app-dist-cost}\\
\log_{2}\inf_{\mathcal{P}\in\text{CPTP}}\left\{  M:\mathcal{P}(|0\rangle
\langle0|)=\rho,\ \mathcal{P}(\pi_{M})=\sigma\right\}  .
\end{multline}
Let us first suppose that $\operatorname{supp}(\rho)\subseteq
\operatorname{supp}(\sigma)$ and $D_{\max}(\rho\Vert\sigma)=0$. By definition,
this means that the condition $\rho\leq\sigma$ holds, which in turn implies
that $\sigma-\rho\geq0$. Given the characterization of the normalized trace
distance in \eqref{eq:TD-dual}, this means that we can set $Y=\sigma-\rho$.
Since $\operatorname{Tr}[Y]=0$, we conclude that $\frac{1}{2}\left\Vert
\sigma-\rho\right\Vert _{1}=0$. Since $\left\Vert \cdot\right\Vert _{1}$ is a
norm, this means that $\rho=\sigma$. So in this trivial case, it follows that
we can take $\mathcal{P}$ in \eqref{eq:app-dist-cost} to be the replacer
channel $\operatorname{Tr}[\cdot]\rho$ and it follows that we can achieve the
dilution task with zero bits of asymmetric distinguishability. So then
$D_{c}^{0}(\rho,\sigma)=0$ if $\operatorname{supp}(\rho)\subseteq
\operatorname{supp}(\sigma)$ and $D_{\max}(\rho\Vert\sigma)=0$.

Now suppose that $\operatorname{supp}(\rho)\subseteq\operatorname{supp}%
(\sigma)$ and $D_{\max}(\rho\Vert\sigma)>0$. Let $\lambda>0$ be such that
$2^{\lambda}\sigma\geq\rho$. This then means that $2^{\lambda}\sigma-\rho
\geq0$, so that $\omega:=\frac{2^{\lambda}\sigma-\rho}{2^{\lambda}-1}$ is a
quantum state. Furthermore, we have that%
\begin{equation}
\sigma=2^{-\lambda}\rho+\left(  1-2^{-\lambda}\right)  \omega.
\end{equation}
Then by means of the following channel%
\begin{equation}
\mathcal{P}(\tau)=\langle0|\tau|0\rangle\rho+\langle1|\tau|1\rangle\omega,
\end{equation}
we have that%
\begin{align}
\mathcal{P}(|0\rangle\langle0|) &  =\rho,\label{eq:dilute-protocol-one-shot-1}%
\\
\mathcal{P}(\pi_{2^{\lambda}}) &  =2^{-\lambda}\rho+\left(  1-2^{-\lambda
}\right)  \omega=\sigma,\label{eq:dilute-protocol-one-shot-2}%
\end{align}
so that this protocol accomplishes the distinguishability dilution task. This
means that%
\begin{equation}
D_{c}^{0}(\rho,\sigma)\leq\lambda.
\end{equation}
Now taking the infimum over all $\lambda$ satisfying $2^{\lambda}\sigma
\geq\rho$, we conclude that%
\begin{equation}
D_{c}^{0}(\rho,\sigma)\leq D_{\max}(\rho\Vert\sigma
).\label{eq:app-dist-cost-lower-Dmax}%
\end{equation}

Now consider an arbitrary channel $\mathcal{P}$\ that accomplishes the
transformation $(|0\rangle\langle0|,\pi)\rightarrow(\rho,\sigma)$. By the data
processing inequality for the max-relative entropy as recalled in
\eqref{eq:dmax-DP}, we have that%
\begin{align}
\log_{2}M &  =D_{\max}(|0\rangle\langle0|\Vert\pi_{M}%
)\label{eq:app-dist-cost-upper-Dmax-1}\\
&  \geq D_{\max}(\mathcal{P}(|0\rangle\langle0|)\Vert\mathcal{P}(\pi_{M}))\\
&  =D_{\max}(\rho\Vert\sigma).\label{eq:app-dist-cost-upper-Dmax-3}%
\end{align}
Taking an infimum over all such protocols, we conclude that%
\begin{equation}
D_{c}^{0}(\rho,\sigma)\geq D_{\max}(\rho\Vert\sigma
).\label{eq:app-dist-cost-upper-Dmax}%
\end{equation}
Putting together \eqref{eq:app-dist-cost-lower-Dmax} and
\eqref{eq:app-dist-cost-upper-Dmax}, we conclude the equality in
\eqref{eq:dist-cost}, i.e., $D_{c}^{0}(\rho,\sigma)=D_{\max}(\rho\Vert\sigma)$.

In the case that $\operatorname{supp}(\rho)\not \subseteq \operatorname{supp}%
(\sigma)$, we have that $\operatorname{Tr}[\Pi_{\sigma}\rho]<1$ and by
definition $D_{\max}(\rho\Vert\sigma)=\infty$. This is consistent with the
fact that, in such a case, there is no finite $\lambda\geq0$ such that
$2^{\lambda}\sigma-\rho\geq0$. For if there were, then we would have that%
\begin{align}
2^{\lambda}-1  &  =\operatorname{Tr}[\left(  2^{\lambda}\sigma-\rho\right)
]\\
&  =\operatorname{Tr}[\left\{  2^{\lambda}\sigma\geq\rho\right\}  \left(
2^{\lambda}\sigma-\rho\right)  ]\\
&  \geq\operatorname{Tr}[\Pi_{\sigma}\left(  2^{\lambda}\sigma-\rho\right)
]\\
&  =\operatorname{Tr}[2^{\lambda}\sigma]-\operatorname{Tr}[\Pi_{\sigma}\rho]\\
&  =2^{\lambda}-\operatorname{Tr}[\Pi_{\sigma}\rho],
\end{align}
where the inequality follows from $\operatorname{Tr}[\left\{  A\geq0\right\}
A]\geq\operatorname{Tr}[\Pi A]$ for any Hermitian operator $A$, projector
$\Pi$, and $\left\{  A\geq0\right\}  $ denoting the projection onto the positive
eigenspace of~$A$. The above implies that%
\begin{equation}
\operatorname{Tr}[\Pi_{\sigma}\rho]\geq1,
\end{equation}
contradicting the fact that $\operatorname{Tr}[\Pi_{\sigma}\rho]<1$ when
$\operatorname{supp}(\rho)\not \subseteq \operatorname{supp}(\sigma)$.

As explained in the main text, when $\operatorname{supp}(\rho)\not \subseteq
\operatorname{supp}(\sigma)$, there is no finite value of $M$ nor any quantum
channel $\mathcal{P}$\ such that $\mathcal{P}(|0\rangle\langle0|)=\rho$ and
$\mathcal{P}(\pi_{M})=\sigma$. If there were, then by the general fact that,
for a quantum channel $\mathcal{N}$ and states $\tau$ and $\omega$,
$\operatorname{supp}(\mathcal{N}(\tau))\subseteq\operatorname{supp}%
(\mathcal{N}(\omega))$ if $\operatorname{supp}(\tau)\subseteq
\operatorname{supp}(\omega)$ \cite[Appendix~B]{Renner2005} and the fact that
$\operatorname{supp}(|0\rangle\langle0|)\subseteq\operatorname{supp}(\pi_{M})$
for all $M<\infty$, the existence of such a channel $\mathcal{P}$ would
contradict the assumption that $\operatorname{supp}(\rho)\not \subseteq
\operatorname{supp}(\sigma)$. The interpretation then is as stated in the main
text:\ that $D_{c}^{0}(\rho,\sigma)=\infty$ when $\operatorname{supp}%
(\rho)\not \subseteq \operatorname{supp}(\sigma)$, which is consistent with
the fact that $D_{\max}(\rho\Vert\sigma)=\infty$ in such a case.

\section{Entropic characterizations of approximate distinguishability
distillation and dilution}

\subsection{Approximate distillable distinguishability}

\label{app:one-shot-distill-distin}We prove the equality in
\eqref{eq:approx-distill}:%
\begin{equation}
D_{d}^{\varepsilon}(\rho,\sigma)=D_{\min}^{\varepsilon}(\rho\Vert\sigma).
\end{equation}
First recall that%
\begin{multline}
D_{d}^{\varepsilon}(\rho,\sigma):=\\
\log_{2}\sup_{\mathcal{P}\in\text{CPTP}}\{M:\mathcal{P}(\rho)\approx
_{\varepsilon}|0\rangle\langle0|,\ \mathcal{P}(\sigma)=\pi_{M}\}.
\end{multline}
Let $\Lambda$ be an arbitrary measurement operator satisfying $0\leq
\Lambda\leq I$ and $\operatorname{Tr}[\Lambda\rho]\geq1-\varepsilon$. Then we
can take the channel $\mathcal{P}$ to be as follows:%
\begin{equation}
\mathcal{P}(\omega)=\operatorname{Tr}[\Lambda\omega]|0\rangle\langle
0|+\operatorname{Tr}[\left(  I-\Lambda\right)  \omega]|1\rangle\langle1|,
\end{equation}
and we find that%
\begin{align}
&  \frac{1}{2}\left\Vert \mathcal{P}(\rho)-|0\rangle\langle0|\right\Vert
_{1}\nonumber\\
&  =\frac{1}{2}\left\Vert \operatorname{Tr}[\Lambda\rho]|0\rangle
\langle0|+\operatorname{Tr}[\left(  I-\Lambda\right)  \rho]|1\rangle
\langle1|-|0\rangle\langle0|\right\Vert _{1}\\
&  =\frac{1}{2}\left\Vert -\operatorname{Tr}[\left(  I-\Lambda\right)
\rho]|0\rangle\langle0|+\operatorname{Tr}[\left(  I-\Lambda\right)
\rho]|1\rangle\langle1|\right\Vert _{1}\\
&  =\left(  \operatorname{Tr}[\left(  I-\Lambda\right)  \rho]\right)  \frac
{1}{2}\left\Vert |1\rangle\langle1|-|0\rangle\langle0|\right\Vert _{1}\\
&  =\operatorname{Tr}[\left(  I-\Lambda\right)  \rho]\leq\varepsilon.
\end{align}
Furthermore, we have that%
\begin{align}
\mathcal{P}(\sigma) &  =\operatorname{Tr}[\Lambda\sigma]|0\rangle
\langle0|+\operatorname{Tr}[\left(  I-\Lambda\right)  \sigma]|1\rangle
\langle1|\\
&  =\pi_{M=1/\operatorname{Tr}[\Lambda\sigma]}.
\end{align}
So this means that%
\begin{align}
D_{d}^{\varepsilon}(\rho,\sigma) &  \geq\log_{2}\!\left(  1/\operatorname{Tr}%
[\Lambda\sigma]\right)  \\
&  =-\log_{2}\operatorname{Tr}[\Lambda\sigma].
\end{align}
Now maximizing the right-hand side with respect to all $\Lambda$ satisfying
$0\leq\Lambda\leq I$ and $\operatorname{Tr}[\Lambda\rho]\geq1-\varepsilon$, we
conclude that%
\begin{equation}
D_{d}^{\varepsilon}(\rho,\sigma)\geq D_{\min}^{\varepsilon}(\rho\Vert
\sigma).\label{eq:app-dist-upper-DHeps}%
\end{equation}

To see the other inequality, let $\mathcal{P}$ be an arbitrary channel
satisfying $\mathcal{P}(\rho)\approx_{\varepsilon}|0\rangle\langle0|$
and$\ \mathcal{P}(\sigma)=\pi_{M}$. Then by the data processing inequality for
$D_{\min}^{\varepsilon}$, we have that%
\begin{align}
D_{\min}^{\varepsilon}(\rho\Vert\sigma)  &  \geq D_{\min}^{\varepsilon
}(\mathcal{P}(\rho)\Vert\mathcal{P}(\sigma))\\
&  =D_{\min}^{\varepsilon}(\mathcal{P}(\rho)\Vert\pi_{M})\\
&  \geq\log_{2}M. \label{eq:dh-eps-to-distill-logM}%
\end{align}
The last inequality above is a consequence of the following reasoning:\ Let
$\Delta(\cdot)=|0\rangle\langle0|\left(  \cdot\right)  |0\rangle
\langle0|+|1\rangle\langle1|\left(  \cdot\right)  |1\rangle\langle1|$ denote
the completely dephasing channel. Since $\mathcal{P}(\rho)\approx
_{\varepsilon}|0\rangle\langle0|$, we find from applying the data processing
inequality for normalized trace distance that%
\begin{align}
\varepsilon &  \geq\frac{1}{2}\left\Vert \mathcal{P}(\rho)-|0\rangle
\langle0|\right\Vert _{1}\nonumber\\
&  \geq\frac{1}{2}\left\Vert (\Delta\circ\mathcal{P})(\rho)-\Delta
(|0\rangle\langle0|)\right\Vert _{1}\nonumber\\
&  =\frac{1}{2}\left\Vert (\Delta\circ\mathcal{P})(\rho)-|0\rangle
\langle0|\right\Vert _{1}\nonumber\\
&  =\frac{1}{2}\left\Vert \langle0|\mathcal{P}(\rho)|0\rangle|0\rangle
\langle0|+\langle1|\mathcal{P}(\rho)|1\rangle|1\rangle\langle1|-|0\rangle
\langle0|\right\Vert _{1}\nonumber\\
&  =1-\langle0|\mathcal{P}(\rho)|0\rangle,
\end{align}
which implies that $\langle0|\mathcal{P}(\rho)|0\rangle\geq1-\varepsilon$.
Thus, we can take $\Lambda=|0\rangle\langle0|$ in the definition of $D_{\min
}^{\varepsilon}(\mathcal{P}(\rho)\Vert\pi_{M})$, and we have that
$\operatorname{Tr}[\Lambda\mathcal{P}(\rho)]\geq1-\varepsilon$ while
$\operatorname{Tr}[\Lambda\pi_{M}]=1/M$. Since $D_{\min}^{\varepsilon
}(\mathcal{P}(\rho)\Vert\pi_{M})$ involves an optimization over all
measurement operators $\Lambda$\ satisfying $\operatorname{Tr}[\Lambda
\mathcal{P}(\rho)]\geq1-\varepsilon$, we conclude the inequality in \eqref{eq:dh-eps-to-distill-logM}.

Since the inequality $D_{\min}^{\varepsilon}(\rho\Vert\sigma)\geq\log_{2}M$
holds for all possible distinguishability distillation protocols, we conclude
that%
\begin{equation}
D_{\min}^{\varepsilon}(\rho\Vert\sigma)\geq D_{d}^{\varepsilon}(\rho
,\sigma).\label{eq:app-dist-lower-DHeps}%
\end{equation}
By combining the inequalities in \eqref{eq:app-dist-upper-DHeps} and
\eqref{eq:app-dist-lower-DHeps}, we conclude the equality in
\eqref{eq:approx-distill}, i.e., $D_{\min}^{\varepsilon}(\rho\Vert
\sigma)=D_{d}^{\varepsilon}(\rho,\sigma)$.

It is worthwhile to mention a somewhat singular case. In the case that
$\operatorname{supp}(\rho)\not \subseteq \operatorname{supp}(\sigma)$, we have
that $\operatorname{Tr}[\Pi_{\sigma}\rho]<1$, which means that
$\operatorname{Tr}[\left(  I-\Pi_{\sigma}\right)  \rho]>0$. If we also have
that $\operatorname{Tr}[\left(  I-\Pi_{\sigma}\right)  \rho]\geq1-\varepsilon
$, then we can take the channel $\mathcal{P}$ to be as follows:%
\begin{equation}
\mathcal{P}(\omega)=\operatorname{Tr}[\left(  I-\Pi_{\sigma}\right)
\omega]|0\rangle\langle0|+\operatorname{Tr}[\Pi_{\sigma}\omega]|1\rangle
\langle1|.
\end{equation}
In such a case, we have that $\mathcal{P}(\rho)\approx_{\varepsilon}%
|0\rangle\langle0|$, while $\mathcal{P}(\sigma)=|1\rangle\langle
1|=\lim_{M\rightarrow\infty}\pi_{M}$, implying that $D_{\min}^{\varepsilon
}(\rho\Vert\sigma)=D_{d}^{\varepsilon}(\rho,\sigma)=\infty$ in this case.

\subsection{Approximate distinguishability cost}

\label{app:one-shot-dist-cost}Here we prove the equality in
\eqref{eq:eps-approx-cost}:%
\begin{equation}
D_{c}^{\varepsilon}(\rho,\sigma)=D_{\max}^{\varepsilon}(\rho\Vert\sigma).
\end{equation}
Recall that%
\begin{multline}
D_{c}^{\varepsilon}(\rho,\sigma):=\\
\log_{2}\inf_{\mathcal{P}\in\text{CPTP}}\{M:\mathcal{P}(|0\rangle
\langle0|)\approx_{\varepsilon}\rho,\ \mathcal{P}(\pi_{M})=\sigma\}.
\end{multline}
Let $\widetilde{\rho}$ be a state such that $\frac{1}{2}\left\Vert
\rho-\widetilde{\rho}\right\Vert _{1}\leq\varepsilon$. Then by executing the
protocol in
\eqref{eq:dilute-protocol-one-shot-1}--\eqref{eq:dilute-protocol-one-shot-2},
but replacing $\rho$ with $\widetilde{\rho}$, we find that%
\begin{equation}
D_{c}^{\varepsilon}(\rho,\sigma)\leq D_{\max}(\widetilde{\rho}\Vert\sigma).
\end{equation}
Since this is possible for any state $\widetilde{\rho}$ satisfying $\frac
{1}{2}\left\Vert \rho-\widetilde{\rho}\right\Vert _{1}\leq\varepsilon$, we
conclude that%
\begin{equation}
D_{c}^{\varepsilon}(\rho,\sigma)\leq D_{\max}^{\varepsilon}(\rho\Vert
\sigma).\label{eq:d-cost-lower-smooth-d-max}%
\end{equation}

To see the other inequality, consider an arbitrary channel $\mathcal{P}%
$\ performing the transformation $\mathcal{P}(|0\rangle\langle0|)\approx
_{\varepsilon}\rho$ and $\mathcal{P}(\pi_{M})=\sigma$. Then from the data
processing inequality for the max-relative entropy, as recalled in
\eqref{eq:dmax-DP}, and its definition, we conclude that%
\begin{align}
\log_{2}M &  =D_{\max}(|0\rangle\langle0|\Vert\pi_{M})\\
&  \geq D_{\max}(\mathcal{P}(|0\rangle\langle0|)\Vert\mathcal{P}(\pi_{M}))\\
&  =D_{\max}(\mathcal{P}(|0\rangle\langle0|)\Vert\sigma)\\
&  \geq D_{\max}^{\varepsilon}(\rho\Vert\sigma).
\end{align}
Since the inequality holds for an arbitrary channel $\mathcal{P}$\ performing
the transformation $\mathcal{P}(|0\rangle\langle0|)\approx_{\varepsilon}\rho$
and $\mathcal{P}(\pi_{M})=\sigma$, we conclude that%
\begin{equation}
D_{c}^{\varepsilon}(\rho,\sigma)\geq D_{\max}^{\varepsilon}(\rho\Vert
\sigma).\label{eq:d-cost-upper-smooth-d-max}%
\end{equation}
By combining the inequalities in \eqref{eq:d-cost-lower-smooth-d-max} and
\eqref{eq:d-cost-upper-smooth-d-max}, we conclude the equality in
\eqref{eq:eps-approx-cost}, i.e., $D_{c}^{\varepsilon}(\rho,\sigma)=D_{\max
}^{\varepsilon}(\rho\Vert\sigma)$.

\section{Operational proof for inequality relating smooth min- and
max-relative entropies}

\label{app:bnd-dmax-dheps}Here we prove the inequality in
\eqref{eq:dh-dmax-relation-operational}:%
\begin{equation}
D_{\min}^{\varepsilon_{1}}(\rho\Vert\sigma)\leq D_{\max}^{\varepsilon_{2}%
}(\rho\Vert\sigma)+\log_{2}\!\left(  \frac{1}{1-\varepsilon_{1}-\varepsilon
_{2}}\right)  , \label{eq:app-dmax-dheps-bound}%
\end{equation}
for $\varepsilon_{1},\varepsilon_{2}\geq0$ and $\varepsilon_{1}+\varepsilon
_{2}<1$.

First, consider that an arbitrary protocol performing the transformation
$(|0\rangle\langle0|,\pi_M)\rightarrow(\widetilde{0}_{\varepsilon},\pi_{K})$ is
required to obey the following inequality%
\begin{align}
\log_{2}K &  \leq D_{\min}^{\varepsilon}(|0\rangle\langle0|\Vert\pi_{M})\\
&  =\log_{2}M+\log_{2}\left(  1/\left[  1-\varepsilon\right]  \right)
.\label{eq:dheps-eq-m-bits-of-AD}%
\end{align}
To see the equality in \eqref{eq:dheps-eq-m-bits-of-AD}, consider that
$\Lambda=\left(  1-\varepsilon\right)  |0\rangle\langle0|$ is a measurement
operator achieving $\operatorname{Tr}[\Lambda|0\rangle\langle0|]\geq
1-\varepsilon$, while $\operatorname{Tr}[\Lambda\pi_{M}]=\left(
1-\varepsilon\right)  /M$, implying that%
\begin{equation}
D_{\min}^{\varepsilon}(|0\rangle\langle0|\Vert\pi_{M})\geq\log_{2}M+\log
_{2}\left(  1/\left[  1-\varepsilon\right]  \right)  .
\end{equation}
To see the other inequality, suppose that $\operatorname{Tr}[\Lambda
|0\rangle\langle0|]\geq1-\varepsilon$. Then we have that%
\begin{align}
&  \operatorname{Tr}[\Lambda\pi_{M}]\nonumber\\
&  =\frac{1}{M}\operatorname{Tr}[\Lambda|0\rangle\langle0|]+\left(  1-\frac
{1}{M}\right)  \operatorname{Tr}[\Lambda|1\rangle\langle1|]\\
&  \geq\frac{1}{M}\operatorname{Tr}[\Lambda|0\rangle\langle0|]\\
&  \geq\frac{1-\varepsilon}{M}.
\end{align}
Since this is a uniform bound holding for all measurement operators $\Lambda$
satisfying $\operatorname{Tr}[\Lambda|0\rangle\langle0|]\geq1-\varepsilon$, we
conclude that%
\begin{align}
D_{\min}^{\varepsilon}(|0\rangle\langle0|\Vert\pi_{M}) &  \leq-\log
_{2}\!\left(  \frac{1-\varepsilon}{M}\right)  \\
&  =\log_{2}M+\log_{2}(1/\left[  1-\varepsilon\right]  ),
\end{align}
completing the proof of the equality in \eqref{eq:dheps-eq-m-bits-of-AD}.

Given that the bound $\log_{2}K\leq\log_{2}M+\log_{2}\left(  1/\left[
1-\varepsilon\right]  \right)  $ holds for an arbitrary channel performing the
transformation $(|0\rangle\langle0|,\pi_M)\rightarrow(\widetilde{0}%
_{\varepsilon},\pi_{K})$, we can consider a particular way of completing it in
two steps. Fix $\varepsilon_{1},\varepsilon_{2}\geq0$ such that $\varepsilon
_{1}+\varepsilon_{2}<1$. In the first step, we perform the dilution
transformation $(|0\rangle\langle0|,\pi_M)\rightarrow(\rho_{\varepsilon_{2}%
},\sigma)$ optimally and in the second, we perform the distillation
transformation $(\rho,\sigma)\rightarrow(\widetilde{0}_{\varepsilon_{1}}%
,\pi_{K})$ optimally. For the dilution part, we have that $\log_{2}M=D_{\max
}^{\varepsilon_{2}}(\rho\Vert\sigma)$ and there exists a channel
$\mathcal{P}_{1}$ such that $\mathcal{P}_{1}(|0\rangle\langle0|)=\rho
_{\varepsilon_{2}}\approx_{\varepsilon_{2}}\rho$ and $\mathcal{P}_{1}(\pi
_{M})=\sigma$. For the distillation part, we have that $\log_{2}K=D_{\min
}^{\varepsilon_{1}}(\rho\Vert\sigma)$ and there exists a channel
$\mathcal{P}_{2}$ such that $\mathcal{P}_{2}(\rho)=\widetilde{0}%
_{\varepsilon_{1}}\approx_{\varepsilon_{1}}|0\rangle\langle0|$ and
$\mathcal{P}_{2}(\sigma)=\pi_{K}$. By composing the two channels, we have that%
\begin{equation}
(\mathcal{P}_{2}\circ\mathcal{P}_{1})(\pi_{M})=\pi_{K},
\end{equation}
while%
\begin{align}
&  \frac{1}{2}\left\Vert (\mathcal{P}_{2}\circ\mathcal{P}_{1})(|0\rangle
\langle0|)-|0\rangle\langle0|\right\Vert _{1}\nonumber\\
&  \leq\frac{1}{2}\left\Vert (\mathcal{P}_{2}\circ\mathcal{P}_{1}%
)(|0\rangle\langle0|)-\mathcal{P}_{2}(\rho)\right\Vert _{1}\nonumber\\
&  \qquad+\frac{1}{2}\left\Vert \mathcal{P}_{2}(\rho)-|0\rangle\langle
0|\right\Vert _{1}\\
&  \leq\frac{1}{2}\left\Vert \mathcal{P}_{1}(|0\rangle\langle0|)-\rho
\right\Vert _{1}+\varepsilon_{1}\\
&  \leq\varepsilon_{2}+\varepsilon_{1}.
\end{align}
So this means that we have a protocol $(|0\rangle\langle0|,\pi_M)\rightarrow
(\widetilde{0}_{\varepsilon_{1}+\varepsilon_{2}},\pi_{K})$ with $\log
_{2}M=D_{\max}^{\varepsilon_{2}}(\rho\Vert\sigma)$ and $\log_{2}K=D_{\min
}^{\varepsilon_{1}}(\rho\Vert\sigma)$. By \eqref{eq:dheps-eq-m-bits-of-AD}, we
then conclude the inequality in \eqref{eq:dh-dmax-relation-operational}, as
restated in \eqref{eq:app-dmax-dheps-bound}.

\section{Asymptotic distillable distinguishability and distinguishability
cost}

\label{app:asymp-dist-cost}As a direct consequence of
\eqref{eq:approx-distill} and results from
\cite{tomamichel2013hierarchy,li14}, the following
expansion holds for sufficiently large~$n$:%
\begin{multline}
D_{d}^{\varepsilon}(\rho^{\otimes n},\sigma^{\otimes n}%
)=\label{eq:app-d-dist-2nd-order}\\
nD(\rho\Vert\sigma)+\sqrt{nV(\rho\Vert\sigma)}\Phi^{-1}(\varepsilon)+O(\log
n),
\end{multline}
where $D(\rho\Vert\sigma)$ is the quantum relative entropy. The relative
entropy variance $V(\rho\Vert\sigma)$ \cite{tomamichel2013hierarchy,li14} is
defined as%
\begin{equation}
V(\rho\Vert\sigma):=\operatorname{Tr}[\rho\left(  \log_{2}\rho-\log_{2}%
\sigma-D(\rho\Vert\sigma)\right)  ^{2}],
\end{equation}
if $\operatorname{supp}(\rho)\subseteq\operatorname{supp}(\sigma)$ and is
otherwise undefined. Furthermore, $\Phi^{-1}(\varepsilon)$ is the inverse of
the cumulative normal distribution function, defined as%
\begin{equation}
\Phi^{-1}(\varepsilon)=\sup\left\{  a\in\mathbb{R}\ |\ \Phi(a)\leq
\varepsilon\right\}  ,
\end{equation}
where%
\begin{equation}
\Phi(a)=\frac{1}{\sqrt{2\pi}}\int_{-\infty}^{a}dx\ \exp\!\left(  \frac{-x^{2}%
}{2}\right)  .
\end{equation}

Based on the inequality in \eqref{eq:dh-dmax-relation-operational}, we have
that%
\begin{equation}
D_{\min}^{1-\varepsilon-\delta}(\rho\Vert\sigma)\leq D_{\max}^{\varepsilon}%
(\rho\Vert\sigma)+\log_{2}\!\left(  \frac{1}{\delta}\right)  .\nonumber
\end{equation}
Then by picking $\delta=1/\sqrt{n}$, and applying \eqref{eq:eps-approx-cost},
\eqref{eq:approx-distill}, \eqref{eq:app-d-dist-2nd-order}, and the fact that
$\Phi^{-1}(1-\varepsilon)=-\Phi^{-1}(\varepsilon)$, we find that%
\begin{multline}
D_{c}^{\varepsilon}(\rho^{\otimes n},\sigma^{\otimes n})\geq\\
nD(\rho\Vert\sigma)-\sqrt{nV(\rho\Vert\sigma)}\Phi^{-1}(\varepsilon)+O(\log
n).
\end{multline}

By following the proof of \cite[Eq.~(21)]{tomamichel2013hierarchy}, but
instead using the normalized trace distance as the metric for smooth
max-relative entropy, we find that%
\begin{multline}
D_{\max}^{\varepsilon}(\rho\Vert\sigma)\leq D_{\min}^{1-\varepsilon^{2}}%
(\rho\Vert\sigma)\\
+\log_{2}\left\vert \text{spec}(\sigma)\right\vert +\log_{2}\!\left(  \frac
{1}{1-\varepsilon^{2}}\right)  ,\label{eq:app-dmaxeps-<=-dh1-eps}%
\end{multline}
where $\varepsilon\in(0,1)$ and $\left\vert \text{spec}(\sigma)\right\vert $
is equal to the number of distinct eigenvalues of $\sigma$. We give a detailed
proof of \eqref{eq:app-dmaxeps-<=-dh1-eps} in
Appendix~\ref{app:opposite-bnd-dmax-dheps}. By the operational interpretations
of $D_{\max}^{\varepsilon}$ and $D_{\min}^{1-\varepsilon^{2}}$, the inequality in
\eqref{eq:app-dmaxeps-<=-dh1-eps} can equivalently be written as
\begin{multline}
D_{c}^{\varepsilon}(\rho,\sigma)\leq D_{d}^{1-\varepsilon^{2}}(\rho,\sigma)\\
+\log_{2}\left\vert \text{spec}(\sigma)\right\vert +\log_{2}\!\left(  \frac
{1}{1-\varepsilon^{2}}\right)  .
\end{multline}
Now accounting for the fact that $\left\vert \text{spec}(\sigma^{\otimes
n})\right\vert =O(\log n)$ and applying \eqref{eq:app-d-dist-2nd-order}, we
conclude that%
\begin{multline}
D_{c}^{\varepsilon}(\rho^{\otimes n},\sigma^{\otimes n})\leq\\
nD(\rho\Vert\sigma)-\sqrt{nV(\rho\Vert\sigma)}\Phi^{-1}(\varepsilon
^{2})+O(\log n).
\end{multline}

Thus, we have that%
\begin{equation}
D_{c}^{\varepsilon}(\rho^{\otimes n},\sigma^{\otimes n})=nD(\rho\Vert
\sigma)+O(\sqrt{n}).
\end{equation}

\section{Bound relating smooth max- and min-relative entropies}

\label{app:opposite-bnd-dmax-dheps} Here we prove the following bound:%
\begin{multline}
D_{\max}^{\varepsilon}(\rho\Vert\sigma)\leq D_{\min}^{1-\varepsilon^{2}}%
(\rho\Vert\sigma)\label{eq:app-dmax-dheps-relation-specs}\\
+\log_{2}\left\vert \text{spec}(\sigma)\right\vert +\log_{2}\!\left(  \frac
{1}{1-\varepsilon^{2}}\right)  ,
\end{multline}
where $\left\vert \text{spec}(\sigma)\right\vert $ is equal to the number of
distinct eigenvalues of $\sigma$.

The proof follows the proof of \cite[Eq.~(21)]{tomamichel2013hierarchy}
closely, but instead using the normalized trace distance as the metric for
smooth max-relative entropy and accounting for a minor typo present in the
proof of \cite[Eq.~(21)]{tomamichel2013hierarchy}.

Let the eigendecomposition of $\sigma$ be $\sigma=\sum_{x}\lambda_{x}^{\sigma
}\Pi_{x}^{\sigma}$, where $\Pi_{x}^{\sigma}$ is the projection onto the
eigenspace of $\sigma$ with eigenvalue $\lambda^{\sigma}_{x}$. Let
$\mathcal{E}_{\sigma}(\cdot)=\sum_{x}\Pi_{x}^{\sigma}(\cdot)\Pi_{x}^{\sigma}$
denote the pinching quantum channel. In what follows, we make use of the
pinching inequality \cite{Hay02}:
\begin{equation}
\rho\leq\left\vert \text{spec}(\sigma)\right\vert \mathcal{E}_{\sigma}(\rho).
\end{equation}

Let $\mu$ be the largest value such that $\operatorname{Tr}[Q\mathcal{E}%
_{\sigma}(\rho)]=1-\varepsilon^{2}$, where $Q=\{\mathcal{E}_{\sigma}(\rho
)\leq2^{\mu}\sigma\}$. Due to the fact that $Q$ commutes with $\sigma$, we
have that $\mathcal{E}_{\sigma}(Q)=Q$, which implies that%
\begin{align}
\operatorname{Tr}[Q\mathcal{E}_{\sigma}(\rho)]  &  =\operatorname{Tr}%
[\mathcal{E}_{\sigma}(Q)\rho]\\
&  =\operatorname{Tr}[Q\rho]\\
&  =1-\varepsilon^{2}.
\end{align}
Then we set%
\begin{equation}
\widetilde{\rho}=\frac{Q\rho Q}{\operatorname{Tr}[Q\rho]},
\end{equation}
for which we have that%
\begin{equation}
F(\rho,\widetilde{\rho})\geq1-\varepsilon^{2},
\end{equation}
by applying \cite[Lemma~9.4.1]{W17book}. This in turn implies that%
\begin{equation}
\frac{1}{2}\left\Vert \rho-\widetilde{\rho}\right\Vert _{1}\leq\varepsilon,
\end{equation}
via the inequality $\frac{1}{2}\left\Vert \rho-\widetilde{\rho}\right\Vert
_{1}\leq\sqrt{1-F(\rho,\widetilde{\rho})}$ \cite{FG98}, so that $\widetilde
{\rho}$ is a candidate for the optimization involved in $D_{\max}%
^{\varepsilon}(\rho\Vert\sigma)$. Now consider that%
\begin{align}
\widetilde{\rho}  &  =\frac{Q\rho Q}{\operatorname{Tr}[Q\rho]}\\
&  \leq\frac{Q\rho Q}{1-\varepsilon^{2}}\\
&  \leq\frac{\left\vert \text{spec}(\sigma)\right\vert }{1-\varepsilon^{2}%
}Q\mathcal{E}_{\sigma}(\rho)Q\\
&  \leq\frac{2^{\mu}\left\vert \text{spec}(\sigma)\right\vert }{1-\varepsilon
^{2}}Q\sigma Q\\
&  \leq\frac{2^{\mu}\left\vert \text{spec}(\sigma)\right\vert }{1-\varepsilon
^{2}}\sigma.
\end{align}
So it follows that%
\begin{multline}
D_{\max}^{\varepsilon}(\rho\Vert\sigma)\leq D_{\max}(\widetilde{\rho}%
\Vert\sigma)\label{eq:app-dmax-to-dheps-1}\\
\leq\mu+\log_{2}\left\vert \text{spec}(\sigma)\right\vert +\log_{2}\left(
\frac{1}{1-\varepsilon^{2}}\right)  .
\end{multline}
Now consider that $\operatorname{Tr}[\left(  I-Q\right)  \rho]=\varepsilon
^{2}$ and $I-Q=\{\mathcal{E}_{\sigma}(\rho)>2^{\mu}\sigma\}$, for which we
have that%
\begin{equation}
\operatorname{Tr}[\{\mathcal{E}_{\sigma}(\rho)>2^{\mu}\sigma\}\left(
\mathcal{E}_{\sigma}(\rho)-2^{\mu}\sigma\right)  ]\geq0,
\end{equation}
implying that%
\begin{align}
\operatorname{Tr}[\left(  I-Q\right)  \sigma]  &  =\operatorname{Tr}%
[\{\mathcal{E}_{\sigma}(\rho)>2^{\mu}\sigma\}\sigma]\\
&  \leq2^{-\mu}\operatorname{Tr}[\{\mathcal{E}_{\sigma}(\rho)>2^{\mu}%
\sigma\}\mathcal{E}_{\sigma}(\rho)]\\
&  \leq2^{-\mu}.
\end{align}
Taking a negative logarithm, this gives%
\begin{equation}
-\log\operatorname{Tr}[\left(  I-Q\right)  \sigma]\geq\mu.
\end{equation}
Since $\operatorname{Tr}[\left(  I-Q\right)  \rho]=\varepsilon^{2}$, this
means that $I-Q$ is a candidate for $\Lambda$ in the definition of smooth
min-relative entropy, from which we conclude that%
\begin{align}
\mu &  \leq D_{\min}^{1-\varepsilon^{2}}(\mathcal{E}_{\sigma}(\rho)\Vert\sigma)\\
&  \leq D_{\min}^{1-\varepsilon^{2}}(\rho\Vert\sigma),
\label{eq:app-dmax-to-dheps-final}%
\end{align}
where the latter inequality follows from the data processing inequality in
\eqref{eq:hypo-DP}. Putting together \eqref{eq:app-dmax-to-dheps-1} and
\eqref{eq:app-dmax-to-dheps-final}, we arrive at \eqref{eq:app-dmax-dheps-relation-specs}.

\section{Asymptotic box transformations}

\label{app:asymp-box-tr}We now provide a proof of
Eq.~\eqref{eq:box-trans-main-result}, i.e.,%
\begin{equation}
R((\rho,\sigma)\rightarrow(\tau,\omega))=\widetilde{R}((\rho,\sigma
)\rightarrow(\tau,\omega))=\frac{D(\rho\Vert\sigma)}{D(\tau\Vert\omega)},
\end{equation}
so that the quantum relative entropy gives the optimal conversion rate for
boxes. We prove this result in two steps, called the direct part and strong
converse part.

\subsection{Achievability:\ Direct part}

We begin with the direct part. The goal is to show that for all $\varepsilon
\in(0,1]$, $\delta>0$, and sufficiently large $n$, there exists an
$(n,n[R-\delta],\varepsilon)$\ box transformation protocol%
\begin{equation}
(\rho^{\otimes n},\sigma^{\otimes n})\rightarrow
(\widetilde{\tau^{\otimes
n\left[  R-\delta\right]  }},\omega^{\otimes n\left[  R-\delta\right]  })
\end{equation}
with $R=\frac{D(\rho\Vert\sigma)}{D(\tau\Vert\omega)}$. The approach we take
here is related to an approach from \cite{KW19}.

Fix $\varepsilon\in\left(  0,1\right)  $ and $\delta>0$. Suppose that
$\varepsilon=\varepsilon_{1}+\varepsilon_{2}$, so that $\varepsilon
_{1},\varepsilon_{2}\in(0,1)$ and $\varepsilon_{1}+\varepsilon_{2}<1$. Also,
suppose that $\delta=\delta_{1}+\delta_{2}+\delta_{3}+\delta_{4}$, such that
$\delta_{1},\delta_{2},\delta_{3},\delta_{4}>0$.

Consider that we can perform the transformation $(\rho^{\otimes n}%
,\sigma^{\otimes n})\rightarrow (\widetilde{0}_{\varepsilon_{1}},\pi_{M})$
such that%
\begin{equation}
\log_{2}M=D_{\min}^{\varepsilon_{1}}(\rho^{\otimes n}\Vert\sigma^{\otimes n}).
\end{equation}
Then applying the following inequality from \cite[Proposition~3.2]%
{AMV12}\ (see also \cite[Proposition~3]{QWW17})%
\begin{equation}
D_{\min}^{\varepsilon}(\rho\Vert\sigma)\geq D_{\alpha}(\rho\Vert\sigma
)+\frac{\alpha}{\alpha-1}\log_{2}\!\left(  \frac{1}{\varepsilon}\right)  ,
\end{equation}
we find that%
\begin{equation}
\log_{2}M\geq nD_{\alpha}(\rho\Vert\sigma)+\frac{\alpha}{\alpha-1}\log
_{2}\!\left(  \frac{1}{\varepsilon_{1}}\right)  .
\end{equation}
Set $\alpha\in(0,1)$ such that%
\begin{equation}
\delta_{1}\cdot D(\tau\Vert\omega)\geq D(\rho\Vert\sigma)-D_{\alpha}(\rho
\Vert\sigma),
\end{equation}
which is possible due to \eqref{eq:conv-PR-q-rel-ent} and
\eqref{eq:app-Petz-Renyi-ordered}, and for this choice of $\alpha$, take $n$
large enough so that%
\begin{equation}
\delta_{2}\cdot D(\tau\Vert\omega)\geq\frac{\alpha}{n(1-\alpha)}\log
_{2}\!\left(  \frac{1}{\varepsilon_{1}}\right)
.\label{eq:app-large-enough-n-1}%
\end{equation}
Then we have that%
\begin{equation}
\log_{2}M\geq nD(\rho\Vert\sigma)-nD(\tau\Vert\omega)\left[  \delta_{1}%
+\delta_{2}\right]  .\label{eq:app-low-distill-dilute-tech}%
\end{equation}

Also, consider that we can perform the transformation $(|0\rangle\langle
0|,\pi)\rightarrow(\widetilde{\tau^{\otimes m}},\omega^{\otimes m})$ (with
error $\varepsilon_{2}$), for fixed $M$, by taking $m$ as large as possible so
that the following inequality still holds%
\begin{equation}
\log_{2}M\geq D_{\max}^{\varepsilon_{2}}(\tau^{\otimes m}\Vert\omega^{\otimes
m}).
\end{equation}
If it is not possible to find an $m$ to saturate the inequality, then one can
find states $\tau^{\prime}$ and $\omega^{\prime}$ with just enough
distinguishability such that%
\begin{equation}
\log_{2}M=D_{\max}^{\varepsilon_{2}}(\tau^{\otimes m}\otimes\tau^{\prime}%
\Vert\omega^{\otimes m}\otimes\omega^{\prime}),
\end{equation}
while having a negligible impact on the final parameters of the protocol. The
resulting protocol then produces the states $\approx_{\varepsilon}%
\tau^{\otimes m}\otimes\tau^{\prime}$ and $\omega^{\otimes m}\otimes
\omega^{\prime}$, and the final step is to perform a partial trace over the
extra ancilla system. By applying the following inequality from
Proposition~\ref{prop:smooth-dmax-sandwiched-up-bnd}%
\begin{multline}
D_{\max}^{\varepsilon_{2}}(\rho\Vert\sigma)\leq\widetilde{D}_{\beta}(\rho
\Vert\sigma)+\log_{2}(1/\left[  1-\varepsilon_{2}^{2}\right]  )\\
+\frac{1}{\beta-1}\log_{2}(1/\varepsilon_{2}^{2}),
\end{multline}
proved in Appendix~\ref{sec:app-dmaxeps-bnds}, we find that%
\begin{multline}
\log_{2}M\leq m\widetilde{D}_{\beta}(\tau\Vert\omega)+\widetilde{D}_{\beta
}(\tau^{\prime}\Vert\omega^{\prime})+\log_{2}(1/\left[  1-\varepsilon_{2}%
^{2}\right]  )\\
+\frac{1}{\beta-1}\log_{2}(1/\varepsilon_{2}^{2}).
\end{multline}
Now set $\beta>1$ such that%
\begin{equation}
\delta_{3}nD(\tau\Vert\omega)\geq m\left[  \widetilde{D}_{\beta}(\tau
\Vert\omega)-D(\tau\Vert\omega)\right]  ,
\end{equation}
which is possible due to \eqref{eq:app-sandwiched-Renyi-a-1-rel-ent} and
\eqref{eq:app-sandwiched-Renyi-ordered}, and for this choice of $\beta$, take
$n$ sufficiently large so that%
\begin{multline}
\delta_{4}\cdot D(\tau\Vert\omega)\geq\frac{1}{n}\widetilde{D}_{\beta}%
(\tau^{\prime}\Vert\omega^{\prime})+\frac{1}{n}\log_{2}(1/\left[
1-\varepsilon_{2}^{2}\right]  )\label{eq:app-large-enough-n-2}\\
+\frac{1}{n\left(  \beta-1\right)  }\log_{2}(1/\varepsilon_{2}^{2}).
\end{multline}
(Note that we require $n$ large enough so that both
\eqref{eq:app-large-enough-n-1}\ and \eqref{eq:app-large-enough-n-2}\ hold.)
Then we have that%
\begin{equation}
\log_{2}M\leq mD(\tau\Vert\omega)+nD(\tau\Vert\omega)\left[  \delta_{3}%
+\delta_{4}\right]  .\label{eq:app-up-distill-dilute-tech}%
\end{equation}
Putting together \eqref{eq:app-low-distill-dilute-tech}
and\ \eqref{eq:app-up-distill-dilute-tech}, we find that%
\begin{multline}
nD(\rho\Vert\sigma)-nD(\tau\Vert\omega)\left[  \delta_{1}+\delta_{2}\right]
\\
\leq mD(\tau\Vert\omega)+nD(\tau\Vert\omega)\left[  \delta_{3}+\delta
_{4}\right]  .
\end{multline}
Now dividing both sides by $nD(\tau\Vert\omega)$, we find that%
\begin{align}
\frac{m}{n} &  \geq\frac{D(\rho\Vert\sigma)}{D(\tau\Vert\omega)}-\left[
\delta_{1}+\delta_{2}+\delta_{3}+\delta_{4}\right]  .\\
&  =\frac{D(\rho\Vert\sigma)}{D(\tau\Vert\omega)}-\delta.
\end{align}
The rate of this scheme is equal to $m/n$. The error of the protocol is no
larger then $\varepsilon_{1}+\varepsilon_{2}=\varepsilon$, following from an
application of the triangle inequality.

Thus, we have shown that for all $\varepsilon\in(0,1)$, $\delta>0$, there
exists an $(n,n\left[  R-\delta\right]  ,\varepsilon)$ box transformation
protocol with $R=\frac{D(\rho\Vert\sigma)}{D(\tau\Vert\omega)}$, concluding
the proof of the achievability part.

\subsection{Strong converse via sandwiched R\'{e}nyi relative entropy}

Before proving the strong converse, we establish the following lemma as a
generalization of \cite[Proposition~2.8]{LWD16}. In fact, the proof of the
following lemma is contained in the proof of \cite[Proposition~2.8]{LWD16}.
The following lemma serves as a pseudo-continuity inequality for the
sandwiched R\'enyi relative entropies.

\begin{lemma}
\label{lem:app-continuity-sandwiched-Renyi}Let $\rho_{0}$, $\rho_{1}$, and
$\sigma$ be quantum states such that $\operatorname{supp}(\rho_{0}%
)\subseteq\operatorname{supp}(\sigma)$. Fix $\alpha\in(1/2,1)$ and
$\beta\equiv\beta(\alpha):=\alpha/\left(  2\alpha-1\right)  >1$. Then%
\begin{equation}
\widetilde{D}_{\beta}(\rho_{0}\Vert\sigma)-\widetilde{D}_{\alpha}(\rho
_{1}\Vert\sigma)\geq\frac{\alpha}{1-\alpha}\log_2 F(\rho_{0},\rho_{1}).
\end{equation}

\end{lemma}

\begin{proof}
Consider that%
\begin{align}
&  \widetilde{D}_{\beta}(\rho_{0}\Vert\sigma)-\widetilde{D}_{\alpha}(\rho
_{1}\Vert\sigma)\nonumber\\
&  =\frac{2\beta}{\beta-1}\log_{2}\left\Vert \rho_{0}^{1/2}\sigma^{\left(
1-\beta\right)  /2\beta}\right\Vert _{2\beta}\nonumber\\
&  \qquad-\frac{2\alpha}{\alpha-1}\log_{2}\left\Vert \sigma^{\left(
1-\alpha\right)  /2\alpha}\rho_{1}^{1/2}\right\Vert _{2\alpha}\\
&  =\frac{2\alpha}{1-\alpha}\log_{2}\left\Vert \rho_{0}^{1/2}\sigma^{\left(
1-\beta\right)  /2\beta}\right\Vert _{2\beta}\nonumber\\
&  \qquad+\frac{2\alpha}{1-\alpha}\log_{2}\left\Vert \sigma^{\left(
1-\alpha\right)  /2\alpha}\rho_{1}^{1/2}\right\Vert _{2\alpha}\\
&  =\frac{2\alpha}{1-\alpha}\log_{2}\left[  \left\Vert \rho_{0}^{1/2}%
\sigma^{\left(  1-\beta\right)  /2\beta}\right\Vert _{2\beta}\left\Vert
\sigma^{\left(  1-\alpha\right)  /2\alpha}\rho_{1}^{1/2}\right\Vert _{2\alpha
}\right] \\
&  \geq\frac{2\alpha}{1-\alpha}\log_{2}\left\Vert \rho_{0}^{1/2}%
\sigma^{\left(  1-\beta\right)  /2\beta}\sigma^{\left(  1-\alpha\right)
/2\alpha}\rho_{1}^{1/2}\right\Vert _{1}\\
&  =\frac{2\alpha}{1-\alpha}\log_{2}\left\Vert \rho_{0}^{1/2}\rho_{1}%
^{1/2}\right\Vert _{1}\\
&  =\frac{\alpha}{1-\alpha}\log_{2}F(\rho_{0},\rho_{1}).
\end{align}
The sole inequality follows from the H\"{o}lder inequality.
\end{proof}

\bigskip

The following is an auxiliary lemma that serves as a one-shot converse for any
approximate box transformation $(\rho,\sigma)\ \underrightarrow{\mathcal{N}%
}\ (\tau,\omega)$ where $\omega=\mathcal{N}(\sigma)$:

\begin{lemma}
Let $\rho$, $\sigma$, $\tau$, and $\omega$ be quantum states and $\mathcal{N}$
a quantum channel such that $\mathcal{N}(\sigma)=\omega$. Then for $\alpha
\in(1/2,1)$ and $\beta\equiv\beta(\alpha):=\alpha/\left(  2\alpha-1\right)  $,
we have that%
\begin{equation}
\widetilde{D}_{\beta}(\rho\Vert\sigma)\geq\widetilde{D}_{\alpha}(\tau
\Vert\omega)+\frac{\alpha}{1-\alpha}\log_{2}F(\mathcal{N}(\rho),\tau).
\end{equation}

\end{lemma}

\begin{proof}
Consider that%
\begin{align}
\widetilde{D}_{\beta}(\rho\Vert\sigma)  &  \geq\widetilde{D}_{\beta
}(\mathcal{N}(\rho)\Vert\mathcal{N}(\sigma))\\
&  =\widetilde{D}_{\beta}(\mathcal{N}(\rho)\Vert\omega)\\
&  \geq\widetilde{D}_{\alpha}(\tau\Vert\omega)+\frac{\alpha}{1-\alpha}\log
_{2}F(\mathcal{N}(\rho),\tau).
\end{align}
The first inequality follows from the quantum data processing inequality in
\eqref{eq:app-qdp-sandwiched-Renyi}\ and the other from
Lemma~\ref{lem:app-continuity-sandwiched-Renyi}.
\end{proof}

\begin{proposition}
\label{prop:strong-converse-exp}Let $n,m\in\mathbb{Z}^{+}$ and $\varepsilon
\in\lbrack0,1)$. Let $\rho$, $\sigma$, $\tau$, and $\omega$ be quantum states
and $\mathcal{N}^{(n)}$ a quantum channel constituting an $(n,m,\varepsilon
)$\ box transformation protocol (i.e., so that $\mathcal{N}^{(n)}%
(\rho^{\otimes n})\approx_{\varepsilon}\tau^{\otimes m}$ and $\mathcal{N}%
^{(n)}(\sigma^{\otimes n})=\omega^{\otimes m}$).\ Then for $\alpha\in(1/2,1)$
and $\beta\equiv\beta(\alpha):=\alpha/\left(  2\alpha-1\right)  $, we have
that%
\begin{equation}
\frac{\widetilde{D}_{\beta}(\rho\Vert\sigma)}{\widetilde{D}_{\alpha}(\tau
\Vert\omega)}\geq\frac{m}{n}+\frac{2\alpha}{n\left(  1-\alpha\right)
\widetilde{D}_{\alpha}(\tau\Vert\omega)}\log_{2}(1-\varepsilon).
\end{equation}
Alternatively, if we set $R=m/n$, then the above bound can be written as%
\begin{align}
&  -\frac{1}{n}\log_{2}(1-\varepsilon)\nonumber\\
&  \geq\frac{1}{2}\left(  \frac{1-\alpha}{\alpha}\right)  \left(
R\ \widetilde{D}_{\alpha}(\tau\Vert\omega)-\widetilde{D}_{\beta}(\rho
\Vert\sigma)\right) \\
&  =\frac{1}{2}\left(  \frac{\beta-1}{\beta}\right)  \left(  R\ \widetilde
{D}_{\alpha}(\tau\Vert\omega)-\widetilde{D}_{\beta}(\rho\Vert\sigma)\right)  .
\end{align}

\end{proposition}

\begin{proof}
Consider that%
\begin{align}
&  n\widetilde{D}_{\beta}(\rho\Vert\sigma)\nonumber\\
&  =\widetilde{D}_{\beta}(\rho^{\otimes n}\Vert\sigma^{\otimes n})\\
&  \geq\widetilde{D}_{\alpha}(\tau^{\otimes m}\Vert\omega^{\otimes m}%
)+\frac{\alpha}{1-\alpha}\log_{2}F(\mathcal{N}^{(n)}(\rho^{\otimes n}%
),\tau^{\otimes m})\\
&  =m\widetilde{D}_{\alpha}(\tau\Vert\omega)+\frac{\alpha}{1-\alpha}\log
_{2}F(\mathcal{N}^{(n)}(\rho^{\otimes n}),\tau^{\otimes m})\\
&  =m\widetilde{D}_{\alpha}(\tau\Vert\omega)+\frac{2\alpha}{1-\alpha}\log
_{2}\sqrt{F}(\mathcal{N}^{(n)}(\rho^{\otimes n}),\tau^{\otimes m})\\
&  \geq m\widetilde{D}_{\alpha}(\tau\Vert\omega)+\frac{2\alpha}{1-\alpha}%
\log_{2}(1-\varepsilon),
\end{align}
where to get the last inequality, we used the fact that \cite{FG98}%
\begin{equation}
\frac{1}{2}\left\Vert \rho_{0}-\rho_{1}\right\Vert _{1}\geq1-\sqrt{F}(\rho
_{0},\rho_{1}).
\end{equation}
Dividing by $n$, we find that%
\begin{equation}
\widetilde{D}_{\beta}(\rho\Vert\sigma)\geq\frac{m}{n}\widetilde{D}_{\alpha
}(\tau\Vert\omega)+\frac{2\alpha}{n\left(  1-\alpha\right)  }\log
_{2}(1-\varepsilon),
\end{equation}
which concludes the proof.
\end{proof}

\bigskip

We now give a proof for the strong converse statement in
\eqref{eq:box-trans-main-result}. Our proof is related to the approach from
\cite{KW19}. Fix $\varepsilon\in\lbrack0,1)$ and $\delta>0$. We need to show
that there is an $n$ large enough such that there does not exist an
$(n,n[R+\delta],\varepsilon)$ box transformation protocol, with $R$ set as
follows:%
\begin{equation}
R=\frac{D(\rho\Vert\sigma)}{D(\tau\Vert\omega)}.
\end{equation}

From Proposition~\ref{prop:strong-converse-exp}, the following bound holds for
an arbitrary $(n,m,\varepsilon)$ protocol, $\alpha\in(1/2,1)$, and
$\beta\equiv\beta(\alpha):=\alpha/\left(  2\alpha-1\right)  $:%
\begin{equation}
n\widetilde{D}_{\beta}(\rho\Vert\sigma)+\frac{2\alpha}{1-\alpha}\log
_{2}(1/[1-\varepsilon])\geq m\widetilde{D}_{\alpha}(\tau\Vert\omega).
\label{eq:app-strong-conv-steps-1}%
\end{equation}
Set $\delta_{2}$ such that $0<\delta_{2}<\delta D(\tau\Vert\omega)$. Then set
$\delta_{1}>0$ such that the following equation is satisfied%
\begin{equation}
\frac{D(\rho\Vert\sigma)+\delta_{1}+\delta_{2}}{D(\tau\Vert\omega)-\delta_{1}%
}=\frac{D(\rho\Vert\sigma)}{D(\tau\Vert\omega)}+\delta,
\end{equation}
i.e.,%
\begin{equation}
\delta_{1}=\frac{D(\tau\Vert\omega)\left[  \delta D(\tau\Vert\omega
)-\delta_{2}\right]  }{D(\rho\Vert\sigma)+D(\tau\Vert\omega)\left[
1+\delta\right]  }.
\end{equation}
Set $\alpha\in(1/2,1)$ such that%
\begin{equation}
\delta_{1}>\max\{D(\tau\Vert\omega)-\widetilde{D}_{\alpha}(\tau\Vert
\omega),\widetilde{D}_{\beta}(\rho\Vert\sigma)-D(\rho\Vert\sigma)\},
\end{equation}
which is possible due to \eqref{eq:conv-PR-q-rel-ent},
\eqref{eq:app-Petz-Renyi-ordered},
\eqref{eq:app-sandwiched-Renyi-a-1-rel-ent},
\eqref{eq:app-sandwiched-Renyi-ordered}, and the fact that $\beta
=\alpha/\left(  2\alpha-1\right)  $, and for this choice of $\alpha$, pick $n$
large enough so that%
\begin{equation}
\delta_{2}>\frac{2\alpha}{n(1-\alpha)}\log_{2}(1/[1-\varepsilon]).
\label{eq:app-delta-2ineq-SC}%
\end{equation}
For these choices, we then have that%
\begin{multline}
\widetilde{D}_{\beta}(\rho\Vert\sigma)+\frac{2\alpha}{n(1-\alpha)}\log
_{2}(1/[1-\varepsilon])\\
<D(\rho\Vert\sigma)+\delta_{1}+\delta_{2},
\end{multline}
and we also have that%
\begin{equation}
\frac{m}{n}\widetilde{D}_{\alpha}(\tau\Vert\omega)>\frac{m}{n}\left[
D(\tau\Vert\omega)-\delta_{1}\right]  .
\end{equation}
Putting these inequalities together, we find that%
\begin{equation}
\frac{m}{n}<\frac{D(\rho\Vert\sigma)+\delta_{1}+\delta_{2}}{D(\tau\Vert
\omega)-\delta_{1}}=\frac{D(\rho\Vert\sigma)}{D(\tau\Vert\omega)}+\delta.
\label{eq:app-strong-conv-steps-last}%
\end{equation}
Thus, the rate of the protocol $\frac{m}{n}$ is strictly less than
$\frac{D(\rho\Vert\sigma)}{D(\tau\Vert\omega)}+\delta$, so that an
$(n,n[R+\delta],\varepsilon)$ box transformation protocol cannot exist for the
choice of $n$ taken nor any $n$ larger than that (for the latter statement,
note that \eqref{eq:app-delta-2ineq-SC}\ still holds for larger $n$).

\subsection{Strong converse via Petz--R\'{e}nyi relative entropy}

We now discuss an alternative proof of the strong converse by going through
the Petz--R\'{e}nyi relative entropy. We begin with a pseudo-continuity
inequality for the Petz--R\'{e}nyi relative entropy. The proof of
Lemma~\ref{lem:petz-renyi-continuity} below follows the spirit of the proof of
\cite[Proposition~2.8]{LWD16}, but this time some steps are different.

\begin{lemma}
\label{lem:petz-renyi-continuity}Let $\rho_{0}$, $\rho_{1}$, and $\sigma$ be
quantum states such that $\operatorname{supp}(\rho_{0})\subseteq
\operatorname{supp}(\sigma)$. Fix $\alpha\in(0,1)$ and $\beta\equiv
\beta(\alpha):=2-\alpha\in\left(  1,2\right)  $. Then%
\begin{equation}
D_{\beta}(\rho_{0}\Vert\sigma)-D_{\alpha}(\rho_{1}\Vert\sigma)\geq\frac
{2}{1-\alpha}\log_{2}\!\left[  1-\frac{1}{2}\left\Vert \rho_{0}-\rho
_{1}\right\Vert _{1}\right]  .
\end{equation}

\end{lemma}

\begin{proof}
Consider that $\alpha-1=1-\beta$, so that%
\begin{align}
&  D_{\beta}(\rho_{0}\Vert\sigma)-D_{\alpha}(\rho_{1}\Vert\sigma)\nonumber\\
&  =\frac{1}{\beta-1}\log_{2}\operatorname{Tr}[\rho_{0}^{\beta}\sigma
^{1-\beta}]-\frac{1}{\alpha-1}\log_{2}\operatorname{Tr}[\rho_{1}^{\alpha
}\sigma^{1-\alpha}]\\
&  =\frac{1}{\beta-1}\log_{2}\operatorname{Tr}[\rho_{0}^{\beta}\sigma
^{1-\beta}]+\frac{1}{\beta-1}\log_{2}\operatorname{Tr}[\rho_{1}^{\alpha}%
\sigma^{1-\alpha}]\\
&  =\frac{1}{\beta-1}\log_{2}\left(  \operatorname{Tr}[\rho_{0}^{\beta}%
\sigma^{1-\beta}]\operatorname{Tr}[\rho_{1}^{\alpha}\sigma^{1-\alpha}]\right)
\\
&  =\frac{1}{\beta-1}\log_{2}\left(  \left\Vert \rho_{0}^{\beta/2}%
\sigma^{\left(  1-\beta\right)  /2}\right\Vert _{2}^{2}\left\Vert
\sigma^{\left(  1-\alpha\right)  /2}\rho_{1}^{\alpha/2}\right\Vert _{2}%
^{2}\right) \\
&  \geq\frac{1}{\beta-1}\log_{2}\left\Vert \rho_{0}^{\beta/2}\sigma^{\left(
1-\beta\right)  /2}\sigma^{\left(  1-\alpha\right)  /2}\rho_{1}^{\alpha
/2}\right\Vert _{1}^{2}\\
&  =\frac{2}{\beta-1}\log_{2}\left\Vert \rho_{0}^{\beta/2}\rho_{1}^{\alpha
/2}\right\Vert _{1}\\
&  \geq\frac{2}{\beta-1}\log_{2}\operatorname{Tr}[\rho_{0}^{\beta/2}\rho
_{1}^{\alpha/2}]\\
&  =\frac{2}{\beta-1}\log_{2}\operatorname{Tr}[\rho_{0}^{\beta/2}\rho
_{1}^{\left(  2-\beta\right)  /2}]\\
&  =\frac{2}{\beta-1}\log_{2}\operatorname{Tr}[\rho_{0}^{\beta/2}\rho
_{1}^{1-\beta/2}]\\
&  \geq\frac{2}{\beta-1}\log_{2}\left[  1-\frac{1}{2}\left\Vert \rho_{0}%
-\rho_{1}\right\Vert _{1}\right] \\
&  =\frac{2}{1-\alpha}\log_{2}\left[  1-\frac{1}{2}\left\Vert \rho_{0}%
-\rho_{1}\right\Vert _{1}\right]  .
\end{align}
The fourth equality follows from a rewriting of the Petz--R\'enyi relative
entropy in terms of the Schatten 2-norm, as given in \cite[Eq.~(3.10)]%
{DW2016}. The first inequality follows from an application of the
Cauchy--Schwarz inequality. The second inequality follows from the variational
characterization of the trace norm as $\left\Vert A\right\Vert _{1}=\sup
_{U}\left\vert \operatorname{Tr}[AU]\right\vert $, where the optimization is
over all unitaries and we pick $U=I$ to get the inequality. The last
inequality follows from \cite[Theorem~1]{ACMBMAV07} and because $\beta
/2\in(1/2,1)$.
\end{proof}

\begin{remark}
We note here that the bound from Lemma~\ref{lem:petz-renyi-continuity} can be
used to obtain pseudo-continuity bounds for information quantities derived
from the Petz--R\'enyi relative entropy, such as mutual information and
conditional entropy, much like what is done in \cite[Proposition~2.8]{LWD16}.
\end{remark}

\bigskip

The following is another auxiliary lemma that serves as a one-shot converse
for any approximate box transformation $(\rho,\sigma)\ \underrightarrow
{\mathcal{N}}\ (\tau,\omega)$ where $\omega=\mathcal{N}(\sigma)$:

\begin{lemma}
Let $\rho$, $\sigma$, $\tau$, and $\omega$ be quantum states and $\mathcal{N}$
a quantum channel such that $\mathcal{N}(\sigma)=\omega$. Then for $\alpha
\in(0,1)$ and $\beta\equiv\beta(\alpha):=2-\alpha$, we have that%
\begin{equation}
D_{\beta}(\rho\Vert\sigma)\geq D_{\alpha}(\tau\Vert\omega)+\frac{2}{1-\alpha
}\log_{2}\left[  1-\frac{1}{2}\left\Vert \mathcal{N}(\rho)-\tau\right\Vert
_{1}\right]  .
\end{equation}

\end{lemma}

\begin{proof}
Consider that%
\begin{align}
&  D_{\beta}(\rho\Vert\sigma)\nonumber\\
&  \geq D_{\beta}(\mathcal{N}(\rho)\Vert\mathcal{N}(\sigma))\\
&  =D_{\beta}(\mathcal{N}(\rho)\Vert\omega)\\
&  \geq D_{\alpha}(\tau\Vert\omega)+\frac{2}{1-\alpha}\log_{2}\left[
1-\frac{1}{2}\left\Vert \mathcal{N}(\rho)-\tau\right\Vert _{1}\right]  .
\end{align}
The first inequality follows from the quantum data processing inequality in
\eqref{eq:DP-Petz-Renyi}\ and the other from
Lemma~\ref{lem:petz-renyi-continuity}.
\end{proof}

\begin{proposition}
\label{prop:strong-converse-exp-Petz-Renyi}Let $n,m\in\mathbb{Z}^{+}$ and
$\varepsilon\in\left[  0,1\right]  $. Let $\rho$, $\sigma$, $\tau$, and
$\omega$ be quantum states and $\mathcal{N}^{(n)}$ a quantum channel
constituting an $(n,m,\varepsilon)$\ box transformation protocol (i.e., so
that $\mathcal{N}^{(n)}(\rho^{\otimes n})\approx_{\varepsilon}\tau^{\otimes
m}$ and $\mathcal{N}^{(n)}(\sigma^{\otimes n})=\omega^{\otimes m}$).\ Then for
$\alpha\in(0,1)$ and $\beta\equiv\beta(\alpha):=2-\alpha$, we have that%
\begin{equation}
\frac{D_{\beta}(\rho\Vert\sigma)}{D_{\alpha}(\tau\Vert\omega)}\geq\frac{m}%
{n}+\frac{2}{n\left(  1-\alpha\right)  D_{\alpha}(\tau\Vert\omega)}\log
_{2}(1-\varepsilon).
\end{equation}
Alternatively, if we set $R=m/n$, then the above bound can be written as%
\begin{align}
&  -\frac{1}{n}\log_{2}(1-\varepsilon)\nonumber\\
&  \geq\left(  \frac{1-\alpha}{2}\right)  \left(  R\ D_{\alpha}(\tau
\Vert\omega)-D_{\beta}(\rho\Vert\sigma)\right) \\
&  =\left(  \frac{\beta-1}{2}\right)  \left(  R\ D_{\alpha}(\tau\Vert
\omega)-D_{\beta}(\rho\Vert\sigma)\right)  .
\end{align}

\end{proposition}

\begin{proof}
Consider that%
\begin{align}
&  nD_{\beta}(\rho\Vert\sigma)\nonumber\\
&  =D_{\beta}(\rho^{\otimes n}\Vert\sigma^{\otimes n})\\
&  \geq D_{\alpha}(\tau^{\otimes m}\Vert\omega^{\otimes m})\nonumber\\
&  \qquad+\frac{2}{1-\alpha}\log_{2}\left[  1-\frac{1}{2}\left\Vert
\mathcal{N}^{(n)}(\rho^{\otimes n})-\tau^{\otimes m}\right\Vert _{1}\right] \\
&  \geq mD_{\alpha}(\tau\Vert\omega)+\frac{2}{1-\alpha}\log_{2}(1-\varepsilon
),
\end{align}
Dividing by $n$, we find that%
\begin{equation}
D_{\beta}(\rho\Vert\sigma)\geq\frac{m}{n}D_{\alpha}(\tau\Vert\omega)+\frac
{2}{n\left(  1-\alpha\right)  }\log_{2}(1-\varepsilon),
\end{equation}
which concludes the proof.
\end{proof}

\bigskip

We note here that one could arrive at the strong converse statement by going
through steps similar to those in
\eqref{eq:app-strong-conv-steps-1}--\eqref{eq:app-strong-conv-steps-last}, but
using Proposition~\ref{prop:strong-converse-exp-Petz-Renyi} instead.

\section{Bounding the smooth max-relative entropy with quantum relative
entropies}

\label{sec:app-dmaxeps-bnds}In this appendix, we establish lower and upper
bounds for the smooth max-relative entropy in terms of the R\'{e}nyi relative
entropies. We begin with the following lower bound:

\begin{proposition}
Let $\rho$ and $\sigma$ be quantum states. The following bound holds for all
$\alpha\in\left[  1/2,1\right)  $ and $\varepsilon\in\lbrack0,1)$:%
\begin{equation}
D_{\max}^{\varepsilon}(\rho\Vert\sigma)\geq\widetilde{D}_{\alpha}(\rho
\Vert\sigma)+\frac{2\alpha}{\alpha-1}\log_{2}\!\left(  \frac{1}{1-\varepsilon
}\right)  . \label{eq:app-lower-dmaxeps-da}%
\end{equation}

\end{proposition}

\begin{proof}
First fix $\alpha\in\left(  1/2,1\right)  $. Let $\widetilde{\rho}$ be a state
such that $\frac{1}{2}\left\Vert \widetilde{\rho}-\rho\right\Vert _{1}%
\leq\varepsilon$. Then for $\beta\equiv\beta(\alpha):=\alpha/\left(
2\alpha-1\right)  $, we find that%
\begin{align}
D_{\max}(\widetilde{\rho}\Vert\sigma)  &  \geq\widetilde{D}_{\beta}%
(\widetilde{\rho}\Vert\sigma)\\
&  \geq\widetilde{D}_{\alpha}(\rho\Vert\sigma)+\frac{\alpha}{1-\alpha}\log
_{2}F(\widetilde{\rho},\rho)\\
&  =\widetilde{D}_{\alpha}(\rho\Vert\sigma)+\frac{2\alpha}{1-\alpha}\log
_{2}\sqrt{F}(\widetilde{\rho},\rho)\\
&  \geq\widetilde{D}_{\alpha}(\rho\Vert\sigma)+\frac{2\alpha}{1-\alpha}%
\log_{2}(1-\varepsilon).
\end{align}
The first inequality follows from \eqref{eq:app-dmax-limit} and
\eqref{eq:app-sandwiched-Renyi-ordered}. The second inequality follows from
Lemma~\ref{lem:app-continuity-sandwiched-Renyi}. The final inequality follows
because \cite{FG98}%
\begin{equation}
1-\sqrt{F(\widetilde{\rho},\rho)}\leq\frac{1}{2}\left\Vert \widetilde{\rho
}-\rho\right\Vert _{1}.
\end{equation}
Since the bound holds for an arbitrary $\widetilde{\rho}$ satisfying $\frac
{1}{2}\left\Vert \widetilde{\rho}-\rho\right\Vert _{1}\leq\varepsilon$, we
conclude \eqref{eq:app-lower-dmaxeps-da}.

The inequality in \eqref{eq:app-lower-dmaxeps-da} for $\alpha=1/2$ follows
since \eqref{eq:app-lower-dmaxeps-da} holds for all $\alpha\in\left(
1/2,1\right)  $ and by taking the limit as $\alpha\to1/2$.
\end{proof}

\bigskip Another lower bound on the smooth max-relative entropy is as follows:

\begin{proposition}
Let $\rho$ and $\sigma$ be quantum states. The following bound holds for all
$\alpha\in\lbrack0,1)$ and $\varepsilon\in\lbrack0,1)$:%
\begin{equation}
D_{\max}^{\varepsilon}(\rho\Vert\sigma)\geq D_{\alpha}(\rho\Vert\sigma
)+\frac{2}{\alpha-1}\log_{2}\!\left(  \frac{1}{1-\varepsilon}\right)  .
\label{eq:app-lower-dmaxeps-dapetz}%
\end{equation}

\end{proposition}

\begin{proof}
First fix $\alpha\in\left(  0,1\right)  $. Let $\widetilde{\rho}$ be a state
such that $\frac{1}{2}\left\Vert \widetilde{\rho}-\rho\right\Vert _{1}%
\leq\varepsilon$. Then for $\beta\equiv\beta(\alpha):=2-\alpha$, we find that%
\begin{align}
&  D_{\max}(\widetilde{\rho}\Vert\sigma)\nonumber\\
&  \geq D_{\beta}(\widetilde{\rho}\Vert\sigma)\\
&  \geq D_{\alpha}(\rho\Vert\sigma)+\frac{2}{1-\alpha}\log_{2}\!\left[
1-\frac{1}{2}\left\Vert \widetilde{\rho}-\rho\right\Vert _{1}\right]
\label{eq:app-smoothmaxeps-petz-renyi-1}\\
&  \geq D_{\alpha}(\rho\Vert\sigma)+\frac{2}{1-\alpha}\log_{2}(1-\varepsilon).
\label{eq:app-smoothmaxeps-petz-renyi-2}%
\end{align}
The first inequality follows from \eqref{eq:app-Petz-Renyi-ordered} and
\cite[Eqs.~(43)--(46)]{JRSWW16}, the latter of which we repeat below:%
\begin{align}
D_{2}(\widetilde{\rho}\Vert\sigma)  &  =\log_{2}\operatorname{Tr}%
[\widetilde{\rho}^{2}\sigma^{-1}]\\
&  =\log_{2}\operatorname{Tr}[\widetilde{\rho}\widetilde{\rho}^{1/2}%
\sigma^{-1}\widetilde{\rho}^{1/2}]\\
&  \leq\log_{2}\sup_{\tau}\operatorname{Tr}[\tau\widetilde{\rho}^{1/2}%
\sigma^{-1}\widetilde{\rho}^{1/2}]\\
&  =\log_{2}\left\Vert \widetilde{\rho}^{1/2}\sigma^{-1}\widetilde{\rho}%
^{1/2}\right\Vert _{\infty}\\
&  =D_{\max}(\widetilde{\rho}\Vert\sigma).
\end{align}
Note that the optimization above is over quantum states~$\tau$. The second
inequality in \eqref{eq:app-smoothmaxeps-petz-renyi-1}\ follows from
Lemma~\ref{lem:petz-renyi-continuity}. Since the bound holds for an arbitrary
state~$\widetilde{\rho}$ satisfying $\frac{1}{2}\left\Vert \widetilde{\rho
}-\rho\right\Vert _{1}\leq\varepsilon$, we conclude \eqref{eq:app-lower-dmaxeps-dapetz}.

The inequality in \eqref{eq:app-lower-dmaxeps-dapetz} for $\alpha=0$ follows
since \eqref{eq:app-lower-dmaxeps-dapetz} holds for all $\alpha\in\left(
0,1\right)  $ and by taking the limit as $\alpha\rightarrow0$.
\end{proof}

\bigskip

We now give some upper bounds on the smooth max-relative entropy in terms of
the quantum relative entropy and the sandwiched R\'{e}nyi relative entropy.
The method for doing so follows the proof approach of \cite[Theorem~1]{JN12}
very closely. The upper bound in
Proposition~\ref{prop:smoothdmax-up-bnd-rel-ent} is very similar to
\cite[Theorem~1]{JN12}, but it is expressed in terms of quantum relative
entropy rather than observational divergence.

\begin{proposition}
\label{prop:smoothdmax-up-bnd-rel-ent}Given states $\rho$ and $\sigma$, the
following bound holds for all $\varepsilon\in\left(  0,1\right)  $:%
\begin{multline}
D_{\max}^{\varepsilon}(\rho\Vert\sigma)\leq\frac{1}{\varepsilon^{2}}\left[
D(\rho\Vert\sigma)+\frac{1}{2\ln2}\left\Vert \rho-\sigma\right\Vert
_{1}\right] \\
+\log_{2}\!\left(  \frac{1}{1-\varepsilon^{2}}\right)  .
\end{multline}

\end{proposition}

\begin{proof}
The statement is trivially true if $\rho=\sigma$ or if $\operatorname{supp}%
(\rho)\not \subseteq \operatorname{supp}(\sigma)$. So going forward, we assume
that $\rho\neq\sigma$ and $\operatorname{supp}(\rho)\subseteq
\operatorname{supp}(\sigma)$. The SDP\ dual of $D_{\max}(\tau\Vert\omega)$ is
given by%
\begin{equation}
D_{\max}(\tau\Vert\omega)=\log_{2}\sup_{\Lambda \geq 0}\left\{  \operatorname{Tr}%
[\Lambda\tau]:\operatorname{Tr}[\Lambda\omega]\leq 1\right\}  ,
\label{eq:app-SDP-dual-dmax}%
\end{equation}
implying that%
\begin{equation}
D_{\max}^{\varepsilon}(\rho\Vert\sigma)=\log_{2}\inf_{\widetilde{\rho}%
:\frac{1}{2}\left\Vert \widetilde{\rho}-\rho\right\Vert _{1}\leq\varepsilon
}\sup_{\substack{\Lambda\ :\ \Lambda\geq0,\\\operatorname{Tr}[\Lambda
\sigma]\leq1}}\operatorname{Tr}[\Lambda\widetilde{\rho}].
\end{equation}
Since the objective function $\operatorname{Tr}[\Lambda\widetilde{\rho}]$ is
linear in $\Lambda$ and $\widetilde{\rho}$, the set $\{\Lambda:\Lambda
\geq0,\operatorname{Tr}[\Lambda\sigma]\leq1\}$ is compact and concave, and the
set%
\begin{equation}
\left\{  \widetilde{\rho}:\frac{1}{2}\left\Vert \widetilde{\rho}%
-\rho\right\Vert _{1}\leq\varepsilon,\ \widetilde{\rho}\geq
0,\ \operatorname{Tr}[\widetilde{\rho}]=1\right\}
\end{equation}
is compact and convex (due to convexity of normalized trace distance), the
minimax theorem applies and we find that%
\begin{equation}
D_{\max}^{\varepsilon}(\rho\Vert\sigma)=\log_{2} \sup_{\substack{\Lambda
\ :\ \Lambda\geq0,\\\operatorname{Tr}[\Lambda\sigma]\leq1}}\inf_{\widetilde
{\rho}:\frac{1}{2}\left\Vert \widetilde{\rho}-\rho\right\Vert _{1}%
\leq\varepsilon}\operatorname{Tr}[\Lambda\widetilde{\rho}].
\end{equation}
For a fixed operator $\Lambda\geq0$ with spectral decomposition%
\begin{equation}
\Lambda=\sum_{i}\lambda_{i}|\phi_{i}\rangle\langle\phi_{i}|,
\end{equation}
let us define the following set, for a choice of $\lambda>0$ to be specified
later:%
\begin{equation}
\mathcal{S}:=\left\{  i:\langle\phi_{i}|\rho|\phi_{i}\rangle>2^{\lambda
}\langle\phi_{i}|\sigma|\phi_{i}\rangle\right\}  .
\end{equation}
Let%
\begin{equation}
\Pi=\sum_{i\in\mathcal{S}}|\phi_{i}\rangle\langle\phi_{i}|.
\label{eq:app-proj-dmaxeps-up-bnd}%
\end{equation}
Then from the definition, we find that%
\begin{equation}
\operatorname{Tr}[\Pi\rho]>2^{\lambda}\operatorname{Tr}[\Pi\sigma],
\end{equation}
which implies that%
\begin{equation}
\frac{\operatorname{Tr}[\Pi\rho]}{\operatorname{Tr}[\Pi\sigma]}>2^{\lambda}.
\label{eq:app-ratio-probs-ineq}%
\end{equation}
Now consider from the data processing inequality under the channel%
\begin{equation}
\Delta(\omega):=\operatorname{Tr}[\Pi\omega]|0\rangle\langle
0|+\operatorname{Tr}[\hat{\Pi}\omega]|1\rangle\langle1|,
\end{equation}
where%
\begin{equation}
\hat{\Pi}:=I-\Pi,
\end{equation}
that%
\begin{align}
&  D(\rho\Vert\sigma)\nonumber\\
&  \geq D(\Delta(\rho)\Vert\Delta(\sigma))\nonumber\\
&  =\operatorname{Tr}[\Pi\rho]\log_{2}\!\left(  \frac{\operatorname{Tr}%
[\Pi\rho]}{\operatorname{Tr}[\Pi\sigma]}\right)  +\operatorname{Tr}[\hat{\Pi
}\rho]\log_{2}\!\left(  \frac{\operatorname{Tr}[\hat{\Pi}\rho]}%
{\operatorname{Tr}[\hat{\Pi}\sigma]}\right) \nonumber\\
&  =\operatorname{Tr}[\Pi\rho]\log_{2}\!\left(  \frac{\operatorname{Tr}%
[\Pi\rho]}{\operatorname{Tr}[\Pi\sigma]}\right)  +\frac{1}{\ln2}\left(
\operatorname{Tr}[\Pi\sigma]-\operatorname{Tr}[\Pi\rho]\right) \nonumber\\
&  \quad+\operatorname{Tr}[\hat{\Pi}\rho]\log_{2}\!\left(  \frac
{\operatorname{Tr}[\hat{\Pi}\rho]}{\operatorname{Tr}[\hat{\Pi}\sigma]}\right)
+\frac{1}{\ln2}\left(  \operatorname{Tr}[\hat{\Pi}\sigma]-\operatorname{Tr}%
[\hat{\Pi}\rho]\right) \nonumber\\
&  \geq\operatorname{Tr}[\Pi\rho]\log_{2}\!\left(  \frac{\operatorname{Tr}%
[\Pi\rho]}{\operatorname{Tr}[\Pi\sigma]}\right)  +\frac{1}{\ln2}\left(
\operatorname{Tr}[\Pi\sigma]-\operatorname{Tr}[\Pi\rho]\right) \nonumber\\
&  \geq\lambda\operatorname{Tr}[\Pi\rho]+\frac{1}{\ln2}\left(
\operatorname{Tr}[\Pi\sigma]-\operatorname{Tr}[\Pi\rho]\right)  ,
\end{align}
where the second inequality follows because%
\begin{align}
&  x\log_{2}(x/y)+\frac{1}{\ln2}\left(  y-x\right) \notag \\
&  =\frac{1}{\ln2}\left[  x\ln(x/y)+y-x\right]  \geq0,
\end{align}
for all $x,y\geq0$, and the last inequality follows from
\eqref{eq:app-ratio-probs-ineq}. Then we find that%
\begin{align}
\operatorname{Tr}[\Pi\rho]  &  \leq\lambda^{-1}\left(  D(\rho\Vert
\sigma)+\frac{1}{\ln2}\operatorname{Tr}[\Pi\rho]-\operatorname{Tr}[\Pi
\sigma]\right) \\
&  \leq\lambda^{-1}\left(  D(\rho\Vert\sigma)+\frac{1}{2\ln2}\left\Vert
\rho-\sigma\right\Vert _{1}\right)  .
\end{align}
Pick%
\begin{equation}
\lambda=\frac{1}{\varepsilon^{2}}\left[  D(\rho\Vert\sigma)+\frac{1}{2\ln
2}\left\Vert \rho-\sigma\right\Vert _{1}\right]  ,
\end{equation}
and we conclude from the above that%
\begin{equation}
\operatorname{Tr}[\Pi\rho]\leq\varepsilon^{2}.
\end{equation}
So this means that%
\begin{equation}
\operatorname{Tr}[\hat{\Pi}\rho]\geq1-\varepsilon^{2}.
\label{eq:app-final-steps-1}%
\end{equation}
Thus, the state%
\begin{equation}
\rho^{\prime}:=\frac{\hat{\Pi}\rho\hat{\Pi}}{\operatorname{Tr}[\hat{\Pi}\rho]}
\label{eq:app-rho-prime-state}%
\end{equation}
is such that \cite[Lemma~9.4.1]{W17book}%
\begin{equation}
F(\rho,\rho^{\prime})\geq1-\varepsilon^{2},
\end{equation}
and in turn that \cite{FG98}%
\begin{equation}
\frac{1}{2}\left\Vert \rho-\rho^{\prime}\right\Vert _{1}\leq\varepsilon.
\end{equation}
We also have that%
\begin{equation}
\rho^{\prime}\leq\frac{\hat{\Pi}\rho\hat{\Pi}}{1-\varepsilon^{2}}.
\end{equation}

Now let $\Lambda$ be an arbitrary operator satisfying $\Lambda\geq0$ and
$\operatorname{Tr}[\Lambda\sigma]\leq1$, and let $\Pi$ be the projection
defined in \eqref{eq:app-proj-dmaxeps-up-bnd}\ for this choice of $\Lambda$.
Then we find that%
\begin{align}
\left(  1-\varepsilon^{2}\right)  \operatorname{Tr}[\Lambda\rho^{\prime}]  &
\leq\operatorname{Tr}[\Lambda\hat{\Pi}\rho\hat{\Pi}]\\
&  =\operatorname{Tr}[\hat{\Pi}\Lambda\hat{\Pi}\rho]\\
&  =\sum_{i\notin\mathcal{S}}\lambda_{i}\langle\phi_{i}|\rho|\phi_{i}\rangle\\
&  \leq2^{\lambda}\sum_{i\notin\mathcal{S}}\lambda_{i}\langle\phi_{i}%
|\sigma|\phi_{i}\rangle\\
&  \leq2^{\lambda}\operatorname{Tr}[\Lambda\sigma]\\
&  \leq2^{\lambda}.
\end{align}
Thus, we have found the following uniform bound for any operator $\Lambda$
satisfying $\Lambda\geq0$ and $\operatorname{Tr}[\Lambda\sigma]\leq1$, with
$\rho^{\prime}$ the state in \eqref{eq:app-rho-prime-state} depending on
$\Lambda$ and\ satisfying $\frac{1}{2}\left\Vert \rho-\rho^{\prime}\right\Vert
_{1}\leq\varepsilon$:%
\begin{equation}
\operatorname{Tr}[\Lambda\rho^{\prime}]\leq2^{\lambda+\log_{2}\left(  \frac
{1}{1-\varepsilon^{2}}\right)  }.
\end{equation}
Then it follows that%
\begin{align}
D_{\max}^{\varepsilon}(\rho\Vert\sigma)  &  =\log_{2}\sup_{\substack{\Lambda
\ :\ \Lambda\geq0,\\\operatorname{Tr}[\Lambda\sigma]\leq1}}\inf_{\widetilde
{\rho}:\frac{1}{2}\left\Vert \widetilde{\rho}-\rho\right\Vert _{1}%
\leq\varepsilon}\operatorname{Tr}[\Lambda\widetilde{\rho}]\\
&  \leq\log_{2}\sup_{\substack{\Lambda\ :\ \Lambda\geq0,\\\operatorname{Tr}%
[\Lambda\sigma]\leq1}}\operatorname{Tr}[\Lambda\rho^{\prime}]\\
&  \leq\lambda+\log_{2}\left(  \frac{1}{1-\varepsilon^{2}}\right)  .
\label{eq:app-final-steps-last}%
\end{align}
This concludes the proof.
\end{proof}

\bigskip

The proof of the following proposition follows the same proof approach of
\cite[Theorem~1]{JN12} (as recalled above), but instead employs the sandwiched
R\'{e}nyi relative entropy and its data processing inequality. The following
proposition was also reported recently in \cite{ABJT19}:

\begin{proposition}
\label{prop:smooth-dmax-sandwiched-up-bnd}Given states $\rho$ and $\sigma$,
the following bound holds for all $\alpha>1$ and $\varepsilon\in\left(
0,1\right)  $:%
\begin{multline}
D_{\max}^{\varepsilon}(\rho\Vert\sigma)\leq\widetilde{D}_{\alpha}(\rho
\Vert\sigma)\label{eq:app-dmaxeps-up-sRE}\\
+\frac{1}{\alpha-1}\log_{2}\!\left(  \frac{1}{\varepsilon^{2}}\right)
+\log_{2}\!\left(  \frac{1}{1-\varepsilon^{2}}\right)  .
\end{multline}

\end{proposition}

\begin{proof}
The first steps are exactly the same as
\eqref{eq:app-SDP-dual-dmax}--\eqref{eq:app-ratio-probs-ineq}. Now consider
from the data processing inequality under the channel%
\begin{equation}
\Delta(\omega):=\operatorname{Tr}[\Pi\omega]|0\rangle\langle
0|+\operatorname{Tr}[\hat{\Pi}\omega]|1\rangle\langle1|
\end{equation}
that%
\begin{align}
&  \widetilde{D}_{\alpha}(\rho\Vert\sigma)\nonumber\\
&  \geq\widetilde{D}_{\alpha}(\Delta(\rho)\Vert\Delta(\sigma))\\
&  =\frac{1}{\alpha-1}\log_{2}\left(
\begin{array}
[c]{c}%
\left(  \operatorname{Tr}[\Pi\rho]\right)  ^{\alpha}\left(  \operatorname{Tr}%
[\Pi\sigma]\right)  ^{1-\alpha}\\
+\left(  \operatorname{Tr}[\hat{\Pi}\rho]\right)  ^{\alpha}\left(
\operatorname{Tr}[\hat{\Pi}\sigma]\right)  ^{1-\alpha}%
\end{array}
\right) \\
&  \geq\frac{1}{\alpha-1}\log_{2}\left(  \left(  \operatorname{Tr}[\Pi
\rho]\right)  ^{\alpha}\left(  \operatorname{Tr}[\Pi\sigma]\right)
^{1-\alpha}\right) \\
&  =\frac{1}{\alpha-1}\log_{2}\left(  \operatorname{Tr}[\Pi\rho]\left(
\frac{\operatorname{Tr}[\Pi\rho]}{\operatorname{Tr}[\Pi\sigma]}\right)
^{\alpha-1}\right) \\
&  =\frac{1}{\alpha-1}\log_{2}\left(  \operatorname{Tr}[\Pi\rho]\right)
+\log_{2}\left(  \frac{\operatorname{Tr}[\Pi\rho]}{\operatorname{Tr}[\Pi
\sigma]}\right) \\
&  \geq\frac{1}{\alpha-1}\log_{2}\left(  \operatorname{Tr}[\Pi\rho]\right)
+\lambda.
\end{align}
Now picking%
\begin{equation}
\lambda=\widetilde{D}_{\alpha}(\rho\Vert\sigma)+\frac{1}{\alpha-1}\log
_{2}\left(  \frac{1}{\varepsilon^{2}}\right)  ,
\end{equation}
we conclude that%
\begin{equation}
\operatorname{Tr}[\Pi\rho]\leq\varepsilon^{2}.
\end{equation}
The rest of the proof then proceeds as in
\eqref{eq:app-final-steps-1}--\eqref{eq:app-final-steps-last}, and we find
that $D_{\max}^{\varepsilon}(\rho\Vert\sigma)\leq\lambda+\log_{2}\left(
1/\left[  1-\varepsilon^{2}\right]  \right)  $.
\end{proof}

\section{Resource theory of asymmetric distinguishability based on infidelity}

\label{app:infidelity}In this paper, we employed the normalized trace distance
throughout as the measure for approximation in approximate box
transformations. As emphasized in the main text, the primary reason for doing
so is due to the interpretation of normalized trace distance as the error in a
single-shot experiment, as discussed around~\eqref{eq:trace-distance-error}.
Another advantage is that the optimizations corresponding to the one-shot
operational tasks of distillation and dilution are characterized by
semi-definite programs in both the theory presented here and in \cite{WW19}.

One could alternatively employ the infidelity $1-F(\rho,\widetilde{\rho})$ as
the measure for approximation. The main disadvantage in doing so is that the
interpretation in terms of error is not as strong as it is for normalized
trace distance. Furthermore, in the resource theory of asymmetric
distinguishability for quantum channels, it is not clear whether the
optimizations for the operational tasks of distillation and dilution are
characterized by semi-definite programs~\cite{WW19}.

However, there are some advantages to using the infidelity, which we highlight
briefly here while avoiding detailed proofs. For the rest of this appendix, we
employ the following shorthand:%
\begin{equation}
\rho\approx_{\varepsilon_{F}}\widetilde{\rho}\qquad\Longleftrightarrow
\qquad1-F(\widetilde{\rho},\rho)\leq\varepsilon_{F},
\end{equation}
using the notation $\varepsilon_{F}$ to emphasize that the error is with
respect to infidelity.

First, it is worthwhile to note that the one-shot distillable
distinguishability is unchanged:%
\begin{equation}
D_{d}^{\varepsilon_{F}}(\rho,\sigma)=D_{\min}^{\varepsilon_{F}}(\rho
\Vert\sigma),
\end{equation}
while the one-shot distinguishability cost becomes%
\begin{equation}
D_{c}^{\varepsilon_{F}}(\rho,\sigma)=D_{\max}^{\varepsilon_{F}}(\rho
\Vert\sigma),
\end{equation}
where%
\begin{equation}
D_{\max}^{\varepsilon_{F}}(\rho\Vert\sigma):=\inf_{\widetilde{\rho
}\,:\,1-F(\widetilde{\rho},\rho)\leq\varepsilon}D_{\max}(\widetilde{\rho}%
\Vert\sigma).
\end{equation}
In the above, the superscript $\varepsilon_{F}$ serves to distinguish
$D_{\max}^{\varepsilon_{F}}(\rho\Vert\sigma)$ from the smooth max-relative
entropy in \eqref{eq:smooth-max-rel-ent}. Then we have the following
expansions:%
\begin{multline}
D_{d}^{\varepsilon_{F}}((\rho^{\otimes n},\sigma^{\otimes n}))=\\
nD(\rho\Vert\sigma)+\sqrt{nV(\rho\Vert\sigma)}\Phi^{-1}(\varepsilon
_{F})+O(\log n).
\end{multline}%
\begin{multline}
D_{c}^{\varepsilon_{F}}((\rho^{\otimes n},\sigma^{\otimes n}%
))=\label{eq:app-2nd-order-dcost-fid}\\
nD(\rho\Vert\sigma)-\sqrt{nV(\rho\Vert\sigma)}\Phi^{-1}(\varepsilon
_{F})+O(\log n),
\end{multline}
with the key difference being that the second-order characterization of the
$\varepsilon_{F}$-approximate distinguishability cost is now tight. The
inequality in \eqref{eq:dh-dmax-relation-operational} becomes as follows:%
\begin{multline}
D_{\min}^{\varepsilon_{F}}(\rho\Vert\sigma)\leq D_{\max}^{\varepsilon
_{F}^{\prime}}(\rho\Vert\sigma)\\
-\log_{2}\!\left(  1-\left[  \sqrt{\varepsilon_{F}}+\sqrt{\varepsilon
_{F}^{\prime}}\right]  ^{2}\right)  ,
\end{multline}
for $\varepsilon_{F},\varepsilon_{F}^{\prime}\geq0$ and $\sqrt{\varepsilon
_{F}}+\sqrt{\varepsilon_{F}^{\prime}}<1$, which follows by employing the
triangle inequality for the sine distance \cite{R02,R03,GLN04,R06}. The
inequality in \eqref{eq:app-dmax-dheps-relation-specs}\ becomes as follows:%
\begin{multline}
D_{\max}^{\varepsilon_{F}}(\rho\Vert\sigma)\leq D_{\min}^{1-\varepsilon_{F}%
}(\rho\Vert\sigma)\\
+\log_{2}\left\vert \text{spec}(\sigma)\right\vert +\log_{2}\!\left(  \frac
{1}{1-\varepsilon_{F}}\right)  ,
\end{multline}
which is a key reason why we obtain \eqref{eq:app-2nd-order-dcost-fid}. The
inequality in \eqref{eq:app-lower-dmaxeps-da} becomes
\begin{equation}
D_{\max}^{\varepsilon_{F}}(\rho\Vert\sigma)\geq\widetilde{D}_{\alpha}%
(\rho\Vert\sigma)+\frac{\alpha}{\alpha-1}\log_{2}\!\left(  \frac
{1}{1-\varepsilon_{F}}\right)  ,
\end{equation}
for $\alpha\in\lbrack1/2,1)$, while \eqref{eq:app-dmaxeps-up-sRE}\ becomes%
\begin{multline}
D_{\max}^{\varepsilon_{F}}(\rho\Vert\sigma)\leq\widetilde{D}_{\alpha}%
(\rho\Vert\sigma)+\log_{2}\!\left(  \frac{1}{1-\varepsilon_{F}}\right)  \\
+\frac{1}{\alpha-1}\log_{2}\!\left(  \frac{1}{\varepsilon_{F}}\right)  ,
\end{multline}
for $\alpha\in(1,\infty)$. The converse bound in
Proposition~\ref{prop:strong-converse-exp} becomes%
\begin{align}
&  -\frac{1}{n}\log_{2}(1-\varepsilon_{F})\nonumber\\
&  \geq\left(  \frac{1-\alpha}{\alpha}\right)  \left(  R\ \widetilde
{D}_{\alpha}(\tau\Vert\omega)-\widetilde{D}_{\beta}(\rho\Vert\sigma)\right)
\label{eq:app-str-conv-exp}\\
&  =\left(  \frac{\beta-1}{\beta}\right)  \left(  R\ \widetilde{D}_{\alpha
}(\tau\Vert\omega)-\widetilde{D}_{\beta}(\rho\Vert\sigma)\right)  ,
\end{align}
holding for an arbitrary $(n,m,\varepsilon_{F})$\ box transformation protocol
(i.e., so that $\mathcal{N}^{(n)}(\rho^{\otimes n})\approx_{\varepsilon_{F}%
}\tau^{\otimes m}$ and $\mathcal{N}^{(n)}(\sigma^{\otimes n})=\omega^{\otimes
m}$), $\alpha\in(1/2,1)$, and $\beta=\alpha/\left(  2\alpha-1\right)  $. For
distinguishability distillation with $\tau=|0\rangle\langle0|$ and $\omega
=\pi$, so that $\widetilde{D}_{\alpha}(\tau\Vert\omega)=1$, this bound reduces
to%
\begin{equation}
-\frac{1}{n}\log_{2}(1-\varepsilon_{F})\geq\left(  \frac{\beta-1}{\beta
}\right)  \left(  R-\widetilde{D}_{\beta}(\rho\Vert\sigma)\right)  ,
\end{equation}
holding for all $\beta>1$, which is the optimal strong converse exponent, as
shown in \cite{MO13}. For distinguishability dilution with $\rho
=|0\rangle\langle0|$ and $\sigma=\pi$, so that $\widetilde{D}_{\beta}%
(\rho\Vert\sigma)=1$, and by multiplying \eqref{eq:app-str-conv-exp}\ by $n/m$
and setting the rate $S=n/m$, the bound becomes%
\begin{equation}
-\frac{1}{m}\log_{2}(1-\varepsilon_{F})\geq\left(  \frac{1-\alpha}{\alpha
}\right)  \left(  \widetilde{D}_{\alpha}(\tau\Vert\omega)-S\right)  ,
\end{equation}
holding for all $\alpha\in\lbrack1/2,1)$. It is an interesting open question
to determine the optimal strong converse exponent for distinguishability dilution.

\end{document}